\documentclass[12pt]{article}

\usepackage{fullpage}
\usepackage[T1]{fontenc}
\usepackage{ae,aecompl}
\usepackage[dvips]{graphicx}
\usepackage{amssymb}
\usepackage{amsmath}
\usepackage[round,authoryear]{natbib}
\usepackage{color}
\usepackage{array}
 
\definecolor{webgreen}{rgb}{0,0.4,0}
\definecolor{webbrown}{rgb}{0.6,0,0}
\definecolor{purple}{rgb}{0.5,0,0.25}
\definecolor{darkblue}{rgb}{0,0,0.7}
\definecolor{darkred}{rgb}{0.7,0,0}
\definecolor{darkgreen}{rgb}{0,0.7,0}
\usepackage[pdfborder=false]{hyperref}
\hypersetup{colorlinks,citecolor=darkred,filecolor=black,linkcolor=darkblue,urlcolor=webgreen,pdfpagemode=none,
pdfstartview=FitH}
\newcommand{\ignore}[1]{}
\newtheorem{lemma}{{\sc Lemma}}[section]
\newtheorem{prop}{{\sc Proposition}}[section]
\newtheorem{cor}{{\sc Corollary}}[section]
\newtheorem{theorem}{{\sc Theorem}}[section]
\newtheorem{defn}{{\sc Definition}}[section]

\newtheorem{claim}{{\sc Claim}}
\newtheorem{example}{{\sc Example}}[section]

\newtheorem{remark}{{\sc Remark}}

\newenvironment{proof}{\noindent {\bf \sl Proof\/}:\enspace}
{\hfill $\blacksquare{}$ \vspace{12pt}}
\usepackage{sectsty}
\sectionfont{\centering\normalfont\scshape}
\subsectionfont{\centering\normalfont}

\begin{document}

\title{\bf Strategy-proof Selling: a Geometric Approach}
\author{\bf Mridu Prabal Goswami\thanks{Indian Statistical Institute, Tezpur, India. 
		 I am extremely grateful to Arunava Sen and Tridib Sharma for encouraging me to take up this project. I am grateful to Abinash Panda for going through the draft and making many helpful comments. I am extremely thankful to Dipjyoti Majumdar for helpful comments. Another version of this paper was submitted to a journal and was subsequently rejected.    
		I am thankful to the two anonymous referees, I have tried to address their concerns in this draft.    
		I am extremely thankful to Magesh Kumar K. K. and Manish Yadav.  
		I am thankful to a research grant from SERB, DST, Government of India.}}
\maketitle

\begin{abstract}
\noindent We consider one buyer and one seller. For a bundle 
$(t,q)\in [0,\infty[\times [0,1]=\mathbb{Z}$, $q$ either refers to the wining probability of an object or a share of a good, and $t$ denotes the payment that the buyer makes. We define classical and restricted classical preferences of the buyer on $\mathbb{Z}$; they incorporate quasilinear, non-quasilinear, risk averse preferences with multidimensional pay-off relevant parameters. We define rich single-crossing subsets of the two classes, and characterize strategy-proof mechanisms by using monotonicity of the mechanisms and continuity of the indirect preference correspondences. We also provide a computationally tractable optimization program to compute the optimal mechanism.  We do not use revenue equivalence and virtual valuations as tools in our proofs. Our proof techniques bring out the geometric interaction between the single-crossing property and the positions of bundles $(t,q)$s. Our proofs are simple and provide computationally tractable optimization program to compute the  optimal mechanism. The extension of the optimization program to the $n-$ buyer environment is immediate.

\end{abstract}

\noindent {\bf Key words}: single-crossing, non-quasilinear preference, classical preferences, risk averse, strategy-proof, indirect preference correspondence, optimal mechanism, topology, computation, artificial intelligence

\section{\bf Introduction}

Selling rights by a firm to a vendor so that the vendor can provide specific services to the
 customers
can be abstractly thought of as selling a fraction of a good where the good refers to the full rights.  
Consider a firm that produces soft drinks. The management of the firm may grant selling rights of its products to a party in such a way that only certain  specified brands of soft drinks can be sold by the licensee. In an abstract sense, if we think the set of all brands of soft drinks produced by the firm to be 1, then selling rights to sell certain brands can be considered as selling fractions of a good.  In return, the licensee needs to make a payment to the firm. This payment may be made in the form of a license fee, or in the form a revenue sharing arrangement with the firm. The developer who received the Delhi airport on lease during its privatization in $2003$, and thus received  the rights to develop and manage the airport, in return was required to share revenue with the government. 
We call the firm that sells rights to be a seller, and the agent that buys the rights to be a buyer.      
We denote rights by $q\in [0,1]$  and payment by $t$ with $t\geq 0$, and consider preferences of the buyer of the rights over the tuples $(t,q)$, we call $(t,q)$ a bundle. We assume these preferences to be classical, i.e., monotone in both $t$ and $q$, and continuous, see Definition \ref{defn:classical}. We then relax the monotonicity assumption of classical preferences, an extension of our analysis to non-monotone preferences in Section \ref{sec:non_mon}. We call the extended preferences to be restricted classical.

We study mechanisms, i.e., functions, that map preferences of the buyer to a bundle $(t,q)$. We assume that the buyer knows, and only the buyer knows her preference, i.e., preference is private information of the buyer. 
Hence, we study mechanisms that have the property that the allocation, i.e., a bundle $(t',q')$, obtained by the buyer from the mechanism by misreporting her preference is not better under her true preference compared to the allocation $(t,q)$ obtained from the report of her true preference; in short                
`the buyer has no incentive to misreport her preference'. A mechanism that satisfies this property is called strategy-proof, see Definition \ref{defn:sp}.

In our model $q$ can also represent probability of win. In that case the expected payment that the winner makes is $qt$. However, the winner's net utility function may not be of the form ``expected value $-$ expected payment''. In general it is possible that the form in which     
payment appears in the buyer's utility function is different from the way the seller evaluates her revenue. For example, the buyer's utility may be of the form $\theta q-t$ but the seller considers expected payment, i.e., $tq$, as the measure of revenue.  
Therefore, we consider mechanisms that map  
preferences to bundles $(t,q)$, instead of the ones that map preferences to $(\text{expected payment}, q )$. A classic model that considers the latter preferences is \citep{myer}. 
In our research we allow preferences to take more general forms. This provides avenues to study optimal mechanisms with a more general classes of preferences.

We consider the `domain' of mechanisms, domain refers to the set of preferences from which the preference of the buyer is drawn, that satisfies the single-crossing property. A domain satisfies the single-crossing property if indifference sets of two distinct preferences that belong to the domain intersect at most at one bundle. We provide examples of both quasilinear and non-quasilinear preference domains that satisfy the single-crossing property.  
Further, we provide examples of single-crossing domains such that preferences from the domain admit multidimensional parametric representations. The single-crossing property, i.e., indifference sets of two distinct preferences can intersect at most one bundle, may be considered to be an algorithm  that creates restricted domains from universal set of preferences.
We make a detailed remark about it in the text.

The single-crossing property entails an order on the set of preferences. The order in the domain ensures that we can define monotonicity of the two component functions, i.e., two functions $t$ and $q$, of the mechanisms.{\footnote{Instead of using new notations to denote the two component functions of a mechanism that maps preferences to payments and shares/winning probabilities we denote them by $t$ and $q$.}} We consider rich single-crossing domains. The richness requires for every $z'=(t',q'),z''= (t'',q'')$ with $t'<t'',q'<q''$ there is a preference from the single-crossing domain such that
$z',z''$ are indifferent. By the single-crossing property such preference is unique. We define a natural order topology on rich single-crossing domains. Using the order we define monotonicity of mechanisms, and using the order topology we define continuity of the indirect preference correspondences, see Definition \ref{defn:cont}. We utilize the monotonicity of the mechanisms, and the continuity of the indirect preference correspondences to study the geometry of the ranges of the mechanisms that are strategy-proof. The geometry  provides a novel insight into the interactions between strategy-proofness and domains of mechanisms. The continuity of an indirect preference correspondence reduces to continuity of an indirect utility function if we consider utility representation of the preferences.  
For instance, for the domain $\{\theta q-t\mid \theta>0\}$  it entails continuity of the indirect utility function, where the latter is defined for $\theta\in ]0,\infty[$.   
However, we do not use the revenue equivalence equation in our characterization of strategy-proof mechanisms. We show that monotonicity and continuity of the indirect preference correspondences are necessary for strategy-proofness of mechanisms. If the range of a mechanism is finite, then they are also sufficient.   
We provide an example to show that the continuity of the indirect preference correspondence and monotonicity of mechanisms are not enough to guarantee strategy-proofness of mechanisms with continuum ranges in Section \ref{sec:count}. 
In contrast, if the number of limit points is finite, and the range is countable and closed, then continuity of the  
indirect preference correspondences and monotonicity of the mechanisms are sufficient conditions for mechanisms to be strategy-proof.
In other words the axioms applied to the mechanisms may not alone reveal their strengths when it comes to characterizations of strategy-proof mechanisms, the range of the mechanisms may also be an important element that needs to be considered. 
This analysis related to the robustness of our model is included in Section \ref{sec:count}.

In Theorem \ref{thm:implies_strtagey_proof_finite_range}  we show that if the range of a mechanism is finite, monotone and the indirect preference correspondence is continuous then the mechanism is strategy-proof. The proof of Theorem \ref{thm:implies_strtagey_proof_finite_range} is constructive. 
Given the range, i.e., given the positions of the bundles in the space $\mathbb{Z}$, the proof of the theorem provides the exact rule of the mechanism. By using the single-crossing property we identify this rule and the geometry of the rule manifests itself in terms of the preferences that make subsets of the range indifferent. We discuss this point in more details in the text. Since our proofs depend only on the ordinal feature of the single-crossing property we cannot use revenue equivalence, and in general envelope theorems, in our characterization. In fact, we do not use any. We discuss this point further in the text. Further, from the proof of Theorem \ref{thm:implies_strtagey_proof_finite_range} it is immediate that our axioms do not fix the payment rule, it fixes the rule of the mechanism. Thus, our approach is different from the standard mechanism design approach where given an allocation rule the axioms identify a payment rule that implements the allocation rule. In contrast we answer the following two questions in affirmative: `Given a finite subset of $\mathbb{Z}$ can we find reasonable conditions on the subset and the mechanisms such that the given subset can supported as the range of a mechanism? Can we identify the rule of that mechanism? An immediate advantage of this approach is that we can read- off the expected revenue of the seller easily from the geometry of the mechanism. In fact, this approach helps to setup and show the existence of the optimal mechanism quite easily with a minimal assumption 
on the probability measure on the domain of preferences without using utility representations. In particular, we do not require virtual valuation in our proofs. Since we do not use revenue equivalence in our characterization we do not have access to the technique of substituting the revenue equivalence equation into the expected revenue of the seller to solve for the optimal mechanism. In fact, even if revenue equivalence holds in some domains it does not mean that 
solving for the optimal mechanism is easy, especially when the buyer is not not risk neutral.   
We discuss this in details in Remark \ref{remark:myer_vs_single}. Our proof of the existence of an optimal mechanism entails an optimization program that is computationally tractable. This makes our approach distinct from the approaches that depend on revenue equivalence.    
We prove the existence of optimal mechanism in Theorem \ref{thm:optimal}.

Classical preferences satisfy monotonicity everywhere in the consumption space, 
but risk averse preferences such as $q\sqrt{\theta-t}, 0<\theta, t\leq \theta$  does not. Thus we extend our analysis to non-monotone preferences, and call the general class of extended preferences to be
restricted classical. The preference $q\theta-qt,  0<\theta$ is a classic example of restricted classical preferences which allows for lottery over payment that does not depend on $\theta$ unlike the risk averse model. We modify the definition of the single-crossing property for these preferences and show that characterizations of mechanisms that are analogous to the classical preferences hold for restricted classical preferences as well. All discussions about restricted classical preferences can be found in Section \ref{sec:non_mon}.      
In the next section we review some of the related papers to bring out the difference between these papers and the model in our paper.

The notion of single-crossing preferences that we use may appear similar to the one used in the literature on mechanism design. 
However, the notion in our paper and the one used in the literature are in effect very different. In the literature, for example \citep{Baisa2}, \citep{Saporiti}, \citep{Tian} this property is defined on monotone allocations, and often is defined by using a parametric class utility functions as a primitive. 
In our definition we do not use any parametric class of utility functions. Our definition stems from  the intersection of indifference sets of two `ordinal' preferences, and classical(resp. restricted classical) properties of preferences entail an order on the set of preferences which in fact makes and set of classical single-crossing properties a one dimensional manifold. With the help of various examples we demonstrate that  the dimension of the manifold is not the same thing as the dimension of parameters or types. We give examples to show that one dimensional manifold allows for multidimensional parametric utility representations. Thus we do not work with parameters or types. We use the geometry of the ordinal property of single-crossing, i.e., indifference sets of two distinct preferences  can intersect at most once to study optimal strategy-proof mechanisms, this approach is different from the other papers in the literature.

\subsection{\bf Related Literature}
\label{sec:lit}  
\citep{Baisa2} considers a specific one parameter set of preferences that satisfy the single-crossing property. We allow for multidimensional parametric representations of the preferences. \citep{Saporiti} considers a single-crossing domain to study  voting rules. Instead of assuming an order on the preferences which is the case in \citep{Baisa2}, in \citep{Baisa2} the order is due to the natural order on the parameters that represent the preferences, and \citep{Saporiti} we derive an order  on the set of preferences by using the single-crossing property of the classical preferences. 
\citep{Laffont1} consider a one dimensional parametric class of single-crossing preferences and study 
the implementability of piecewise continuously differentiable allocation rules. 
\citep{Tian} study equivalence between a notion of implementability and notions of cycle monotonicity. 
\citep{Gershkov} use \citep{Saporiti} to construct incentive compatible and ex-ante welfare maximizing. 
Furthermore  in the context of voting \citep{Barbera2} study a model where society’s preferences over
voting rules satisfy the single-crossing property with an objective to analyze self-stable rather than strategy-proof voting rules. \citep{Gans} study
an Arrovian aggregation problem with single-crossing preferences for voters, and 
show that median voters are decisive in all majority elections between pairs of alter
natives. \citep{Barbera3} develop a notion of `top monotonicity' which is a common generalization of single-peakedness and single-crossingness. 
\citep{Corchon} study a public-good-private-good production economy where
agents' preferences satisfy the single-crossing property and prove that smooth strategy-proof and Pareto-efficient social choice functions that give strictly positive amount
of both goods to all agents do not exist.

Since the single-crossing property allows for various kinds of non linearity in the preferences, the literature on non-quasilinear preferences is also relevant.  Some of the recent studies on non-quasilinear preferences are  
\citep{Mishra1} and  \citep{Mishra2} \citep{Serizawa}, \citep{Baisa1}. 
An important domain studied in \citep{Mishra1} admits 
positive income effect.
Later we give an example of a single-crossing domain that satisfies positive income effect. The main idea of the single-crossing domain in this paper is from \citep{Goswami}. In \citep{Goswami} the single-crossing property is defined for classical exchange economies.     
The domain  in \citep{myer} satisfies the single-crossing property for positive pay-off levels, for details see the section on restricted classical preferences. 
The paper is organized as follows. In Section \ref{sec:pre} we introduce important definitions  
In Section \ref{sec:sp} we study strategy-proof mechanisms for classical single-crossing preferences.
In Section \ref{sec:opt} we study optimal mechanisms for classical single-crossing preferences.
In Section \ref{sec:count} we study robustness of our axioms. 
In Section \ref{sec:non_mon} we extend our analysis to restricted classical preferences.         
In Section \ref{sec:con} we make some concluding remarks.

\section{\bf Preliminaries}
\label{sec:pre}
The economic environment in this paper consists of a seller and a buyer. 
The seller wants to sell a fraction of an object.  This fraction is denoted by $q$ and
$q\in [0,1]$. 
If the seller sells an indivisible unit of a good, then $q$ denotes the probability that the object is sold to the buyer.    
In return, the buyer needs to make a payment to the seller. This  payment is denoted by $t$. The set of allocations is denoted by 
$\mathbb{Z}$, and $\mathbb{Z}=[0,\infty[\times [0,1]$, where $[0,\infty[=\Re_{+}$ denotes the set of non-negative real numbers.
A typical bundle denoted by 
$(t,q)$, where $t\in \Re_{+}$ and $q\in [0,1]$.  The buyer's preference over $\mathbb{Z}$ is denoted by $R$. 
The strict counterpart of $R$ is denoted by $P$, and indifference is denoted by $I$.
In the same spirit as in \citep{Mishra1} we consider classical preferences. 
To formally write the definition of a classical preference we introduce  the following notations. For $z\in \mathbb{Z}$, and $R$, let $UC(R,z)=\{z'\in \mathbb{Z}| z' Rz\}$. In words $UC(R,z)$ is the set of bundles that are weakly preferred to $z$ under $R$.
Likewise $LC(R,z)=\{z'|zRz'\}$, $LC(R,z)$ is the set of bundles that are weakly less preferred to $z$ under $R$. The notion of a classical preference is defined below.

\begin{defn}[Classical Preference] \rm 
 	The complete, transitive preference relation $R$ on $\mathbb{Z}$ is \textbf{classical} if $R$ is  
 	{\bf monotone}, i.e., 
 		
 	\begin{itemize}
 		
 	\item \textbf{money-monotone:} for all $q\in [0,1]$, if $t''>t'$, then $(t',q)P(t'',q)$.
 		\item \textbf{$q$-monotone:} for all $t\in \Re_+$, if $q''>q'$, then $(t,q'')P(t,q')$.

 	\end{itemize}
 
 \noindent and  \textbf{continuous} for each $z\in \mathbb{Z}$, the sets $UC(R,z)$ and $LC(R,z)$ are closed sets.{\footnote{These two sets are closed in the product topology on the Euclidean space $\mathbb{Z}$.    }}
 
\label{defn:classical} 
 \end{defn}
 
\noindent 
Let $R$ be a classical preference and $x$ a bundle, define  $IC(R,x)=\{y\in \mathbb{Z}\mid yIx\}$.{\footnote{An $IC(R,x)$ set also represents an equivalence class of the equivalence relation $I$.}} The set $IC(R,x)$ is the set of bundles that are indifferent to $x$ according to the preference $R$. 
It can be seen easily that due to the properties of a classical preference, an $IC$ set can be represented as a curve in $\mathbb{Z}$. Thus we may also call an $IC$ set an $IC$ curve. 
We shall represent $\Re_{+}$ on the horizontal axis, and $[0,1]$ on the vertical axis. An $IC$ curve in $\mathbb{Z}$ is an upward slopping curve, i.e., if $(t',q'),(t'',q'')\in IC(R,x)$ and $t'<t''$, then $q'<q''$. Let $x'=(t',q'), x''=(t'',q'')$. By $x'\leq x''$ we mean either $x'=x''$ or $t'<t'',q'<q''$. Further, by $x'<x''$ we mean $t'<t'', q'<q''$.     
We call two bundles $x'=(t',q'), x''=(t'',q'')$ {\bf diagonal} if $x'< x''$. 
The single-crossing property is defined next.

\begin{defn}[Single-Crossing of two Preferences]\rm
We say that two distinct classical preferences $R', R''$ exhibit the {\bf single-crossing property} if and only if  for all $x,y,z \in \mathbb{Z}$, 
$$\text{if}~ z\in IC(R',x)\cap IC(R'',y),~\text{then}~ IC(R',x)\cap IC(R'',y)=\{z\}.$$
\label{defn:single_crossing}
\end{defn}
\noindent The single-crossing property implies that two $IC$ curves of two distinct preferences can meet ( or cut) at most one bundle. The single-crossing property is an ordinal property of preferences, i.e., this property does not depend on utility representations of preferences.
Next we define the notion of a single-crossing domain.{\footnote{The monotonicity of preferences and the single-crossing property in \citep{Goswami} are defined for the interior of the consumption space. }}            

\begin{defn}[Rich Single-crossing domain]\rm We call a subset of the set of classical preferences {\bf single-crossing domain} if any $R',R''$ that belongs to the subset 
satisfy the single-crossing property. We call a single crossing domain {\bf rich} if for any two bundles $x'=(t',q'), x''=(t'',q'')$ such that $t'<t'',q'<q''$ there is $R$ in the single crossing domain such that 
$x'Ix''$. We denote a rich single crossing domain by $\mathcal{R}^{rsc}$.        	
\end{defn}

\noindent We may interpret the single-crossing property as an algorithm that produces  restricted domains. Consider the set of all classical preferences. Consider two classical preferences $R'$ and $R''$ that admit the single-crossing property. Suppose we wish to add another preference $R'''$ to the list that already contains $R'$ and $R''$. 
The single-crossing property ensures that 
an additional preference can be added to the list only in a specific manner. 
A rich single-crossing domain is a maximal single-crossing domain, i.e., $\mathcal{R}^{rsc}\cup \{R\}$, where $R\not\in \mathcal{R}^{rsc}$ and $R$ is classical, is not a single-crossing domain. For all diagonal bundles $x', x''$ there is a preference $R\in \mathcal{R}^{rsc}$ such that $x'\in IC(R,x'')$. Thus if another preference is added to $\mathcal{R}^{rsc}$, then it violates the single-crossing property. 
Given the initial preferences  $R'$ and $ R''$, a rich single-crossing domain may be interpreted as the limit point of the algorithm. If we add more preferences to a rich single-crossing domain, then we may have a situation pertaining to Maskin Monotonic Transformations which is defined next.
In general, by adding a preference to $\mathcal{R}^{rsc}$ we allow indifference curves of two preferences 
to be tangential.  
 
\begin{defn}\rm ({\bf Maskin Monotonic Transformation}) Let $R',R''$ be two classical preferences. We say 
that $R''$ is a {\bf Maskin Monotonic Transformation, (in short MMT)} of $R'$ through $z$, if 
$(i)$ $UC(R'',z)\subseteq UC(R',z)$, $(ii)$ if $x\neq z, x\in UC(R'',z)$, then $xP'z$.  	
\end{defn}

\noindent MMT implies that the indifference curve of $R''$ through $z$ is tangential to the indifference curve of $R'$ through $z$. Single-crossing domains do not allow MMTs in the interior of $\mathbb{Z}$. The set of all classical preferences satisfy MMTs for all preferences at every bundle. In the context of two goods and two agents exchange economies \citep{Barbera1} find that a strategy-proof and individually rational mechanism must have a range whose elements fall on at most two line segments. \citep{Goswami} provides an example to demonstrate this not to be the case if the domain is rich single-crossing. Thus, the geometry of the range of strategy-proof mechanisms change if the domains of mechanisms are larger than rich single-crossing domains. We explore some important properties of rich single-crossing domains in the next subsection.

\subsection{\bf Topology on $\mathcal{R}^{rsc}$}
The single-crossing property provides a natural way to define an order on $\mathcal{R}^{rsc}$. This order defines an order topology on $\mathcal{R}^{rsc}$.
This order topology is metrizable, more specifically with this topology $\mathcal{R}^{rsc}$ is homeomorphic to the real line with the standard Euclidean metric. 
We proceed to formally state this topology.      
Let $\square(z)=\{x\mid x\leq z\}$.

\begin{defn}\rm 
Let $\mathcal{R}^{rsc}$ be a rich single crossing domain. Consider $z\in \mathbb{Z}$ and $R', R''\in \mathcal{R}^{rsc}$. We say, $R''$ \textbf{cuts} $R'$ \textbf{from above} at $z\in \mathbb{Z}$, if and only if  
\[ \square(z)\cap UC(R'',z) \subseteq \square(z)\cap UC(R',z). \]

\label{defn:cut} 
\end{defn}
In words, the indifference curve of $R''$ through $z$ lies above the indifference curve for $R'$ through $z$ in $\square(z)$ if both indifference curves are viewed from the horizontal axis $\Re_{+}$. By applying arguments similar to ones in  
\citep{Goswami} it can be  established that if $R''$ cuts $R'$ from above at some $z\in \mathbb{Z}$, then $R''$ cuts $R'$ from above at every $z\in \mathbb{Z}$.{\footnote{In \citep{Goswami} preferences are assumed to have strictly convex upper contour sets. However Proposition $3$ in \citep{Goswami} which proves if $R''$ cuts $R'$ from above at some bundle, then $R''$ cuts $R'$ from above at every bundle does not use convexity of preferences. }}
We say that $R''$ cuts $R'$ from above if $R''$ cuts $R'$ from above at every bundle. Thus we define the following order on $\mathcal{R}^{rsc}$: 
for all $R',R''\in \mathcal{R}^{rsc}$ we say $R'\prec R''$ if $R''$ cuts $R'$ from above. We consider the order topology on $\mathcal{R}^{rsc}$ by applying this order.   

We recall a few definitions. An order $\prec$ in any set $L$ is called {\bf simple or linear} if  
it has the following properties (see \citep{Munkres}):
$(1)$ (Comparability) For every $\alpha$ and $\beta$ in $L$ for which $\alpha\neq \beta$, either $\alpha \prec \beta$ or $\beta\prec \alpha$,
$(2)$ (Non-reflexivity) For no $\alpha$ in $L$ does the relation $\alpha \prec \alpha $ hold,
(3) (Transitivity) If  $\alpha \prec \beta $ and $\beta \prec \gamma$, then $\alpha\prec \gamma$. An important class of linearly ordered sets are linear continuums. We define this notion next.    

\begin{defn}[Linear continuum]\rm (see \citep{Munkres}) A simply ordered set $L$ having more than one element is called a linear continuum if the following hold:
$(1)$ $L$ has the least upper bound property, i.e., every bounded subset, i.e.,  bounded according to the order $\prec$ on $L$, has the least upper bound in $L$, (see \citep{Rudin}), $(2)$ if $\alpha \prec \beta$, then there exists $\gamma$ such that $\alpha\prec \gamma \prec \beta$, $\alpha, \beta,\gamma $ are in $L$.   
\end{defn}

\noindent In Theorem \ref{thm:lin_cont} we show that $\mathcal{R}^{rsc}$ is a linear continuum. 
To define the order topology for any two preferences $R', R''\in \mathcal{R}^{rsc}$ with $R'\prec R''$ define open interval $]R', R''[=\{R|R'\prec R \prec R''\}$. The collection of open intervals form a basis on $\mathcal{R}^{rsc}$. The topology generated by this basis is the order topology on $\mathcal{R}^{rsc}$.
In this topology $\mathcal{R}^{rsc}$ is connected and closed interval $[R',R'']=\{R\mid R'\precsim R \precsim R''\}$ is compact, here $R'\precsim R$ means either $R'\prec R$ or $R'=R$, $]-\infty, R[$ and $]-\infty, R]$ denote open and closed intervals respectively that are not bounded below, $]R,\infty[$ and $[R,\infty]$ denote open and closed intervals respectively that not bounded above,
see \citep{Munkres} for details. In this order topology $\mathcal{R}^{rsc}$ is homeomorphic to $\Re$ and thus metrizable. Further, the homeomorphism is an order preserving bijection with the real line. The next lemma formalizes these observations. Before proceeding we note that intervals in all ordered sets  in this paper are written by using the braces $]$ and $[$. For example, if $\alpha,\beta \in \Re $ and $\alpha<\beta$, then $]\alpha,\beta]$ denotes the left open and right closed interval.

\begin{theorem}\rm Every $\mathcal{R}^{rsc}$ is a linear continuum. Further, there is an order preserving homeomorphism, denoted by $h$, between  $\mathcal{R}^{rsc}$ with the order topology and $\Re$ with the standard Euclidean metric topology.       
 
\label{thm:lin_cont}
\end{theorem}

\begin{proof} See the Appendix.
	
\end{proof}

\noindent An implication of Theorem \ref{thm:lin_cont} is that convergence of a sequence in $\mathcal{R}^{rsc}$ is equivalent to saying that monotone sequences converge. A sequence of preferences is denoted by 
$\{R^{n}\}_{n=1}^{\infty}$. A monotone decreasing sequence means any $n$, $R^{n+1}\precsim R^{n}$, and its convergence to $R$ is denoted by $R^{n}\downarrow R$. A monotone 
increasing sequence means for any $n$, $R^{n}\precsim R^{n+1}$, and its convergence to $R$ is denoted by $R^{n}\uparrow R$. In general convergence to $R$ is denoted by $R^{n}\rightarrow R$.
Further, $R^{n}\rightarrow R \iff h(R^{n})\rightarrow h(R)$. It is known that a sequence in $\Re$ converges if and only if every monotone subsequence of the sequence converges to the same limit. Since $h$ is an order preserving homeomorphism the same is true for $\mathcal{R}^{rsc}$.   
Another very important, perhaps the most critical, implication of Theorem \ref{thm:lin_cont} is that 
$\mathcal{R}^{rsc}$ in the order topology is an one dimensional topological manifold. Later we shall give examples to show that even if $\mathcal{R}^{rsc}$ is topologically one dimensional, from the perspective  of utility representation of these preferences in terms parametric classes $\mathcal{R}^{rsc}$ is multidimensional. 
We end this section with two technical lemmas that we use later. By $cl(B)$, we mean closure of the set $B$. The next lemma shows that the infimum and the supremum of a set in $\mathcal{R}^{rsc}$ are in the closure of the set. Infimum and supremum are defined by using the order $\prec$.

\begin{lemma}\rm If $R_{0}=\inf_{\prec} U$ where $ U\subseteq \mathcal{R}^{rsc}$, then $R_{0}\in cl(U)$. 
Analogously, if $R^{0}=\sup_{\prec} U$ where $U \subseteq \mathcal{R}^{rsc}$, then $R^{0}\in cl(U)$.  	  
\label{lemma:closure}
\end{lemma}	

\begin{proof} See the Appendix.

\end{proof}

\noindent The Lemma \ref{lemma:preference_conv} proves a closure property of the order topology on $\mathcal{R}^{rsc}$: indifference in the limit of any sequence of preferences is preserved. Lemma \ref{lemma:preference_conv} is very important when proving the constraint set of the optimization program to find the optimal mechanism in Theorem \ref{thm:optimal} is compact.

\begin{lemma}\rm Let $\{R^{n}\}_{n=1}^{\infty}\subseteq \mathcal{R}^{rsc}$ be a sequence such that $R^{n}\rightarrow R$. Let $(t^{1n},q^{1n})\rightarrow (t^1,q^1),(t^{2n},q^{2n})\rightarrow (t^2,q^2)$, where $\{(t^{in},q^{in})\}_{n=1}^{\infty}\subseteq \mathbb{Z}$, $i=1,2$.
	If $(t^{1n},q^{1n})I^{n}(t^{2n},q^{2n})$, then $(t^{1},q^{1})I(t^{2},q^{2})$. 	   
\label{lemma:preference_conv}
\end{lemma}

\begin{proof} See the Appendix.

\end{proof}

\begin{remark}\rm In this remark we make some observations about the topology. 
First, the fact that $\mathcal{R}^{rsc}$ is a linear continuum is very important. 
Being a linear continuum it has other important properties in the order topology such as closed intervals are compact. Being a linear continuum it has the least upper bound property so that infimum and supremum are well defined. We note that we use infimum and supremum of subsets of preferences in our proofs.
The classical property of preferences and the single-crossing property 
together entail that $\mathcal{R}^{rsc}$ is a linear continuum and also provide the order topology. This order topology is metrizable.  It is worth noting that the single-crossing property by itself does not give the metric on the preferences in our model. It is important that when two indifference curves intersect each other at one bundle, then the way they intersect at that bundle they intersect at all other bundles in the similar way. The classical properties of preferences ensure that intersection at all bundles are similar so that the order is well defined. This fact follows from \citep{Goswami}. Once the order is well defined, the richness entails that any classical single-crossing domain is a linear continuum. 
Not only that the order topology, and hence the metric that generates the topology,
let us define continuity of correspondences, it is very crucial in the proof of Lemma \ref{lemma:preference_conv} and the latter is very important when proving that the constraint set of the optimization program to find the optimal mechanism is compact. In a nutshell, the topology let us study both strategy-proofness and 
optimal mechanisms in more general way without resorting to parametric representations of the preferences.         
The proof of the proposition in \citep{Goswami} and the proof of Theorem \ref{thm:lin_cont} are simple. We further note that $\mathcal{R}^{rsc}$ is nor a subset of any vector space, i.e., we do not define addition or multiplication operation on $\mathcal{R}^{rsc}$. Also we do not assume any smooth structure on $\mathcal{R}^{rsc}$. An important element of our arguments in our proofs is the ordinal feature of 
single-crossing. To the best of our information this approach is novel.                       
\end{remark}

\noindent In the next subsection we provide examples of rich single-crossing domains.   
	
\subsubsection{\bf Examples of $\mathcal{R}^{rsc}$} 	

\noindent The first example is of quasilinear preferences.    

\begin{example}\rm ( {\bf quasilinear preference linear in parameter})  Consider the preference given by $u(t,q;\theta)=\theta q-t, \theta>0$. The indifference curves of this preference are linear and the slope of an indifference curve is given by $\frac{1}{\theta}$. Thus indifference curves of two distinct preferences can intersect at most once. Since $\theta\in ]0,\infty[$, $\{u(t,q;\theta)=\theta q-t\mid \theta>0\}$ is a rich single-crossing domain. In this model $q$ maybe interpreted as the probability of win and $t$ as the payment that needs to be made to obtain the object with probability $q$. 
If $t$ is the payment that needs to be made irrespective of whether the buyer wins, then the expected payment is $qt+(1-q)t=t$. Thus a quasilinear preference is  a special case of preferences with expected payments. A very important model in which probability of win appears non-trivially with payment 
is $\theta q-qt$. This model is not quasilinear in $t$, also does not satisfy the monotonicity. 
We study this model as an example of restricted classical preferences in Section \ref{sec:non_mon}. We call this model Myersonian to distinguish it from the quasilinear model. We note that 
the Myersonian model allows for lottery over payments in a non-trivial way: with probability 
$q$ the buyer pays $t$ and with probability $1-q$ the buyer pays $0$. Myersonian model is a special case of the single-crossing restricted classical preferences that potentially allow for a general class  lottery over payments since these preferences do not require monotonicity to hold when $t=0$. We study these preferences in Section \ref{sec:non_mon}.

\label{ex:qlin}
\end{example}	
  
\begin{example}\rm({\bf A model with positive income effect}) Consider $\{u(t,q;\theta)=\theta \sqrt{q}-t^{2}\mid \theta>0\}$.
To see that this is a  single-crossing domain note that for any bundle with positive $t$ and positive $q$, the slope of an indifference curve is given by $\frac{4t\sqrt{q}}{\theta}$. Slopes are unequal for any two distinct values of $\theta$, therefore the single-crossing property is satisfied. The domain is rich because $\theta\in ]0,\infty[$.                
This domain satisfies the notion of positive income effect introduced in \citep{Mishra1}.
Let $\theta \sqrt{q^*}-t_1^{'2}=\theta \sqrt{q^{**}}-t_1^{''2}$, where $q^{*}<q^{**}, t_1'<t_1''$.
Also let $\theta \sqrt{q^*}-t_2^{'2}=\theta \sqrt{q^{**}}-t_2^{''2}$, where $t_2'<t_2''$. 
Also let $t_1'<t_2',t_1''<t_2''$.
 Then $\theta [\sqrt{q^{**}}-\sqrt{q^{**}}]=t_1^{''2}-t_1^{'2}=t_2^{''2}-t_2^{'2}$. 
 Thus $t_2''^2-t_1''^2=t_2'^2-t_1'^2$. Hence $[t_2''-t_1''][t_2''+t_1'']=[t_2'-t_1'][t_2'+t_1']$. 
 Now $[t_2''+t_1'']>[t_2'+t_1']$. Thus $t_2''-t_1''< t_2'-t_1'$. This entails
 $t_2''-t_2'<t_1''-t_1'$, which is required for positive income effect.     
\label{ex:index}
\end{example}

\begin{example}\rm ({\bf A model when the pay-off relevant parameter concerns with payment})
Consider  $\{u(t,q;\theta)=q - \theta t^2| \theta>0\}$. By an argument analogous to the one applied earlier this is a rich single-crossing domain. Let $q$ denote the probability of winning the object. Cost of paying for the object increases as $\theta$ increases. We may interpret this situation as higher $\theta$ means lower willingness to pay for the object. 
\label{ex:payment}
\end{example}	

\noindent The examples above assume transformations of $t$ and $q$ to have the same curvature.    
This need not be the case. For example a strictly increasing transformation of $q$ maybe convex around $q=0$ such that     
low values of $q$ are considered almost equal to $0$, and at the same time the transformation maybe 
concave around $q=1$ such that high values of $q$ are considered almost equal to $1$.    
Such a transformation of $q$ is an abstract example  of the idea of probability weighting functions discussed in \citep{Kahn1}.     
The preferences in all examples discussed so far are described by one pay-off relevant parameter, i.e., a parameter that appears in the utility representations. In all such situations we can define an order preserving homeomorphism with the positive real numbers naturally. For example, the preferences in Example \ref{ex:qlin} define a function $u(t,q;\theta)\mapsto \theta$. 
The set of preferences in Example \ref{ex:qlin} is a one dimensional topological manifold as well as a one dimensional parametric class. 
Next we give an example that demonstrates that although $\mathcal{R}^{rsc}$ are one dimensional topological manifolds they are not one dimensional parametric classes. Thus from the perspective of pay-off relevant parameters $\mathcal{R}^{rsc}$ can be multidimensional, yet topologically they are one dimensional. The pay-off relevant parameters are called types, and hence our characterization of strategy-proof mechanisms is a study of mechanism design in multidimensional types spaces. The geometry of strategy-proofness is thus a result of a more general structure of preferences, i.e., the single-crossing property, and for such a geometric understanding we may not focus on parametric representations of preferences.                                                   

\begin{example}\rm ({\bf A model of multidimensional types}) Consider the following sets: $U=\{u(t,q)=\theta\sqrt{q}-t^{2}|\theta\in ]0,2]\}$, and $V=\{v(t,q)=2\sqrt{q}-\alpha t^{2}|\alpha\in ]0,1]\}$. We show that $U\cup V$ is single-crossing and rich. 
We show that $U\cup V$ is a single crossing domain. In the interior of $\mathbb{Z}$ slopes of indifference curves of preferences  from $U$ are given by $\frac{4t\sqrt{q}}{\theta}$. 
If $\frac{4t\sqrt{q}}{\theta'}=\frac{4t\sqrt{q}}{\theta''}$, then $\theta'=\theta''$. Slopes of indifference curves of  preferences from $V$ is given by $2\alpha t\sqrt{q}$. Analogously, if $2\alpha^{'}t\sqrt{q}=2\alpha^{''}t\sqrt{q}$, then $\alpha^{'}=\alpha^{''}$. Hence, within $U$ and $V$ preferences satisfy the single-crossing property. We need to show that the single-crossing property holds across $U$ and $V$. If single-crossing did not hold across $U$ and $V$, then for some $(t,q)$ in the interior of $\mathbb{Z}$ and for some  for some $\theta,\alpha$ we have $\frac{4t\sqrt{q}}{\theta}=2\alpha t\sqrt{q}$. 
This implies $\alpha=\frac{2}{\theta}$. Since $\theta\in ]0,2]$, we must have $\alpha\geq 1$. Since   $\alpha\in ]0,1]$, $\alpha=\frac{2}{\theta}$ holds for $\theta=2$ and $\alpha=1$. However, for these values of $\theta=2$ and $\alpha=1$ we obtain a only one preference which is $2\sqrt{q}-t^{2}$, i.e.,  $U \cap V=\{2\sqrt{q}-t^{2}\}$.  
Thus, $U \cup V$ satisfies the single crossing property.

We show that $U\cup V$ is rich.  
To see this consider expanding $U$ and $V$, i.e., $U^{*}=\{u(t,q)\mid \theta\sqrt{q}-t^{2}|\theta\in ]0,\infty[\}$ and $V^{*}=\{v(t,q)\mid 2\sqrt{q}-\alpha t^{2}|\alpha\in ]0,\infty[\}$.
Consider $z=(t,q),z'=(t',q')$ diagonal such that $q<q^{'}$, and $t<t'$. Set $\theta=\frac{t^{'2}-t^{2}}{\sqrt{q}'-\sqrt{q}}$, and $\alpha=\frac{2\sqrt{q}'-2\sqrt{q}}{t^{'2}-t^{2}}$, thus $U^{*}$ and $V^{*}$ are rich.  To show that $U\cup V$ is rich we need to show that $z,z'$ lie on either an indifference curve from $U$ or from $V$. Since $U\cup V$ is single-crossing $z,z'$ can lie on an indifference curves of preferences from either  $U$ or $V$ but not both other that the preference in the intersection. To complete the proof that $U\cup V$ is rich, we show that if $z,z'$ do not lie on an indifference curve from a preference from $U$, then  
$z, z^{'}$ do not lie on an indifference curve of a preference from $2\sqrt{q}-\alpha t^{2}, 1<\alpha$.
Let by way of contradiction, $1<\alpha=\frac{2\sqrt{q}'-2\sqrt{q}}{t^{'2}-t^{2}}$. Then, $\frac{t^{'2}-t^{2}}{ 2\sqrt{q}'-2\sqrt{q} }<1$, i.e., $\frac{t^{2}-t^{2}}{ \sqrt{q}'-\sqrt{q}  }<2$. This means there exists $\theta<2$ such that $z,z'$ are in the same indifference curve of a preference from $U$, this is a contradiction. Therefore, our claim is established, and given that $V^{*}$ is rich $z,z'$ lie on an indifference curve from $V$.    
Hence, $U\cup V$ is rich. For high enough benefits, i.e., when  $\theta=2$, the utility functions with $\alpha<1$ describes that disutility from 
per unit cost is less compared to the situation when the benefit is not high enough i.e., $\theta<2$.

We do not have one dimensional pay-off relevant parametric representation of $U\cup V$.  Let us ask whether $2$ represents  $2\sqrt{q}-t^{2}$  or $2\sqrt{q}-\frac{1}{2}t^{2}$. Suppose we break the tie by  letting $2$ represent $2\sqrt{q}-t^{2}$ and $\frac{1}{2}$ represents $2\sqrt{q}-\frac{1}{2}t^{2}$. Then the question is how do we represent $\frac{1}{2}\sqrt{q}-t^{2}$. The only option available is $1$ since $1$ multiplies $t^2$. Then $1$ cannot represent $\sqrt{q}-t^{2}$. Thus we need two parameters to represent $U\cup V$. It is not critical that two parameters appear in the preferences in $U\cup V$, rather the critical point is that none of the parameters are redundant.
If valuations are required to be pay-off relevant parameters that appear in the utility functions, then the real number that is associated with a preference in $\mathcal{R}^{rsc}$ cannot be called a valuation.
First, as we have seen, a rich single-crossing domain being one dimensional topological manifold does not 
mean that it is also parametrically one dimensional. Also, since 
  any open interval is homeomorphic to the real line there is nothing special about the homeomorphism defined in Theorem \ref{thm:lin_cont}. For example $\mathcal{R}^{rsc}$ is homeomorphic to both $]0,1[$ and $]7, 8[$, and thus it is not clear which interval should represent valuation. 
We note that $U\cup V$ is not convex.{\footnote{A set of real valued functions is convex if for any two functions $u',u''$ from the set, the function $\lambda u'+(1-\lambda )u'', 0\leq \lambda \leq 1$ is in the set. The addition $\lambda u'+(1-\lambda )u''$ is defined pointwise.}}

\label{example:two-parameter}
\end{example}

\begin{remark}\rm The preferences in Example \ref{example:two-parameter} has a two parameter representation. This example can easily be extended to include higher dimensional parametric representations. Let $U=\{u(t,q)=\theta\sqrt{q}-t^2\mid \theta\in ]0,1]\}\cup V=\{v(t,q)=\sqrt{q}-\alpha t^2 \mid \alpha \in [\frac{1}{2},1[\}\cup W=\{w(t,q)=\delta\sqrt{q}-\frac{1}{2}t^2\mid \delta\in [1,\infty [\}$.  One possible parametric representation of $U \cup V\cup W$ is $\{(x,y,z)\mid x\in ]0,1], y=z=0\}\cup \{(x,y,z)\mid y\in [\frac{1}{2},1[, x=z=0\} \cup \{(x,y,z)\mid z\in [1,\infty[, x=y=0\}$. We can create further partition such that preferences are represented by more parameters. 
However, we notice that once we consider $\mathcal{R}^{rsc}$ as a topological manifold we do not need to consider complex parametric representations. Further, as we demonstrate later the topological property and the ordinal feature of the single-crossing property of classical preferences are enough to provide a characterization strategy-proofness and optimal mechanisms for $\mathcal{R}^{rsc}$.           
\end{remark}	           
       
\noindent Our examples above demonstrate that there are domains that are single-crossing and not quasilinear. The converse is also true. For example, $\{q^{\delta}-t\mid 0<\delta<1 \}$ is quasilinear and not single-crossing. For instance, at the bundle $(t=1,q=\frac{1}{8})$ the indifference curves of $q^{\frac{1}{3}}-t$ and $q^{\frac{2}{3}}-t$ are tangential. But we note that
the indirect utility or the value function for these preferences are not maximum of affine functions in $\delta$, thus we cannot study strategy-proofness by applying the standard proof techniques from \citep{myer}. The next example 
demonstrates that we can embed subsets of these parametric class of utility functions  in single-crossing preference domains.

\begin{example}\rm ({\bf Quasilinear preference non-linear in parameter}) The slope of an indifference curve for a preference represented by $q^{\delta}-t$
	is $\frac{1}{\delta q^{\delta-1}}$. Consider $U^{\delta}=\{q^{\delta}-t\mid \delta \in [\frac{1}{4},\frac{1}{3}]\}$. Let $f(\delta)=\delta q^{\delta-1}$. Then $f'(\delta)=q^{\delta}[\delta\ln(q)+1]$. Let $q_{\frac{1}{3}}$ be such that 
	$\frac{1}{3}(q_{\frac{1}{3}})^{\frac{1}{3}-1}[\frac{1}{3}\ln(q_{\frac{1}{3}})+1]=0$.  
	Since $\frac{1}{3}(q_{\frac{1}{3}})^{\frac{1}{3}-1}>0$, the equality implies 
	$\frac{1}{3}\ln(q_{\frac{1}{3}})=-1$. Now $\frac{1}{3}\ln(q_{\frac{1}{3}})=-1\iff \ln(q_{\frac{1}{3}})=-3$. Thus $\ln(q_{ \frac {1}{4} })=-4$, and $q_{\frac{1}{4}}<q_{\frac{1}{3}}$. 
	Then, let $q^{*}$ be such that $q_{ \frac {1}{3}   }   <q^*<1$. Thus, for every
	 $q\in [q^*,1]$, for all 
		$\delta\in [\frac{1}{4},\frac{1}{3}]$, $f'(\delta)>0$. Thus for any $q\in [q^*,1]$, 
		$\frac{1}{\delta'' q^{\delta''-1}}<\frac{1}{\delta' q^{\delta'-1}}$
if $\frac{1}{4}\leq \delta'<\delta''\leq \frac{1}{3}$. That is in $q\in [q^*,1]$ 
$U^{\delta}$ satisfies the single crossing property. Now $g:[\frac{1}{4},\frac{1}{3}]\rightarrow [\frac{1}{8}, \frac{1}{2}]$ be a strictly increasing function which is a bijection. We use this bijection to create a single-crossing domain in the following way. 

To begin with consider $\delta=\frac{1}{3}$. We have $g(\frac{1}{3})=\frac{1}{2}$. Consider a bundle $(t',q^*)$, where $q^*$ is defined in the preceding paragraph. Consider the indifference curves $\{(t,q^*)\mid q^{* \frac{1}{3}}-t= q^{* \frac{1}{3}}-t' \}$
and $\{(t,q^*)\mid \frac{1}{2}q^{*}-t= \frac{1}{2}q^{*}-t' \}$. Define a preference
$\overline{R}$ whose indifference curves for all $(t,q)$ with $q\geq q^*$ are given by 
the ones from $q^{\frac{1}{3}}-t$ and for $q\leq q^*$ are given by $\frac{1}{2}q-t$. 
The preference $\overline{R}$ is monotone and continuous since both the upper contour sets and and lower contour sets are closed in $\mathbb{Z}$. Analogously append indifference curves of $\delta$ and $g(\delta)$. Since the slope of indifference curves of $\theta q-t$ decreases in $\theta$ we have
obtain a compact single-crossing domain $[\underline{R},\overline{R}]$, where $\underline{R}$ is obtained $\delta=\frac{1}{4}$ and $\theta=\frac{1}{8}$. We note that at $(t,q^*)$
$q^{*\delta}-t= g(\delta)q^*-t $ is not required to define the preferences. In fact, this equality does not hold in general. We can label the indifference curves of the new preference as we wish. The point is, $\overline{R}$ has a utility representation, however that utility function cannot be written in terms of $q^{\frac{1}{3}}-t$ and $\frac{1}{2}q-t$. But then this is another advantage of our framework. Since our analysis uses only the ordinal features of the single-crossing property of classical preferences and does not depend on functional form of the utility functions. Figure 0 depicts  
the preferences $\overline{R}$, i.e., the preference corresponding to $\delta=\frac{1}{3}$, $\theta=\frac{1}{2}$ and $\underline{R}$, i.e., the preference corresponding to $\delta=\frac{1}{4}$, $\theta=\frac{1}{8}$.   

   \begin{center}
   	\includegraphics[height=7cm, width=10cm]{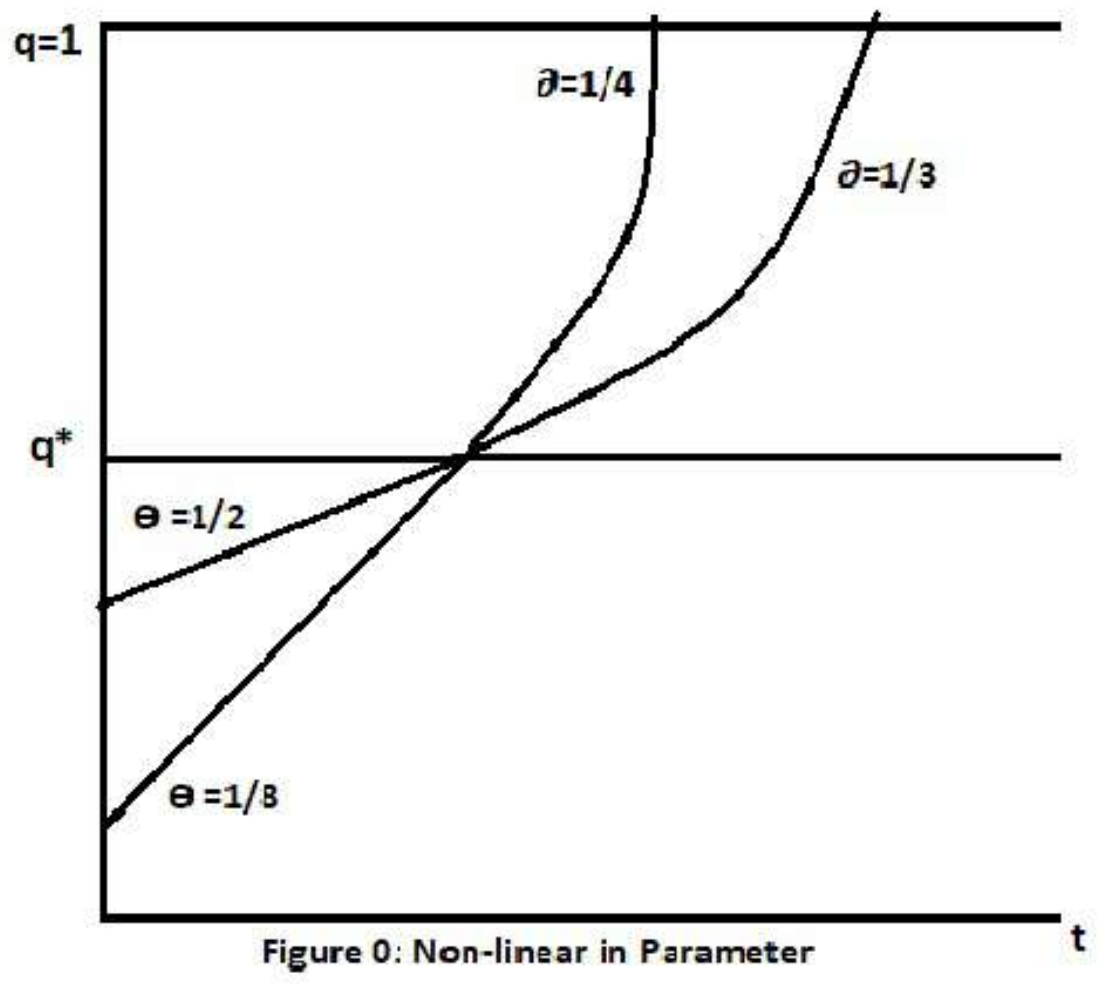}
   \end{center}   
     	
\label{example:delta}	
\end{example}

 If distinct parametric classes are interpreted as distinct behavioral patters, then our set up allows these distinct behavioral patterns to study within one unified setup. We note that in our set up the range of $q$ being $[0,1]$ is not crucial. For example if a seller does not want to sell a good with probability less than $q^*$ or does not want to sell shares below $q^*$, where $q^*$ is as defined in Example \ref{example:delta}, then also for $\delta\in [\frac{1}{4},\frac{1}{3}]$ we can carry the analysis  Now we proceed to study strategy-proof mechanisms.

\section{\bf Strategy-proof Mechanisms in $\mathcal{R}^{rsc}$}
\label{sec:sp}

Consider a domain $\mathcal{R}^{rsc}$. We study mechanisms or social choice function that maps reported 
preferences from $\mathcal{R}^{rsc}$ to a bundle. We define a strategy-proof next.

\begin{defn} [Strategy-proof Mechanisms]\rm A function $F:\mathcal{R}^{rsc}\rightarrow \mathbb{Z}$ is called a {\bf mechanism}. A mechanism is {\bf strategy-proof} if for every $R', R''\in \mathcal{R}^{rsc}$,
	$F(R')R'F(R'')$.   
\label{defn:sp}
\end{defn}

\noindent If $R'$ is the buyer's actual preference, and she reports $R''$, then $F(R')R'F(R'')$ means that the buyer has no incentive to misreport her preference since the bundle that she is allocated by $F$ at $R''$, i.e., $F(R'')$, does not make her better off relative to the bundle $F(R')$ which she is allocated at $R'$.  
We also consider a local version of the notion of strategy-proofness. 
The set $Rn(F)=\{z\in \mathbb{Z}\mid F(R)=z,~\text{for some}~R\in \mathcal{R}^{rsc}\}$ denotes the range of $F$. 

\begin{defn}\rm For a mechanism $F$ we call $Rn(F)$ to be {\bf ordered} if and only if for all $x',x''$ in $Rn(F)$ with  $x'\neq x''$ implies  $x'$ and  $x''$ are diagonal, i.e., either $x'<x''$ or $x''<x'$.     
\end{defn} 

\noindent We define the notion of local strategy-proofness next.

\begin{defn}\rm For the mechanism $F$ let $Rn(F)$ be ordered. Let $F(R')=x',F(R'')=x'' $ and $x'<x''$. Further, let for all $z\in Rn(F)$ with $z\neq x',z\neq x''$ either $z<x'$ or $x''<z$. 
In this case we call $x', x''$ to be {\bf adjacent}. We say $F$ is {\bf locally strategy-proof in range}
if and only if for all adjacent bundles $F(R'),F(R'')$, $F$ satisfies $(i)$ $F(R')R'F(R'')$ and $(ii)$ $F(R'')R''F(R')$.    	 
\label{defn:sp_local}
\end{defn}

\noindent Local strategy-proofness in range requires the buyer to be unable to obtain a better adjacent bundle by misreporting her preference.
We show that if $F$ is strategy-proof, then $Rn(F)$ is an ordered set. Therefore,  if $F$ is strategy-proof, then $F$ is locally strategy-proof in range.  
A different notion of local strategy-proofness is often used in the studies of the implications of local strategy-proofness for strategy-proofness. 
\citep{Carroll} is the first paper to define this notion by means of adjacent preferences. Since our paper is not about characterizations of local strategy-proof mechanisms, we do not discuss it further.
Further, we do not assume that $F$ is local strategy-proof in range, rather it is a consequence of monotonicity of $F$ and continuity of the corresponding indirect preference correspondence defined below.        
The two component functions of $F$ are also denoted by $t,q$, i.e.,  
$F(R)=(t(R),q(R))$. The following is the standard equivalent formation of the notion of an ordered range in terms of monotonicity of $F$.

\begin{defn}\rm 
A SCF $F:\mathcal{R}^{rsc}\rightarrow \mathbb{Z}$ is \textbf{monotone with respect to the order relation} $\prec$ on $\mathcal{R}^{rsc}$ or simply \textbf{monotone} if for every $R', R''\in \mathcal{R}^{rsc}$,
[$R'\prec R''\iff F(R') \leq F(R'')$].
\label{defn:mon}
\end{defn}

\noindent The following lemma shows that if $F$ is strategy-proof, then $F$ is monotone, i.e.,
$Rn(F)$ is an ordered set for strategy-proof mechanisms.       

\begin{lemma}\rm 
 Let  $F:\mathcal{R}^{rsc}\rightarrow \mathbb{Z}$ be a strategy proof SCF. Then $F$ is monotone.
 \label{lemma:mon}
\end{lemma}

\begin{proof} See the Appendix. 
	
 \end{proof}	

 \noindent A strategy-proof mechanism induces an indirect utility function if an utility representation of a preference is given. In this paper we study strategy-proof mechanisms by analyzing  ordinal properties of preferences. Thus, instead of an indirect utility function we consider an indirect preference correspondence. This correspondence is defined next.

\begin{defn} [Indirect Preference Correspondence] \rm 
Consider a mechanism $F$. We call the set valued mapping, $V^{F}:\mathcal{R}^{rsc}\rightrightarrows \mathbb{Z}$ defined as $
V^{F}(R)=IC(R, F(R))=\{z\in \mathbb{Z}\mid z I F(R)\}$ {\bf indirect preference correspondence}. 
\end{defn}
 
\noindent For each preference $R$, $V^{F}$ picks the set of bundles $z$ that are indifferent to $F(R)$ under $R$. For $R',R''$ by $R'\precsim R''$ we mean either $R'=R''$ or $R'\prec R''$. 
Since $\mathcal{R}^{rsc}$ has the least upper bound property, any bounded monotone, i.e., either increasing or decreasing, sequence $\{R^{n}\}_{n=1}^{\infty}$ converges to its supremum if the sequence is increasing and converges to its infimum if the sequence is decreasing. The following notion of continuity of an indirect preference correspondence is crucial for the analysis of strategy-proof mechanisms.        

\begin{defn}[Continuous Indirect preference correspondence]\rm We call
$V^{F}:\mathcal{R}^{rsc}\rightrightarrows \mathbb{Z}$ {\bf continuous}, if for any $R$ and any monotone sequence $\{R^{n}\}_{n=1}^{\infty}$ converging to $R$: (a) the sequence $\{F(R^n)\}_{n=1}^{\infty}$ 
converges and, (b) $z=\lim_{n\rightarrow \infty}F(R^n)I F(R)$.
\label{defn:cont}
\end{defn}

\noindent Suppose $U$ is a utility representation of $R$. Further suppose $V^{F}(U)=U(F(U))$. Then the notion of continuity of $V^{F}$ becomes the standard notion of continuity of a function.
That is, given that the space of utility functions that represents the preferences is a metric space,
$V^{F}$ is continuous if $U^{n}\rightarrow U$, then $V^{F}(U^{n})\rightarrow V^{F}(U)$. 
Consider Example \ref{ex:qlin}. Let $F:]0,\infty[\rightarrow \mathbb{Z}$ be a mechanism. 
Then $V^{F}(\theta)=\theta q(\theta)-t(\theta)$ is a function, and thus instead of a continuous
indirect preference correspondence we have a continuous indirect utility function. We do not assume $V^{F}(\theta'')=V^{F}(\theta')+\int_{\theta'}^{\theta''}q(r)dr$ for $\theta'<\theta''$ as a sufficient condition in our characterization of strategy-proof mechanisms. This integral equation gives the revenue equivalence in \citep{myer} and is one of the sufficient conditions for strategy-proofness. 
We only require continuity of $V^{F}$ to study preferences in Example \ref{ex:qlin}. 
In particular, we do not require absolute continuity of $V^{F}$ in any closed bounded interval
to study strategy-proof mechanisms for quasilinear preferences. The relaxation of the integral equation as a sufficient condition provides an avenue to study the more general notion of single-crossing domains, see Theorem \ref{thm:implies_strtagey_proof_finite_range}.
Next we show that if $F$ is a strategy-proof mechanisms, then  the correspondence $V^{F}$ is continuous.
The following intermediary result is important.

\begin{lemma}\rm Let $z',z''\in \mathbb{Z}$ and $z'<z''$. Let $R\in \mathcal{R}^{rsc}$. 
	Then the following hold: $(i)$ if $z'I z''$, then $z'P^{*}z''$ for all $R^{*}\prec R$; and $z''P^{*}z'$ for all $R\prec R^{*}$,  
	$(ii)$ if $z'Pz''$, then $z'P^{*}z''$ for all $R^{*}\prec R$; and   
	if $z''Pz'$, then $z''P^{*}z'$ for all $R\prec R^{*}$, where $R^{*}\in \mathcal{R}^{rsc}$.   
	\label{lemma:preference_preserve}
\end{lemma}	

\begin{proof} See the Appendix.

\end{proof}         

\noindent Lemma \ref{lemma:preference_preserve}, the second part of the lemma in particular, implies that the order on the set of preferences is preserved over diagonal bundles. The notions of single-crossing property defined in \citep{Saporiti} and \citep{Baisa2} are analogous to the second part of Lemma \ref{lemma:preference_preserve}.

\begin{lemma}\rm 
Let $F:\mathcal{R}^{rsc}\rightarrow \mathbb{Z}$ be strategy-proof. Then $V^F:\mathcal{R}^{rsc}\rightrightarrows \mathbb{Z}$ is continuous.
\label{lemma:cont_correspondence}
\end{lemma}

\begin{proof} See the Appendix. 
	
\end{proof}

\noindent We can summarize the observations from Lemma \ref{lemma:cont_correspondence} and Lemma \ref{lemma:mon} in the following theorem.

\begin{theorem}\rm Let $F:\mathcal{R}^{rsc}\rightarrow \mathbb{Z}$ be strategy-proof. Then $F$ is monotone, and $V^{F}$ is continuous. Further, $F$ is locally strategy-proof in range.    
\label{thm:sp_implies}	     
\end{theorem}

\noindent To prove the converse of Theorem \ref{thm:sp_implies} we require to put some restrictions on the range of $F$.  
To prove the converse of Theorem \ref{thm:sp_implies} we assume that $Rn(F)$ is at most countable with finitely many limit points. First we prove the result for the finite case here, and the result for finitely  many limit points is included in Section \ref{sec:count}.

Next we give  two examples that demonstrate that our two axioms are independent.       
In the next example we construct mechanism where $V^F$ is continuous and $F$ is not monotone.  

\begin{example}\rm Consider the quasilinear preferences $\{u(t,q;\theta)=\theta q-t\mid \theta >0\}$. 
	Let $\underline{\theta}-t^1=0$, and $\overline{\theta}-t^2=0$. Let $\underline{\theta}<\overline{\theta}$. Thus, $t^1<t^2$. Let $F(\theta)=(t^2,1)$ if $\overline{\theta}\leq \theta$, $F(\theta)=(0,0)$, if $\theta\in ]\underline{\theta},\overline{\theta}[$, $F(\theta)=(t^1,1)$ if $\theta<\underline{\theta}$.
	Then $V^{F}(\theta)=\theta-t^2$ if $\overline{\theta}\leq \theta$, $V^{F}(\theta)=0$ if $\theta\in ]\underline{\theta},\overline{\theta}[$ and  $V^{F}(\theta)=\theta-t^1$ if $\theta\leq \underline{\theta}$. Then $V^{F}$ is continuous, and $F$ is not monotone. 
	\label{ex:non_mon} 
\end{example}

\noindent In the next example we construct mechanism where $V^F$ is not continuous and $F$ is monotone.  

\begin{example}\rm Consider the quasilinear preferences $\{u(t,q;\theta)=\theta q-t\mid \theta >0\}$. 
	Let $(t^1,q^1)<(t^2,q^2)$. Let $\theta q^1-t^1=\theta q^2-t^2$. 
	Let $\theta^*<\theta$. 
	Let $F(\theta)=(t^1,q^1)$ if $\theta<\theta^*$ and $F(\theta)=(t^2,q^2)$ if $\theta^*\leq \theta$. 
	Thus, $F$ is monotone. To see that $V^F$ is not continuous consider a sequence $\theta^n\uparrow \theta^*$. $\lim_{n\rightarrow \infty} F(\theta^n)=(t^1,q^1)$. Thus, $\lim_{n\rightarrow \infty}\theta^n q^1-t^1=\theta^* q^1-t^1> \theta^* q^2-t^2$. 
	Hence $V^F$ is not continuous.   	
	
	\label{ex:non_cont}	
\end{example}

\noindent Before we proceed to study strategy-proof mechanisms we make observations about a few fundamental differences between our research and \citep{Laffont1}  and \citep{Tian}. 
These are parametric models of utility functions. 
In \citep{Laffont1} differentiating with respect $\theta$ is well defined and the single-crossing condition in \citep{Laffont1} uses differentiability of the utility functions. The single-crossing condition in \citep{Tian} is analogous to the second part of Lemma \ref{lemma:preference_preserve}. The implementability condition is defined by the operator ``$\max$''.  Our approach is ordinal. To begin with the single-crossing condition in our research is defined in terms of intersection of indifference sets and not by using ordering over parameters. In fact we derive an order on the set of preferences by using the classical properties which follows from \citep{Goswami}. Our axioms about monotonicity of mechanism and continuity of the indirect preference correspondence are ordinal. The mechanism as a differentiable function of $\mathcal{R}^{rsc}$ is not defined in our paper. 
Topologically, $\mathcal{R}^{rsc}$ is a one dimensional manifold, in fact it is globally Euclidean, but that does not mean that preferences are one dimensional type/parametric  spaces. We have given various examples above to clarify this point.          
Although our ordinal approach to study strategy-proofness maybe considered more general than an approach that crucially uses utility representations, it does not mean these papers are special cases of ours. In fact, neither this paper is a special case of the other two papers nor the other two papers are special cases of our research. 
The objectives of these papers are very different and they are incomparable with ours.  \citep{Laffont1} study differential allocation rules and in \citep{Tian} study cycle monotonicity. In the latter utility numbers are important. Our objective is different. We want to study the geometry of strategy-proofness ensued from the single-crossing property and find a way to compute the optimal mechanism by using this geometry which does not depend on utility representations. We have found one, which uses only finitely many incentive constraints, and given that
computations on abstract topological spaces is a very active area of research our optimization program may turn out to be useful.

\subsection{\bf Analysis of Strategy-proofness when $Rn(F)$ is Finite}              

\begin{theorem}\rm Let $F: \mathcal{R}^{rsc}\rightarrow \mathbb{Z}$ be monotone, and $V^F$ be continuous. Further, let $Rn(F)$ be finite. Then $F$ is strategy-proof.
\label{thm:implies_strtagey_proof_finite_range}	
\end{theorem}	

\begin{proof} First we show that  monotonicity of $F$ and continuity of $V^{F}$ imply that $F$ is locally strategy-proof in range. Then using the single-crossing property strategy-proofness is extended to $Rn(F)$.   See the Appendix for details.  
	
\end{proof}

\noindent The proof of Theorem \ref{thm:implies_strtagey_proof_finite_range} that we present is constructive, i.e., the proof reveals how any strategy-proof mechanism looks like. We consider an example of such a rule to illustrate the geometry. Suppose $\#Rn(F)=3$, $\#$ denotes the number of elements in a set. Let $Rn(F)=\{a,b,c\}$, and let $a<b<c$. The proof of Theorem \ref{thm:implies_strtagey_proof_finite_range} shows that if a $F$ is monotone and $V^{F}$ is continuous, then $F$ must be defined as follows.

$$F(R) = \begin{cases}
a, & \text{if $R \prec R_{1}$;}\\
\text{either}~ a~\text{or}~ b, & \text{if $R=R_1$}\\
b & \text{if $R_1\prec R \prec R_2$}\\
\text{either}~b~\text{or}~c & \text{if $R=R_2$,}\\
c & \text{if $R_2\prec R$}
\end{cases}$$

\noindent where $aI_1b$ and $bI_2c$. To see that $F$ is strategy-proof, without loss of generality consider $R$ such $F(R)=b$. By Lemma  \ref{lemma:preference_preserve}, $bPa$ and  
$bPc$. In Figure $1$ we depict this mechanism geometrically. The indifference curves that are drawn in Figure $1$ are convex for the sake of simplicity. We do not require convexity of preferences in our proofs.  

\begin{center}
\includegraphics[height=8.5cm, width=12cm]{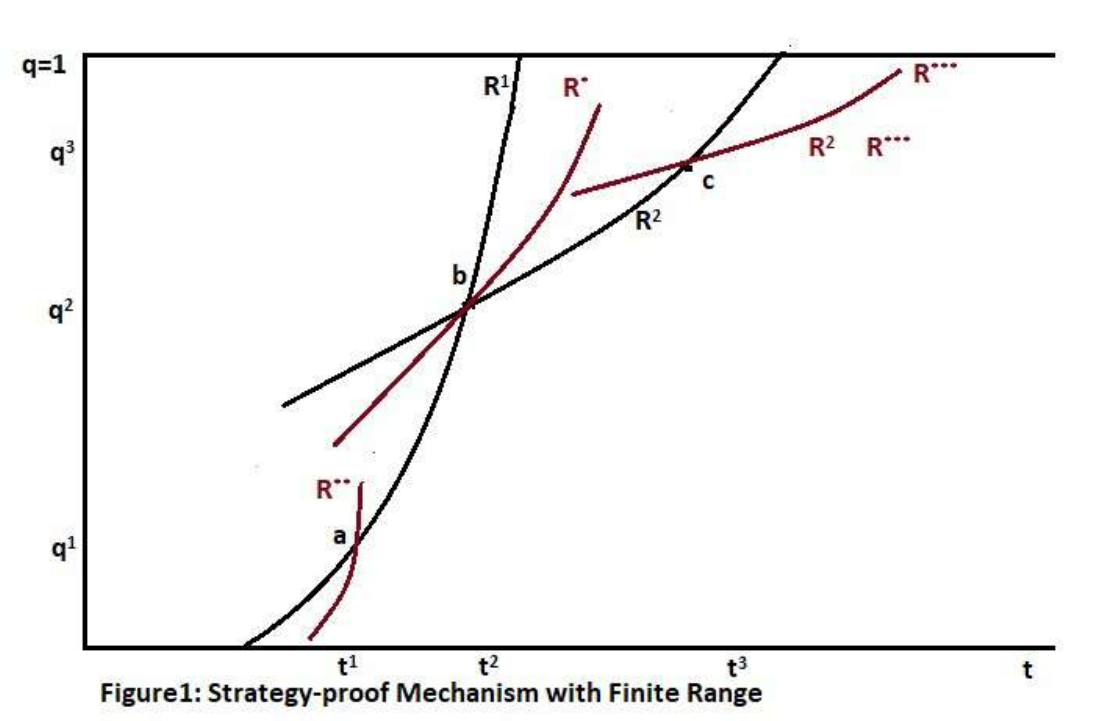}
\end{center}

\noindent In Figure $1$, $a=(t^1,q^1), b=(t^2,q^2),c=(t^3,q^3)$, $R^{**}\prec R^{1}, R^1\prec R^{*}\prec R^{2}$ and $R^{2}\prec R^{***}$. Another possible representation of $F$ with three elements is when all three bundles fall on a single indifference cure, 
i.e., $aIbIc$ for some $R$. By the single-crossing property $F(R)=b$.  
The crucial aspect of our constructive proof is that given a set of the bundles, i.e.,
given their ``coordinates'', monotonicity of $F$ and continuity of $V^{F}$ fix the preferences under which at least two, and maximum three of these, bundles are indifferent. By the single-crossing property these preferences are unique. 
Once these special preferences are obtained the only task that remains is to fix the allocations for preferences in between any two such special preferences. The special preferences in Figure $1$ are $R^1$ and $R^2$. Further, we can easily read off the formula to compute the expected revenue from this geometry. Let $\mu$ be a probability measure that is defined on the intervals from $\mathcal{R}^{rsc}$. 
The formal definition of this measure is provided in Section \ref{sec:opt}. 
For example, the expected revenue from the mechanism in Figure $1$ is given by
$t^1\mu(]-\infty,R^{1}])+t^2\mu(]R^{1},R^{2}])+t^{3}\mu([R^{2},\infty[)$, where $\mu(\{R\})=0$ for all $R\in \mathcal{R}^{rsc}$. If the revenue is of the form $t^1q^1$, then replace $t^1$ by $t^1q^1$ and so on.  
Next we make a remark about ``monotonicity'' condition that is used characterization of strategy-proof mechanisms. 

\begin{remark}\rm \citep{Mishra1} study dominant strategy incentive compatible (DSIC) mechanisms and uses  a generalized version of the weak monotonicity condition studied in \citep{Bik}. The latter considers quasilinear preferences. We try to see how the monotonicity condition in \citep{Mishra1} looks like in the context of our model.
Let $F(\widehat{R})=(t',q')$, and $\widehat{R}\prec R$. Let $F(R)=(t'',q'')$. Let $q''> q'$. Then,  $t''>t'$ and $(t'',q'')I(t',q')$. According to the the notations in \citep{Mishra1} 
$z=(t(\widehat{R}), q(\widehat{R}))$, $a=q''$, $V^{R}(q'',z)=t''$. Thus if $(t',q'), (t'',q'')$ are adjacent bundles $R$ is one of the special preference that we identify in our characterization. 
\citep{Mishra1} provide an extensive literature review on the ``monotonicity'' condition used in the characterization of DSIC mechanisms for quasilinear preferences.        	     
	
\label{remark:mon}	
\end{remark}

\noindent We end this section with the definition of an individually rational mechanism.

\begin{defn}\rm A mechanism $F:\mathcal{R}^{rsc}\rightarrow \mathbb{Z}$ is {\bf individually rational}
if and only if for all $R$, and for all $t$, $F(R)R(t,0)$.     
\end{defn}	   

\noindent The bundle allocated by the mechanism $F$ does not make the buyer worse off relative to any bundle
that has $q=0$. By monotonicity of $R$, $F(R)R(0,0)\implies F(R)R(t,0)$. Thus without loss of generality we shall consider a mechanism to be individually rational if $F(R)R(0,0)$.

\subsection{\bf Optimal Mechanisms when $Rn(F)$ is Finite}
\label{sec:opt}

We study optimal mechanisms for intervals $[\underline{R},\overline{R}]$ as a subspace of $\mathcal{R}^{rsc}$ in the order topology. In the proof of Theorem \ref{thm:implies_strtagey_proof_finite_range}, Lemma \ref{lemma:mon} and Lemma \ref{lemma:cont_correspondence} we have made observations that these results hold for 
closed intervals as well.   
Let the Borel sigma-algebra on $[\underline{R},\overline{R}]$ generated by the subspace order topology    
be denoted by ${\cal B}$. Let $\mu$ denote a probability measure on ${\cal B}$. Thus, $([\underline{R},\overline{R}],{\cal B}, \mu)$ is a probability space. We assume $\mu $ to be $(i)$ ({\bf non-trivial}) for any non-empty and non-singleton interval $A\subseteq [\underline{R},\overline{R}]$ that is  $\mu(A)>0$ ; and $(ii)$ ({\bf continuous}) for any $A\in {\cal B}$ with $\mu(A)>0$ and any $c\in \Re$ with $0<c<\mu(A)$, there exists $B\in {\cal B}$ with $B\subseteq A$ and $\mu(B)=c$. Continuity of $\mu$ implies $\mu$ is non-atomic, in particular $\mu(\{R\})=0$ for all $R\in [\underline{R},\overline{R}]$. We assume the real line $\Re$ to be endowed with the standard Borel
sigma algebra obtained from the Euclidean metric space. We denote this sigma-algebra by $B(\Re)$.
We show that if $F$ is strategy-proof, then the component functions $t, q $ are ${\cal B}/B(\Re)$ measurable.

\begin{lemma}\rm Let $F:[\underline{R},\overline{R}]\rightarrow Rn(F)$ be monotone and 
	$Rn(F)$ be finite. Then the component functions of $F$, $t,q$ are $\mathcal{B}/B(\Re)$measurable.      
\label{lemma:measurable}
\end{lemma}	

\begin{proof}  See the Appendix.          	
\end{proof}

\noindent Now we can define expected revenue of the seller.

\begin{defn}\rm Let $[\underline{R},\overline{R}]$ and $\mu$ be given. Let $Rn(F)$ be finite. The expected revenue form $F$ is $\int_{[\underline{R},\overline{R}]}t(R)d\mu(R)$.
We denote the expected revenue from $F$ by $E[F]$. 
Since we fix a probability measure we do not mention $\mu$.       	      
\end{defn}

\noindent The important aspect of the optimal mechanism is the constraint set. Thus without loss of generality we shall assume that revenue of the seller is given by $t$.   
We define optimal mechanism next. 

\begin{defn}\rm Let $[\underline{R},\overline{R}]$ and $\mu$ be given. Let 
	\[\mathcal{F}=\{F: [\underline{R},\overline{R}]\rightarrow \mathbb{Z}\mid F~\text {is strategy-proof, indvidually rational, has finite}~Rn(F)\}.\]
Then $F^{*}\in \mathcal{F}$ is {\bf optimal} if and only if for all $G\in \mathcal{F}$, $E[F^{*}]\geq E[G]$. 
    	
\label{defn:optimal} 
\end{defn}

\noindent Throughout our discussion of optimal mechanisms we assume $\mu$ to be continuous and non-trivial.      	
                   
\begin{theorem}\rm If $\mu$ is non-trivial and continuous, then 
there is $\overline{T}$ such that for all $F\in \mathcal{F}$, $0\leq E(F)\leq \overline{T}$. 	
\label{thm:optimal}
\end{theorem}

\begin{proof} We first show that for range with maximum $l$ number of bundles, where $l$ is an arbitrary positive integer, there exists a mechanism that maximizes expected revenue. The optimal mechanism may have less than $l$ number of bundles in the range.

Let $\mathcal{F}_{l}=\{F_l\mid F_l\in \mathcal{F}, \#Rn(F_l)\leq l\}$, $E(F_{l})$ denotes the expected revenue of a mechanism from $\mathcal{F}_l$.               
Then the optimal mechanism in $\mathcal{F}_{l}$ is denoted by $F^{*}_{l}$. 
By individual rationality $F_l(\underline{R})R(0,0)$. Further, $F_l$ such that $F_l(R)=(0,0)$ for all $R$ is not optimal.    
Without loss of generality we assume that $\#Rn(F_l)\geq 2$. 
If $Rn(F_l)=\{(t^1,q^1)\mid t^1>0\}$, i.e., $\#Rn(F_l)=1$, then by individual rationality we can set $F(\underline{R})=(0,q^{0})$
where $(0,q^{0})\overline{I}(t^1,q^1)$. Since $\mu$ is continuous such readjustment does not change the expected revenue. If $Rn(F_l)=\{(0,q^1)\}$, then let $F_l(\overline{R})=(t^1,1)$ where $(0,q^1)\overline{R} (t^1,1)$. Again by continuity of $\mu$ such readjustment does not change the revenue. 
Further, we do not consider a bundle $(t,q)$ in the range with $0<t^1<\overline{T}, 0<q<1$ such that the bundle is allocated at only for one single $R$. This is because $\mu(\{R\})=0$.

We show that   $F^{*}_l $ exists. In particular, we show that $E(F^{*}_{l})$ is given by 
Equation (\ref{eqn_optimization_problem}). In the following let $l\geq 2$.   
\begin{equation}
	E(F^{*}_{l})=\max_{t^{k},q^{k},R^{k}, k=0,\ldots,l-1} \sum_{k=0}^{l-1}~~t^{k}\mu([R^{k},R^{k+1}])=\max_{F_l\in \mathcal{F}_l}E(F_l)
\label{eqn_optimization_problem}
\end{equation}

	\noindent s.t.
	$R^{0}= \underline{R}, R^{l}=\overline{R}$, $(t^0,q^0)R^{0}(0,0)$, 
	$(t^{k},q^{k})I^{k}(t^{k-1},q^{k-1}), k=1\ldots,l-1$
	$R^{k-1}\precsim R^{k}$  for $k=1,\ldots,l$. 
	$0\leq t^{0}, t^{k-1}\leq t^{k}$ and $0\leq q^{k-1}\leq q^{k}\leq 1$ for $k=1,\ldots,l-1$. 
	\label{lemma:sp_ir_1}

\medskip

\noindent  The objective function in the optimization problem can be identified with the $F_l$  described by the following procedure: $F_l(R^{0})=(t^0,q^{0})$, $F_l(R^{1})=(t^{1},q^{1}),\ldots, F_l(\overline{R})=(t^{l-1},q^{l-1})$. 
Further, $F_l(R)=(t^0,q^{0})$ for all $R\in [\underline{R},R^{1}[$,  $F_l(R)=(t^{1},q^{1})$ for all $R\in [R^{1},R^{2}[$, $F_l(R)=(t^{2},q^{2})$ for all $R\in [R^{2},R^{3}[\ldots$
and $F_l(R)=(t^{l-1},q^{l-1})$ for all $R\in [R^{l-1},\overline{R}]$.
By the single-crossing property $F_l$ is strategy-proof and individually rational. 
In other words, if we solve the maximization problem in Equation (\ref{eqn_optimization_problem}) subject to the constraints, then the maximum obtained can indeed be ascribed to a strategy-proof and individually rational mechanism. Since a strategy-proof mechanism with range $l$ or less is completely determined by the special preference $R^1,\ldots, R^{l-1}$ the sum in Equation (\ref{eqn_optimization_problem}) corresponds is the expected revenue from the mechanism. 

\begin{remark}\rm 
In other words the geometry of strategy-proof mechanisms that follows from the single-crossing property provides us a general way of computing the optimal mechanism. If we are looking for an optimal mechanism that can have at most $l$ bundles, then the number of constraints is finite. This happens because of the single-crossing property since the single crossing property allows us to formulate the optimization problem in terms of the  special preferences and thus we do not need to consider all the incentive constraints. Although it is clear but for the sake of completeness we note that the objective function 
does not searches among preferences, i.e., the optimization program looks for the special preferences, and these preferences do not come from a vector space since we have not imposed any such restriction in this paper.                            
\end{remark}

\noindent Next we show that the optimization problem has a solution. We prove this result in the following lemma.   
The upper bound $\overline{T}$ is given by $(0,0)\overline{I}(\overline{T},1)$. 
One of the steps in the proof of Lemma \ref{lemma:soln_exists} shows that the constraint set of this optimization problem is compact. The other step  shows that the objective function is continuous. This implies that the optimal solution exists
and is a strategy-proof and individually rational mechanism.

\begin{lemma}[{\bf Optimal}]\rm Let $\mu$ be non-trivial and continuous. The optimization problem in Equation (\ref{eqn_optimization_problem}) has a solution.
\label{lemma:soln_exists}
\end{lemma}

\noindent See the Appendix for the rest of the proof.

\end{proof}

\noindent If the $\sup\{E(F)\mid F\in \mathcal{F}\}$ is attained, then $F^{*}$ exists.  Solving an optimization problem analytically maybe hard. However, our results provide a general form of the optimization problem that needs to be solved such that only a finite number of incentive constraints are required to be considered. 
The nature of solution to the optimization problem depends on more specific form of domains and the measure $\mu$. The optimization problem in Equation (\ref{eqn_optimization_problem}) uses only the ordinal feature of preferences, i.e., only $R$ appears in the objective function and not its utility representation. Our result brings out the general nature of the problem of finding the optimal mechanism for rich single crossing preference domains irrespective of how preferences are represented. As an application of our result we solve the optimal mechanism problem for quasilinear preferences next. At the end of Section \ref{subsec:qlin} we also how the optimal mechanism can be extended to scenarios where the number of buyers is more than one.

\subsubsection{\bf Optimal Mechanism with quasilinear preferences}  
\label{subsec:qlin}            
 
The preferences in Example \ref{ex:qlin} are given by $\theta q-t$. Where $q$ denotes a share of a good, and $t$ is the payment that is made by the agent for the share $q$.
That is, the buyer pays $t$, and the seller receives $t$.  Let $\theta\in [\underline{\theta},\overline{\theta}]$. Also let $\underline{\theta}>0$ and $\overline{\theta}<\infty$. 
We can identify the preference $\theta q-t$ by the pay-off relevant parameter $\theta$.
That is the order on the preferences follow the natural order in $[\underline{\theta},\overline{\theta}]$.         
Further, let $\theta$ follow the distribution $\Gamma$, and $\gamma$ is the continuous density. In the next proposition we show that 
if $\Gamma$ satisfies increasing hazard rate, then the optimal mechanism for every 
$l\geq 2$
is deterministic.
A deterministic mechanism has two bundles in its range, one is of the kind 
$(t,1), t>0$ and the other is $(0,0)$. 
In Proposition \ref{prop:qlin} we assume that $\Gamma$ admits increasing hazard rate.    
We keep the proof of Proposition \ref{prop:qlin}
in the text since the proof brings out the interaction between increasing hazard rate and slopes 
of indifference curves which entails the deterministic mechanism. We wish to clarify that the by no means we are claiming that \citep{myer} does not study a model without monotone hazard rate. In fact, in our general model the probability measure is quite general and thus trivially we also do not have any assumption concerning monotone hazard rate. Proposition \ref{prop:qlin} demonstrates that we can study  highly applied results in mechanism design  theory by using our framework as well.

\begin{prop}\rm Let $\mathcal{F}=\{F\mid F:[\underline{\theta},\overline{\theta}]\rightarrow \mathbb{Z}, ~\text{strategy-proof, indvidually ration}$
				$\text{al, has finite}~Rn(F)\}$. 
	Let $\frac{\gamma(\theta')}{1-\Gamma(\theta')}\leq \frac{\gamma(\theta'')}{1-\Gamma(\theta'')}$ if $\theta'<\theta''$. Then for all $l$, $F_{l}^{*}$ is a deterministic mechanism.
		     	
		\label{prop:qlin}   
	\end{prop}          
	
	\begin{proof} Without loss of generality we can assume that $Rn(F)$ has 	
		bundles $(0,0)$ and $(t^{l},1)$. Since $R^{l}=\overline{R}$, by individual rationality at $F(R^{l-1})$
		consider $(t, 1)$ such that $\theta^{l-1}q^{l}-t^{l}=\theta^{l-1}-t$. 
		Then for every $R\in ]R^{l-1},R^{l}]$, 
		set $F(R)=(t,1)$, call this payment $t^{l}$. 
		Considering $(t^l,1)$ is without loss of generality since we are analyzing optimal mechanisms. 
		We fix $q$, and show that the first order conditions imply that the optimal mechanism must have at most two bundles in its range. Consider the optimization problem so that that range can have at most $l$ bundles. By Lemma \ref{lemma:soln_exists} we know that optimum solution exists.    
		Let by way of contradiction  $t^{k-1}<t^{k}$ and $q^{k-1}<q^{k}, \theta^{k-1}<\theta^{k}$.
		Further, we do not have to consider a bundle $(t,q)$ in the range such that $0<t^1<\overline{T}, 0<q<1$ such that the bundle is allocated at only for one single $R$. This is because probability of $\{R\}$ occurring is zero.           
		
	Consider the Lagrange  $L= \sum_{k=0}^{l-1}~~t^{k}[\Gamma(\theta^{k+1})-\Gamma(\theta^{k})]$
		$-\lambda_{1}[\theta^{1}q^{1}-t^{1}]-\sum_{k=2}^{l-1} \lambda_{k} [\theta^{k}q^{k}-t^{k}-\theta^{k}q^{k-1}+t^{k-1}]$. Since $F$ is optimal $t^{0}=0=q^{0},q^{l-1}=1$. Also, $\Gamma(\theta^{l})=\Gamma(\overline{\theta})=1$ and $\Gamma(\theta^{0})=\Gamma(\underline{\theta})=0$. The first order conditions with respect to $\theta^{1}$ is:         
		
		\begin{equation}
		-t^{1}\gamma(\theta^{1})-\lambda_1 q^{1}=0\\
		\implies \frac{\gamma(\theta^{1})}{-\lambda_{1}}=\frac{q^{1}}{t^{1}}. 
		\label{eqn:theta_1_myer}
		\end{equation} 
		\noindent The first order condition with respect to $\theta^k$, $k=1,\ldots,l-1$ is:   
		\begin{equation}
		t^{k-1}\gamma(\theta^{k})-t^{k}\gamma(\theta^{k})-\lambda_{k}[q^{k}-q^{k-1}]=0
		\implies \frac{\gamma(\theta^{k})}{-\lambda_{k}}=\frac{q^{k}-q^{k-1}}{t^{k}-t^{k-1}}.
		\label{eqn:theta-k_myer}
		\end{equation} 
				\noindent The first order condition with respect to $t^{k}$ for $k=1,\ldots,l-2$ is
				\begin{equation}
		[\Gamma(\theta^{k+1})-\Gamma(\theta^{k})]+\lambda_{k}-\lambda_{k+1}=0.
		\label{eqn:t_k_myer}
		\end{equation}
				\noindent The first order condition with respect to $t^{l-1}$ is
				\begin{equation}
		[1-\Gamma(\theta^{l-1})]+\lambda_{l-1}=0.
		\label{eqn:t_l-1-kmyer}
		\end{equation}

		\noindent from equations (\ref{eqn:t_k_myer}) and (\ref{eqn:t_l-1-kmyer})  $\lambda_{k}=-[1-\Gamma(\theta^{k})]$ for $k=1,\ldots, l-1$. Then from equations (\ref{eqn:theta_1_myer}) and (\ref{eqn:theta-k_myer}) we obtain 
		$\frac{\gamma(\theta^{1})}{1-\Gamma(\theta^{1})}=\frac{q^{1}-0}{t^{1}-0}$, and 
		$\frac{\gamma(\theta^{k})}{1-\Gamma(\theta^{k})}=\frac{q^{k}-q^{k-1}}{t^{k}-t^{k-1}}=\frac{1}{\theta^{k}}$ for
		$k=2,\ldots,l-1$. Thus $\theta^{k-1}<\theta^{k}$ implies 
		$\frac{q^{k}-q^{k-1}}{t^{k}-t^{k-1}}<\frac{q^{k-1}-q^{k-2}}{t^{k-1}-t^{k-2}}$. 
		Since $\frac{\gamma(\theta')}{1-\Gamma(\theta')}\leq \frac{\gamma(\theta'')}{1-\Gamma(\theta'')}$ if $\theta'<\theta''$ we reach at a contradiction. Thus, the proof follows.

	\end{proof}
	
\noindent Since $\theta^{k}q^k-t^k=\theta^{k}q^{k-1}-t^{k-1}$, $\frac{q^{k}-q^{k-1}}{t^{k}-t^{k-1}}$ is the slope of the indifference curves of preferences that correspond to $\theta^{k}$. 
The optimal mechanism requires this slope to be equal to the hazard rate at $\theta^{k}$, i.e., 
$\frac{\gamma(\theta^{k})}{1-\Gamma(\theta^{k})}=\frac{q^{k}-q^{k-1}}{t^{k}-t^{k-1}}\equiv \frac{\text{incraese in winnig probability}}{\text{incraese in payment}}$

$\equiv\text{relative, i.e., relative to increase in payment, increase in winning probability}$.

\noindent The hazard rate at $\theta^{k}$ measures the conditional probability that the buyer's type fails to be below $\theta^{k}$. Thus increasing hazard rate in a sense implies that it is more likely that the buyer's type is high and not low. The single-crossing property implies that the
relative increase in the winning probability decreases as types increase. 
The proof of Proposition \ref{prop:qlin} shows that the equality between the hazard rate and the relative increase in the winning probability holds only for a deterministic mechanism.       
The next proposition pins down the the optimal mechanism $F^{*}$.

\begin{cor}\rm Let $\frac{\gamma(\theta')}{1-\Gamma(\theta')}\leq \frac{\gamma(\theta'')}{1-\Gamma(\theta'')}$ if $\theta'<\theta''$. Then $F^{*}$ exists. Further, $F^{*}$ is defined as follows:    		
	
	$$t^{*}(\theta) = \begin{cases}
	\theta^{*}, & \text{if $\theta>\theta^{*}$;}\\
	$0$, & \text{if $\theta\leq \theta^{*}$.}
	\end{cases}$$

	$$q^{*}(\theta) = \begin{cases}
		1, & \text{if $\theta>\theta^{*}$;}\\
		$0$, & \text{if $\theta\leq \theta^{*}$.}
		\end{cases}$$

		\label{cor:qlin_opt}

\noindent Where $t^{*}, \theta^*$ solve $\max _{t,\theta}t[1-\Gamma(\theta)]$ subject to $\theta-t=0$.  
	\end{cor}

	\begin{proof} From Proposition \ref{prop:qlin} it us enough to look for an optimal mechanism within the class of deterministic mechanisms. Then we note that a solution to the optimal problem exists.{\footnote{The constraints satisfy the non degenerate constraint qualification, NDCQ, condition as stated in Theorem $18.5$ in \citep{Simon}.}}     
		
	\end{proof}

\noindent Next we make a remark about extending our ideas to multi-buyer scenarios. 

\begin{remark}\rm We can consider two approaches to extend single-buyer optimal mechanisms to multi-buyer mechanisms. First approach is to consider the multi-buyer version of the optimization program described in Theorem \ref{thm:optimal}. This is straightforward. 
First, consider extending  the constraint set of the one buyer model. To do this we  stack $n$ copies of the variables in the constraint set from the one buyer problem and add a condition $\sum_{i=1}^{n}q_{i}^{k}\leq 1$, where the $k^{\text{th}}$ alternative in the range of a $n$-buyer mechanism is given by 
$(t_1^k,q_1^k, \ldots,  t_n^k,q_n^k)$.  
The objective function of the seller now have a sum of expected revenues from the $n$ buyers. Thus, the extension Theorem \ref{thm:optimal} follows immediately. 
If the number of buyers is large, then it may take a lot of computation time.

The other approach is to consider an axioms that extend one buyer mechanisms to multi-buyer mechanisms.        
Let without loss of generality there be two buyers. 
We demonstrate a way to extend the optimal mechanism obtained for the one buyer scenario to the multi-buyer scenarios for the quasilinear preferences first. 
Consider a two-buyer function $F:[\underline{\theta}, \overline{\theta}]\times [\underline{\theta}, \overline{\theta}]\rightarrow \mathbb{Z}\times \mathbb{Z}$, where the component functions for buyer $i$ are given by 
$t_i,q_i$. Consider the following axiom $F:[\underline{\theta}, \overline{\theta}]\times [\underline{\theta}, \overline{\theta}]\rightarrow \mathbb{Z}\times \mathbb{Z} $ is  {\bf lower-efficient} if  
$\theta_i<\theta_j \implies q_{i}(\theta_i,\theta_j)=0$. Further, let if $\theta_i=\theta_j$, then the tie be broken with equal probability
If $F$ satisfies lower-efficiency, and ties are broken with equal probability  
then $F$ is a mechanism i.e., for all $(\theta_1,\theta_2)$ $q_{1}(\theta_1,\theta_2)+q_{2}(\theta_1,\theta_2)\leq 1$. 
To see how it works, 
consider a function 
$(\theta_1,\theta_2)\mapsto F(\theta_1,\theta_2)=( q_1(\theta_1, \theta_2),t_1(\theta_1,\theta_2), q_2(\theta_1, \theta_2),t_2(\theta_1,\theta_2) )$. 
Further, let for any buyer $i$, for all $\theta_i, \theta_i', \theta_j$, $\theta_{i}q_{i}(\theta_i,\theta_j)+t_{i}(\theta_i,\theta_j)\geq 
\theta_{i}q_{i}(\theta_i',\theta_j)+t_{i}(\theta_i',\theta_j)$.
Just because $F$ satisfies these inequalities does not mean that $F$ is a mechanism. For $F$ to be a mechanism for all $(\theta_1,\theta_2)$, $q_{1}(\theta_1,\theta_2)+q_{2}(\theta_1,\theta_2)\leq 1$ is required. 
Along with $F$ satisfying these inequalities lower efficiency makes $F$ a strategy-proof mechanism.

Lower-efficiency says that 
the buyer with the lower valuation does not obtain the object, at the same time it does not say that 
the buyer with the higher valuation obtains the object. Thus, this axiom does not make a multi-buyer mechanism efficient. In case of selling a good, and in the context of selling shares to at most one buyer this axiom quite reasonable. 
An example of the latter is the airport privation auction of the New Delhi airport, 2003. The rights to carry out the development of the airport  to begin with was to be sold to one construction company. Depending on the context, we may find other axioms to extend the one buyer optimal auction to multi-buyer scenarios.    

To cross validate this axiom we show that it entails the optimal mechanism 
in \citep{myer}.        
Here $\theta_i$ is the preference of the $i^{\text{th}}$ buyer. Assume that both buyers' valuations are drawn independently from the same distribution with support $[\underline{\theta},\overline{\theta}]$.
Let the $\Gamma$ and $\gamma$ denote the distribution and density respectively. 
Let $(\theta_1, \theta_2)$ denote a profile of valuations.  
Consider
$\theta_i>\theta_j$. Thus, let by lower-efficiency $q_{i}(\theta_i,\theta_j)>0\implies \theta_i>\theta_j$. 
Thus, consider the interval $[\theta_j,\overline{\theta}]$ to study buyer $i$.
Assume that $\Gamma$ satisfies the increasing hazard rate.  
By Corollary \ref{cor:qlin_opt} the optimal way to sell the object to buyer $i$ is    

	$$t^{*}_i(\theta_i,\theta_j) = \begin{cases}
	\max \{\theta^{*}, \theta_j\}, & \text{if $\theta_i>\max \{\theta^{*}, \theta_j\}$;}\\
	$0$, & \text{if $\theta_i<\max \{\theta^{*}, \theta_j\}$.}
\end{cases}$$

$$q_i^{*}(\theta_i,\theta_j) = \begin{cases}
	1, & \text{if $\theta_i>\max \{\theta^{*}, \theta_j\}$;}\\
	$0$, & \text{if $\theta_i<\max \{\theta^{*}, \theta_j\}$.}
\end{cases}$$

\noindent The reasoning behind above mechanism is as follows. If there were no other buyer, then buyer $i$ would have obtained the object for all valuations above $\theta^{*} \in [\underline{\theta},\overline{\theta}]$. 
Consider the optimization problem in Corollary \ref{cor:qlin_opt}    
$\theta_i[1-\Gamma(\theta_{i})]$ where $\theta_{i}\in [\underline{\theta},\overline{\theta}]$. 
The critical point is given by 
$\theta_{i}=\frac{1-\Gamma(\theta_i)}{\gamma(\theta_{i})}$. The solution to this equation is $\theta^{*}$ that is obtained in Corollary \ref{cor:qlin_opt}.
We argue that $\theta_i[1-\Gamma(\theta)]$ is increasing up to $\theta^{*}$
and decreasing after $\theta^*$. Define $\psi(\theta_i)=\theta_i-\frac{1-\Gamma(\theta_i)}{\gamma(\theta_{i})}$. 
This function is called virtual valuation in \citep{myer}. 
We have seen $\psi(\theta^*)=\theta^*-\frac{1-\Gamma(\theta^*)}{\gamma(\theta^*)}=0$.
Since the hazard rate is increasing, $\psi$ is increasing. Thus if $\theta_i<\theta^{*}$, then 
$\theta_i<\frac{1-\Gamma(\theta_i)}{\gamma(\theta_i)}$.       
Now $\frac{d \theta_{i}[1-\Gamma(\theta_{i})]}{d\theta_{i}}>0$ if $\theta_{i}<\frac{1-\Gamma(\theta_{i})}{\gamma(\theta_{i})}$. Thus the derivative of $\theta_i[1-\Gamma(\theta_i)]$ is positive if $\theta_i<\theta^{*}$.  
Since $\psi$ is increasing and $\psi(\theta^*)=0$, if $\theta^*<\theta_i$, then 
$\frac{1-\Gamma(\theta_i)}{\gamma(\theta_)}<\theta_i$.
Further, $\frac{d \theta_{i}[1-\Gamma(\theta_{i})]}{d\theta_{i}}<0$ if $\theta_{i}>\frac{1-\Gamma(\theta_{i})}{\gamma(\theta_{i})}$.
Thus the derivative of $\theta_i[1-\Gamma(\theta_i)]$ is negative if $\theta_i>\theta^{*}$.
That the optimal mechanism does not charge more than $\max\{\theta^*,\theta_j\}$ from buyer $i$ if $\theta_i>\theta_j$.

Now note that the condition $\theta_i>\max\{\theta^*, \theta_j\}$ is same as the condition 
$\psi(\theta_i)>\psi(\theta_j)$ and $\psi(\theta_i)>0$, since $\psi$ is increasing and $\psi(\theta^*)=0$. Now if $\theta^{*}<\theta_{j}<\theta_i$, the mechanism fixes the payment to be
$\theta_j$ and not $\theta^{*}$, even though the revenue is sub-optimal compared to the one buyer scenario and the payment $\theta^*$ is feasible. This is because if the payment is $\theta^*$ or less than $\theta_j$, then for $\theta_i'$ such that $\theta^*<t_i(\theta_i,\theta_j)<\theta_i'<\theta_j$
we obtain $\theta_i'q_i(\theta_i',\theta_j)-t_i(\theta_i',\theta_j)=0<
\theta_i'q_i(\theta_i,\theta_j)-t_i(\theta_i,\theta_j)=\theta_i'-t_i(\theta_i,\theta_j)$. This contradicts strategy-proofness. 
The first equality follows by lower-efficiency.  
Thus the multi-buyer mechanism  described above is exactly the mechanism optimal mechanism if we had followed the techniques in \citep{myer}.
Hence, we also provide a geometric understanding of the optimal mechanisms in \citep{myer}.  
We note that the optimal mechanism for quasilinear preferences is an equilibrium outcome of a second price auction with reserve price $\psi^{-1}(0)$. However, note that the preference of buyer $i$ given by $q_i(\theta_i-t_i)$ is not quasilinear, where $q_it_i$ is expected payment. Further, this preference is not monotone everywhere and thus not classical according to Definition \ref{defn:classical}. Thus we define restricted classical preference in Section \ref{sec:non_mon}; and we refer this model as the 
``Myerson model'' to distinguish it from the model with quasilinear preferences.{\footnote{If the buyer has to pay $t$ whether she wins or not, then the payment $t$ in a quasi-linear preference is an expected payment.}} A similar argument shows that the optimal mechanism for these preferences is also deterministic.

The notion of lower-efficiency can be defined for any $\mathcal{R}^{rsc}$. For any  $R$, the payment $t$ such that $(0,0)I(t,1)$ may be interpreted as the maximum willingness to pay for the good. 
In case of preferences such as $\theta q-t$, the maximum willingness to pay is $\theta$, since 
$\theta -t=0=\theta 0-0$. That is, the maximum willing to pay appearing as a  parameter in the pay-off function is a special case. Since in our model the primitive is preference, and not maximum willingness to pay, we do not have be concerned with utility representations in terms of parametric classes.    
We say a two-buyer mechanism $F:[\underline{R}, \overline{R}]\times [\underline{R}, \overline{R}]\rightarrow \mathbb{Z}\times \mathbb{Z} $ is  {\bf lower-efficient} if  
$R_i\prec R_j \implies q_{i}(R_i,R_j)=0$. In the case of the $n$-buyer scenario lower-efficiency can be defined as: 
if $R_i\prec R_j$ for some $j\neq i$, then  $q_{i}(R_i,R_{-i})=0$, where $R_{-i}$
denote the profile where buyer  $i$  is not included.

Another axiom is related to priority over the set of agents. For example, let there be two buyers and the seller wants to sell shares first to buyer $1$ and if anything is left then it will be sold to buyer $2$. Consider the one buyer problem of finding the optimal mechanism for buyer $1$ whose preferences lie in $[\underline{R},\overline{R}]$. Let buyer $1$'s preference be $R_1$ and suppose $F_{1}(R_1,R_2)=(t_1,q_1)$ for all $R_2$. Now fix $R_1$. Then consider the one buyer problem of finding the optimal mechanism for buyer $2$ where $q\in [0,  1-q_1]$. This defines a strategy-proof mechanism. Such priority over the set of agents may arise from the difference in the reputations of the buyers to fulfill the contractual obligations such as making timely payments, see \citep{Houser} and \citep{Englemann} for theory and importance of reputation in auctions.                              

\end{remark}

\noindent Next we wish to make a remark about an important difference between our solution technique of finding the optimal mechanism and that in \citep{myer}. 

\begin{remark}\rm We ask: ``what if we apply techniques in \citep{myer} to solve the optimization problem for quasilinear utility for any given maximum finite range of a mechanism.   
In \citep{myer} strategy-proofness entail a equation that connects payment 
and probability of win. This equation is called revenue equivalence or the envelope condition for the model. The mathematical form in which payment appears in the utility function of a buyer in \citep{myer} is the seller's objective function as well, and thus by substituting the payment obtained from the envelope condition equation the optimization problem can be simplified. 
After the substitution, the proof involves constructing a mechanism that involves maximization of virtual valuation, and the mechanism is indeed strategy-proof and individually rational.

Our technique is different. Since we do not use revenue equivalence equation as a sufficient condition for strategy-proofness we cannot take the route in \citep{myer}. 
Rather we show that the constraint set is compact for any given maximum finite range of a mechanism. Given that the objective function is continuous for finite range, there is a solution. This technique does not depend on utility representations in terms parametric classes. The proof holds for a general class of preferences and probability measures defined on the space of the preferences. The proof technique does not use the virtual valuation as a tool. Our proof of the existence of an optimal mechanism entails an optimization program.    
For more complicated problems ensued from more complicated preferences, substituting the seller's revenue by the payment part of the utility function of the buyer may entail more complicated form of the objective function compared with \citep{myer}. 
For instance let $q$ denote probability of win and the buyer pays only when she wins. Thus the seller considers expected revenue $tq$. But let utility of the buyer be quasilinear, this can happen if the buyer is bounded rational in the sense that 
the buyer does not understand the rule of the auction well. Alternatively, the buyer may value 
gain and loss differently so that the buyer evaluates gain as $\theta q$ but loss by $-t$.       
The integral equation that follows due to strategy-proofness, i.e., the revenue equivalence equation, still holds. However, if we replace $t$ in the objective of the seller by using the revenue equivalence equation the objective function will have $q$  multiplied. This entails a more complicated form of the objective function than the one in \citep{myer}. \citep{Tian} provide further concerns with the general envelop conditions. We have provided different kinds of single-crossing preference domains that have  multidimensional parametric representations, non-quasilinear, and even if quasilinear the the indirect utility or the value function is not a maximum of a family of affine functions. Our approach provides an unified framework to study strategy-proofness and optimal mechanisms for these different domains                              
Since our solution technique guarantees that for finite range the optimal solution exits, even if an analytical solution maybe difficult to obtain approximation may be feasible and then we can simulate on the number of elements in the range. In general, it is possible  that the buyer and the seller evaluate payments differently so that there is a mismatch between what the seller receives and the buyer's utility from the payment. Since our technique does not involve substitution of constraint in the objective function, our approach of finding the optimal mechanism incorporates such complicated situations.                   
                      
\label{remark:myer_vs_single}
\end{remark}

     \section{\bf Analysis of Strategy-proofness when $Rn(F)$ is Countable with Finite Number of Limit Points} 
    \label{sec:count}

     \noindent In this section we extend the construction of strategy-proof mechanisms for $F$ with finite range to $F$s with ranges that contain finitely many limit points.
     We do this exercise to provide an insight about 
     the extent to which we can apply monotonicity of $F$ and continuity of $V^{F}$ to obtain 
     characterizations of strategy-proof mechanisms. First we given an example to demonstrate that if $Rn(F)$ is a continuum, then monotonicity of $F$ and continuity of $V^{F}$ do not imply that $F$ is strategy-proof. 
     
     \begin{example}\rm ({\bf A mechanism with continuum range which is not strategy-proof}) 
     	We consider 
     	a set of preferences given by $\{u(t,q;\theta)=\theta q-t, \theta \in [1,2]\}$. Consider a line segment that connects
     	$(0, 0)$ and $(\frac{1}{3},1)$. The equation of this line is $q = 3t$.
     	Let $F:[1,2]\rightarrow \mathbb{Z}$ be defined as $F(\theta)=(\frac{1}{3}\theta-\frac{1}{3}, \theta-1)$. 
     	Then $Rn(F)=\{(t,q)\mid q=3t, t\in [0,\frac{1}{3}]\}$ is a continuum. The mechanism $F$ is monotone and continuous. Further, $F$ is individually rational and not strategy-proof. 
     	To see that $F$ is individually rational, note that $u(F(\theta);\theta)=\theta(\theta-1)-(\frac{1}{3}\theta-\frac{1}{3})=\theta(\theta-1)-\frac{1}{3}(\theta-1)=(\theta-1)(\theta-\frac{1}{3})\geq 0$ since $\theta\geq 1$.
     	To see that $F$ is not strategy-proof note that the slope of indifference curves of any
     	preference in the domain of $F$ lies in $[ \frac{1}{2},1]$. This slope is smaller than the slope of the line segment $q=3t$. Thus the maximizations of utility for all preferences in the domain occur at $(t,q)=(\frac{1}{3},1)$. Now $V^{F}(\theta)=u(F(\theta);\theta)=(\theta-1)(\theta-\frac{1}{3})$ is continuous. Further, $\frac{dV(\theta)}{d\theta}=(\theta -\frac{1}{3})+(\theta- 1)\neq (\theta-1)=q(\theta)$. Also $\frac{d^{2}V(\theta)}{d\theta^{2}}=2>0$. 
     	Thus the indirect utility function is convex. 
     	The crucial point here is that the derivative of the indirect utility function
     	is not equal to the allocation probability. Therefore this example does not satisfy the integral condition discussed in \citep{myer}. Our analysis shows that situations in which the range of a mechanism is finite we do not need to assume the integral condition to characterize strategy-proof mechanisms.         
     	\label{ex:not_sp}
     \end{example}

     \noindent The continuum is an extreme of finiteness, and countable sets are in between these extremes. 
     Consider $Rn(F)$ which is closed and all bundles are limit points. Such sets are called perfect, and it is known that such sets are uncountable. Thus countable $Rn(F)$ which is closed and all bundles are limit points is not well defined. 
     However if $Rn(F)$ is countable and closed and the number of limit points is finite, then the set of strategy-proof mechanisms is not empty. Example \ref{example:finite_lm_pt} demonstrates this fact.

     \begin{example}\rm Consider $\{u(t,q;\theta)=\theta\sqrt{q}-t,\theta>0\}$. 
     	Consider $D=\{(t,q)\mid q=3t, t\in [\frac{1}{12},\frac{1}{3}]\}$. 
     	We note that $\arg\max_{(t,q)\in D} \frac{2}{3}\sqrt{q}-t=(\frac{1}{3},1)$. 
     	Also $\arg\max_{(t,q)\in D} \frac{1}{3}\sqrt{q}-t=(\frac{1}{12},\frac{1}{4})$.
     	Consider a sequence $\{\theta^{n}=\frac{2}{3}-\frac{1}{n}\}_{n=3}^{\infty}$. 
     	Then $\arg\max_{(t,q)\in D}\theta^{n}\sqrt{q}-t=( \frac{1}{3}((\frac{3}{2}(\frac{2}{3}-\frac{1}{n})))^{2}, ((\frac{3}{2}(\frac{2}{3}-\frac{1}{n})))^{2})$. 
     	By richness consider  $\theta$ such that $u((\frac{1}{12},\frac{1}{4});\theta)=u( \arg\max_{(t,q)\in D}\theta^{4}\sqrt{q}-t;\theta)$. Call the $\theta$ that solves this equality  $\theta^{*4}$.  
     	Set $F(\theta)=(\frac{1}{12}, \frac{1}{4})$ for all $\theta<\theta^{*4}$. 
     	Then consider $\theta^{*5}$ such that $u(\arg\max_{(t,q)\in D}\theta^{4}\sqrt{q}-t,\theta^{*5})=u(\arg\max_{(t,q)\in D}\theta^{5}\sqrt{q}-t,\theta^{*5})$. 
     	Set $F(\theta)=\arg\max_{(t,q)\in D}\theta^{4}\sqrt{q}-t$ for all $\theta \in [\theta^{*4},\theta^{*5}[$. 
     	In general set $F(\theta)=\arg\max_{(t,q)\in D}\theta^{n}\sqrt{q}-t$ for all $\theta \in [\theta^{*n},\theta^{*(n+1)}[, n\geq 4$. Also set $F(\theta)=(\frac{1}{3},1)$ for $\theta\geq \frac{2}{3}$.  
     	Since indifference curves are strictly convex, $F$ is a strategy-proof mechanism.{\footnote{This mechanism is not individually rational, but this is not a concern since we are discussing strategy-proofness.}}
     	This mechanism has a countable range, it is closed and the only limit point is the bundle $(\frac{1}{3},1)=\lim_{n\rightarrow\infty}\arg\max_{(t,q)\in D}\theta^{n}\sqrt{q}-t$. 
     	Figure $2$ is a pictorial depiction of this mechanism.

     	\begin{center}
     		\includegraphics[height=6.6cm, width=12cm]{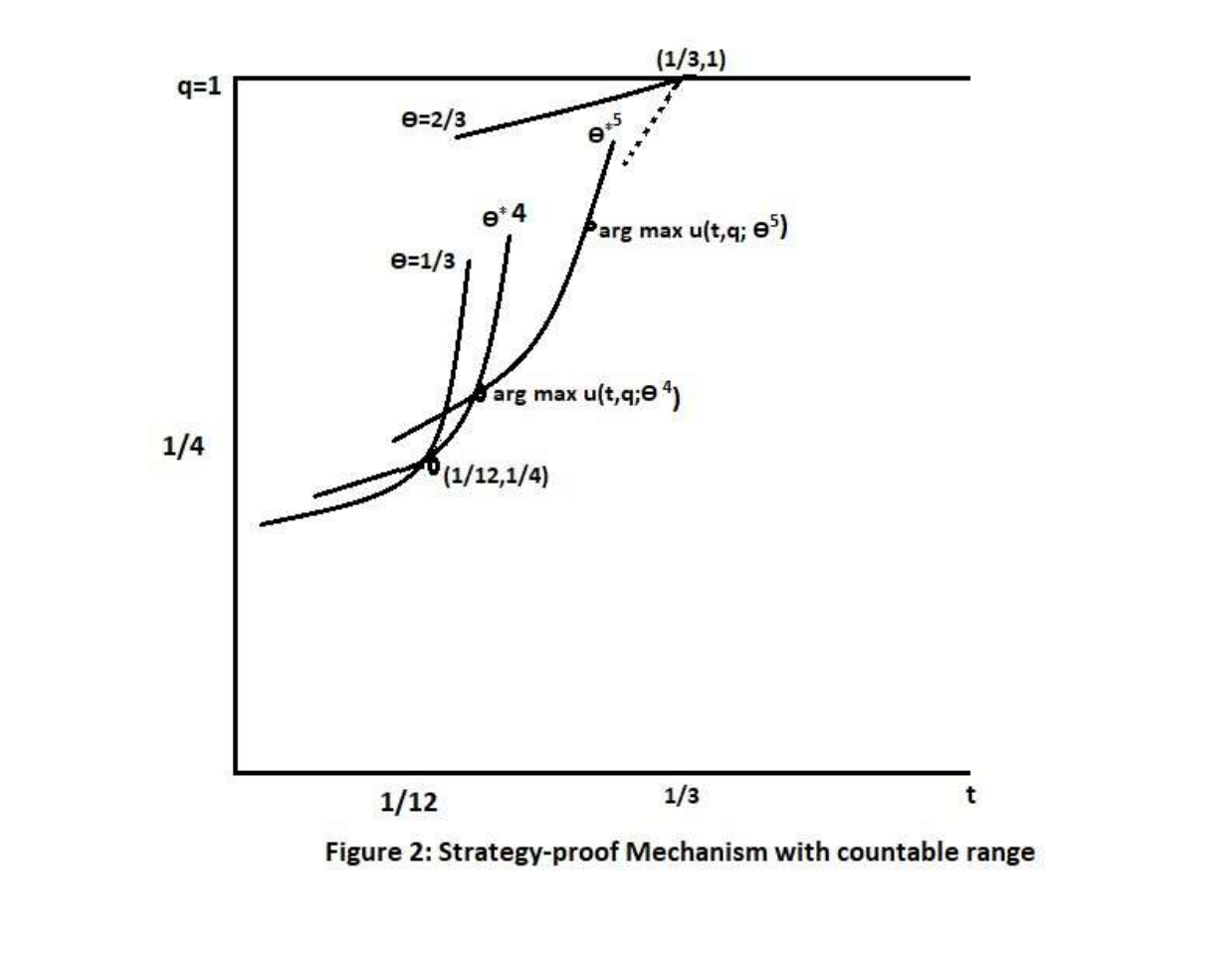}
     	\end{center}

     	\label{example:finite_lm_pt}
     \end{example}

     \noindent  We generalize this example in the following theorem.

     \begin{theorem}\rm Let $F$ be monotone and $V^{F}$ continuous. Let $Rn(F)$ be countable, closed and has finitely many limit points. 
     	Then $F$ is strategy-proof.    
     	\label{thm:countable}
     \end{theorem}	
     
     \begin{proof} Since the number of limit points is finite, around each small neighborhood of each limit point we can use Theorem \ref{thm:implies_strtagey_proof_finite_range} for the bundles other than the limit point. Then strategy-proofness holds in the limit. The details of the proof are in the Appendix. 
     \end{proof}

     \noindent An immediate implication of Lemma \ref{lemma:limit_point_supremum} is that 
     if $F(R')$ is a limit of $Rn(F)$ and is approached by both
     increasing and decreasing sequences of bundles from $Rn(F)$, then
     $\{R\mid F(R)=F(R')\}$ is a closed interval in $\mathcal{R}^{rsc}$. 
     The next section explores implications of the single-crossing property for optimal mechanisms.

     \subsubsection{Optimal Mechanisms for Countable Range with Finitely Many Limit Points}   
     
     In this section we make a remark about optimal mechanisms that have countable range. In particular, we try to understand optimal mechanisms for the countable case from the perspective of its relationship with mechanisms that have finite range.
     For $[\underline{R},\overline{R}]\subseteq \mathcal{R}^{rsc}$, let $F^{c}:[\underline{R},\overline{R}]\rightarrow \mathbb{Z}$ be strategy-proof and individually rational.
     Let $Rn(F^c)$ be countable, closed and has finitely many limit points. For the sake of simplicity let 
     there be just one limit point. Let $(t^*,q^*)\in Rn(F^{c})$ be the limit point. 
     Let $Rn(F^{c})=\{(t^n,q^n)|n=0,1,2,\ldots\}\cup\{(t^*,q^*)\}\cup \{(t^k,q^k)|k=1,2,\ldots\}$ such that $(t^{n},q^n)<(t^{n+1},q^{n+1})$ and $(t^{k+1},q^{k+1})<(t^{k},q^k)$. Further, let $(t^n,q^n)\uparrow (t^*,q^*)$ and  
     $(t^k,q^k)\downarrow (t^*,q^*)$. Since we assume the limit point can approached by
     both increasing and decreasing sequences the assumption of one limit point is without loss of generality from the perspective of Proposition \ref{prop:opt_count}. 
     Set $(t^0,q^0)=(0,0)$. Let $\underline{R}=R^{0}$ and $F^{c}(R^0)=(0,0)$. 
     Let $0<t^1, 0<q^1$. Now let for $n=1$, $R^1$ be the special preference such that
     $(t^0,q^0)I^1(t^1,q^1)$ so that $(t^1,q^1)=F^{c}(R)$ for $R\in ]R^1,R^2[$. 
     In general $(t^n,q^n)=F^{c}(R)$ for $R\in ]R^{n},R^{n+1}[$. Since $F^c$ in $[R^{0},R^{n}]$ is finite, by Theorem \ref{thm:implies_strtagey_proof_finite_range} $F^{c}$ as defined in $[R^{0},R^{n}[$ is well defined.            
     Further for $k=1$, let $R^1=\overline{R}$. Then $(t^2,q^2)I^2 (t^1,q^1)$. 
     Let for $k=1$, $F^c(R)=(t^1,q^1)$ for $R\in ]R^2,R^1[$. In general 
     set $F^{c}(R)=(t^n,q^n)$ for $R\in ]R^{n+1},R^n[$. Since $\mu$ is continuous 
     the expected revenue from the mechanism is given by

     \begin{equation}
     \sum_{n=1}^{\infty}t^n\mu([R^{n},R^{n+1}])+t^{*}\mu([R',R''])+\sum_{k=1}^{\infty} t^{k}\mu([R^{k+1},R^{k}]).
     \label{eqn:opt_count}
     \end{equation}
     
     \noindent Since both infinite sums are sums of non-negative numbers the partial sums are increasing and by individual rationality they are bounded by $\overline{T}$, where $(\overline{T},1)\overline{I}(0,0)$. For example $\sum_{n=1}^{l}t^n\mu([R^{n},R^{n+1}])\leq \sum_{n=1}^{l}\overline{T}\mu([R^{n},R^{n+1}])\leq \overline{T}$, the last inequality follows since $\mu$ is a probability measure. Thus both sums converge. 
     It may happen that the expected revenues from the mechanisms with finite ranges are dominated by a mechanism with countable range.  
     For such a scenario next we show that for every $\epsilon>0 $, there is a mechanism with finite range so that the revenue from the mechanism with countable range is larger than the mechanism with finite range by at most $\epsilon$. 
     
     \begin{prop}\rm Let $([\underline{R},\overline{R}], \mathcal{B}, \mu)$ be given. Let $F^{c}:[\underline{R},\overline{R}]\rightarrow \mathbb{Z}$ be strategy-proof, individually rational, with countable closed range and has finitely many limit points.
     	Let for every mechanism $F$ with finite range $E(F)<E(F^{c})$. 	
     	Then for every $\epsilon>0$, there is a mechanism $F^{\epsilon}$ with finite range such that $E(F^{c})-E(F^{\epsilon})\leq \epsilon$. Further if an optimal finite  mechanism $F^{*}$ exists, then $E(F^{c})\leq E(F^*)$. 
     	\label{prop:opt_count}
     \end{prop}

     \begin{proof} See the Appendix.

     \end{proof}	 
     
     \noindent If we wish to study optimal mechanisms within a framework that allows only monotonicity of $F$ and continuity of $V^{F}$, then Example \ref{ex:not_sp} and Proposition \ref{prop:opt_count} imply that without loss we can consider mechanisms with finite ranges.

\section{\bf Single-crossing Property for Restricted Classical Preferences}         
\label{sec:non_mon}

In our discussions so far we have considered preferences that are strictly monotone in $t$. However, it can happen that if an object is not allocated to an agent, then she does not differentiate between low and high payment. Such situations create non-monotonicity in the preferences. As an example consider an indivisible object, and let $q$ denote the probability of winning that object. Let 
$\theta$ denote the valuation for the object. Consider the buyer whose type or the valuation is $\theta$.
Consider the lottery in which with probability $q$ the consequence for the buyer is $\theta-t$ and with probability  $1-q$ the consequence is $0$. Consider two kinds of expected pay-offs
$(i)$ {\bf risk neutral} $u(t,q ;\theta)=q(\theta -t)=\theta q-qt=\text{expected value-expected payment}, \theta \in [\underline{\theta}, \overline{\theta}]$, 
$0<\underline{\theta}<\overline{\theta}<\infty$. 
$(ii)$ {\bf risk averse} $u(t,q ;\theta)=q\sqrt{\theta -t}, \theta \in [\underline{\theta}, \overline{\theta}]$, 
$0<\underline{\theta}<\overline{\theta}<\infty$, $t\leq \theta$. If $q>0$ and $t>\theta$, then $q\sqrt{\theta -t}$ is not a real number.  
In the risk averse model, for every $\theta$ the lottery over payments takes into account the final wealth position
$\theta-t$. The insights from the risk neutral and risk-averse models can be put in a more general context.
\citep{Seriwey} require classical preferences to be monotone when both arguments of a pay-off function are positive. The analogous requirement for our model is that    
a preference $R$ satisfies money and $q$ monotonicity when $q>0,t>0$. However, this is not enough since 
as demonstrated by the risk averse model for all positive bundles real utility representations may not be feasible. A straightforward way to address these issues is to consider the following restricted classical class of preferences.   

\begin{defn} \rm ({\bf Restricted Classical Preferences}) The complete, transitive preference relation $R$
	defined for $[0,t_R]\times [0,1]$ where $t_{R}$ is a unique payment bound for $R$ and $0<t_R < \infty$ is called restricted classical if 
	\begin{itemize}
		\item {\bf money-monotone: } for all $q\in ]0,1]$, if $t_R\geq t''>t'\geq 0$, then $(t',q)P (t'',q)$
		\item {\bf q-monotone:} for all $t$ with  $t_R > t\geq 0$ if $1\geq q''>q'\geq 0$, then $(t,q'')P(t,q')$
		\item {\bf 0-equivalence:} $(0,0)I (t_R,q)$ for all $q\in [0,1]$  and for all $t\in[0,t_R[$, $(0,0)I (t,0)$.  
		\item {\bf continuous:} $R$ is continuous on $[0,t_R]\times [0,1]$
	\end{itemize}   
	
	\label{defn:gen_classical}
\end{defn}             

\noindent Thus risk-neutral and risk averse preferences are restricted classical preferences with the payment bound $t_\theta=\theta$. Negative and complex number pay-offs are ruled out by restricted classical preferences. The single-crossing property needs to be modified as well. 

\begin{defn}\rm ({\bf Single-crossing of two Restricted Classical Preferences} )
	Consider two restricted classical preferences $R',R''$, $t_{R'}\neq t_{R''}$. Consider $(t,q)$ where $q>0$, and $t<\min \{t_{R'},t_{R''} \}$. We say $R'$ and $R''$ exhibit the single-crossing property
	if and only if $\{(t,q)\}=IC(R',(t,q))\cap IC(R'',(t,q))$. Given that $R'$ and $R''$ exhibit the single-crossing property, we say {\bf $R''$ cuts $R'$ from above} if 
	$\square (t,q) \cap UC(R'',(t,q))\subseteq \square (t,q) \cap UC(R',(t,q))$.  
	\label{defn:single_crosiing_restricted}
\end{defn} 

\noindent For $(t,q)$ with $q>0$ and $t<\min \{t_{R'},t_{R''}\}$ the classical properties ensures that 
$R''$ cuts $R'$ from above is well defined. That is, for these specific bundles if $R''$ cuts $R'$ from above at one of these bundles, then $R''$ cuts $R'$ at all other such bundles from above. We define rich single crossing domain for restricted classical preferences exactly as before.

\begin{defn}[Rich Single-crossing domain for Restricted Classical Preferences]\rm We call a subset of the set of restricted classical preferences {\bf restricted single-crossing domain} if any $R',R''$ that belongs to the subset 
	satisfy the single-crossing property for restricted classical preferences. We call a single crossing domain for restricted classical preferences {\bf rich} if for any two bundles $x'=(t',q'), x''=(t'',q'')$ such that $t'<t'',q'<q''$ there is $R$ in the restricted single crossing domain such that $x'Ix''$. We denote a rich single crossing domain by $\mathcal{R}^{rrc}$.        	
\end{defn}

 \noindent The domain $\{R_{\theta}=q\sqrt{\theta-t}|\theta>0, t_{R_{\theta}}=\theta\}$ is a rich restricted single-crossing  domain. We can pin down the geometry of strategy-proof mechanisms for $\mathcal{R}^{rrc}$ as well in the same fashion as before. 
The domain $\{R_{\theta}=q\sqrt{\theta-t}|\theta \in ]0,2], t_{R_{\theta}}=\theta\}\cup 
\{R_{\alpha}=q\sqrt{2-\alpha t}|\alpha \in ]0,1], t_{R_{\alpha}}=\frac{2}{\alpha}\}$ is also a domain of rich single-crossing domain of restricted classical preferences. This domain is of multidimensional types. Now consider the following lemma.

\begin{lemma}\rm Let, $R',R'' \in \mathcal{R}^{rrc}$.  If $R''$ cuts $R'$ from above, then $t_{R'}<t_{R''}$.   
	\label{lemma:order_payment_bound}
\end{lemma} 

\begin{proof} See the Appendix.

\end{proof}

\noindent The following is another observation about $\mathcal{R}^{rrc}$.

\begin{lemma}\rm Consider $\mathcal{R}^{rrc}$. For every $t>0$, there is $R\in \mathcal{R}^{rrc}$ such that $t=t_{R}$.  
\end{lemma}

\begin{proof} By richness of $\mathcal{R}^{rrc}$, there is $R$ such that $(t,1)I(0,0)$. By $0-$ equivalence $t=t_{R}$. 	
\end{proof}	

\noindent The order topology on $\mathcal{R}^{rrc}$ is given naturally by the order on $t_R$: 
for distinct $R',R'' \in \mathcal{R}^{rrc}$, $R'\prec R''$ if and only if $t_{R'}<t_{R''}$. Thus in this order topology $\mathcal{R}^{rrc}$ is homeomorphic, with order preserving bijection,  to an open interval in $\Re$ hence $\mathcal{R}^{rrc}$ is metrizable.
We note that if $R'\prec R''$, then for all $(t,q)$ such that $t<t_{R'}$ and $0<q$, the classical properties ensure that indifference curves of $R''$ cut the indifference curves of $R'$ from above. 
Thus we can define the order on $\mathcal{R}^{rrc}$ as before.       
Since restricted classical preferences are not defined everywhere in $\mathbb{Z}$ we need to modify the definition of strategy-proofness as well. 

\begin{defn}\rm ({\bf Restricted Strategy-proof mechanism}) The function $F:\mathcal{R}^{rrc}\rightarrow \mathbb{Z}$ is a {\bf mechanism} if for $R\in \mathcal{R}^{rrc}$, $F(R)\in [0,t_R]\times [0,1]$. 
Then $F$ is {\bf restricted strategy-proof} if for all $R,R'$, $F(R')\in [0,t_R]\times [0,1] $, $F(R)RF(R')$.	   
\end{defn}

\noindent We have seen that without restricting allocations to $[0,t_R]\times [0,1]$ we may encounter situations where allocations cannot be compared as demonstrated by the risk averse preferences. 
By definition restricted strategy-proofness requires $F(R')$ to be comparable with $F(R')$.   
Next we can define continuity of $V^F$.

\begin{defn}\rm We call
	$V^{F}:\mathcal{R}^{rrc}\rightrightarrows \mathbb{Z}$ {\bf continuous}, if for any $R$ and any monotone sequence
	$\{R^{n}\}_{n=1}^{\infty}$ converging to $R$: (a) the sequence $\{F(R^n)\}_{n=1}^{\infty}$ 
	converges to $z=(t_z,q_z)$ and $t_z\leq t_R$, (b)  $z=\lim_{n\rightarrow \infty}F(R^n)I F(R)$.
\label{defn:cont_restricted}
\end{defn}

\noindent The next proposition provides a characterization of strategy-proof and individually rational mechanism. A mechanism is individually rational if $F(R)R(0,0)$. Note that mechanisms defined in $\mathcal{R}^{rrc}$ are individually rational.

\begin{prop} \rm Let $F:[\underline{R},\overline{R}]\rightarrow \mathbb{Z}$ be a mechanism and $Rn(F)$ be finite, $[\underline{R},\overline{R}]\subseteq \mathcal{R}^{rrc}$. Also assume that if $F(R)I(0,0)$, then $F(R)=(0,0)$. $F$ is strategy-proof, if and only  if $F$ is monotone and $V^{F}$ is continuous.

	\label{prop:rest_sp}	
\end{prop}

\begin{proof} See the Appendix. 

\end{proof}

\subsection{\bf Optimization Programs for Optimal Mechanisms in the Risk neutral and the Risk Averse Model}

\citep{myer} analyzes the risk-neutral model, thus we refer the model with risk-neutral preferences as {\bf Myerson model}. We show that the optimal mechanism in the Myerson model is deterministic if $\mu$ exhibits monotone hazard rate. The cumulative distribution function of $\theta \in [\underline{\theta},\overline{\theta}]$ is given by $\Gamma$ and the density by $\gamma$. 
In this model the expected revenue from the buyer with valuation $\theta$ is given by $q(\theta)t(\theta)$. A mechanism is a function $F:[\underline{\theta},\overline{\theta}]\rightarrow [0,\overline{T}]\times [0,1]$, $\overline{T}=\overline{\theta}$ so that $1 \overline{\theta} - 1 \overline{T}=0$. We consider $F$ with finite  $Rn(F)$.

\begin{center}
	\includegraphics[height=7cm, width=13cm]{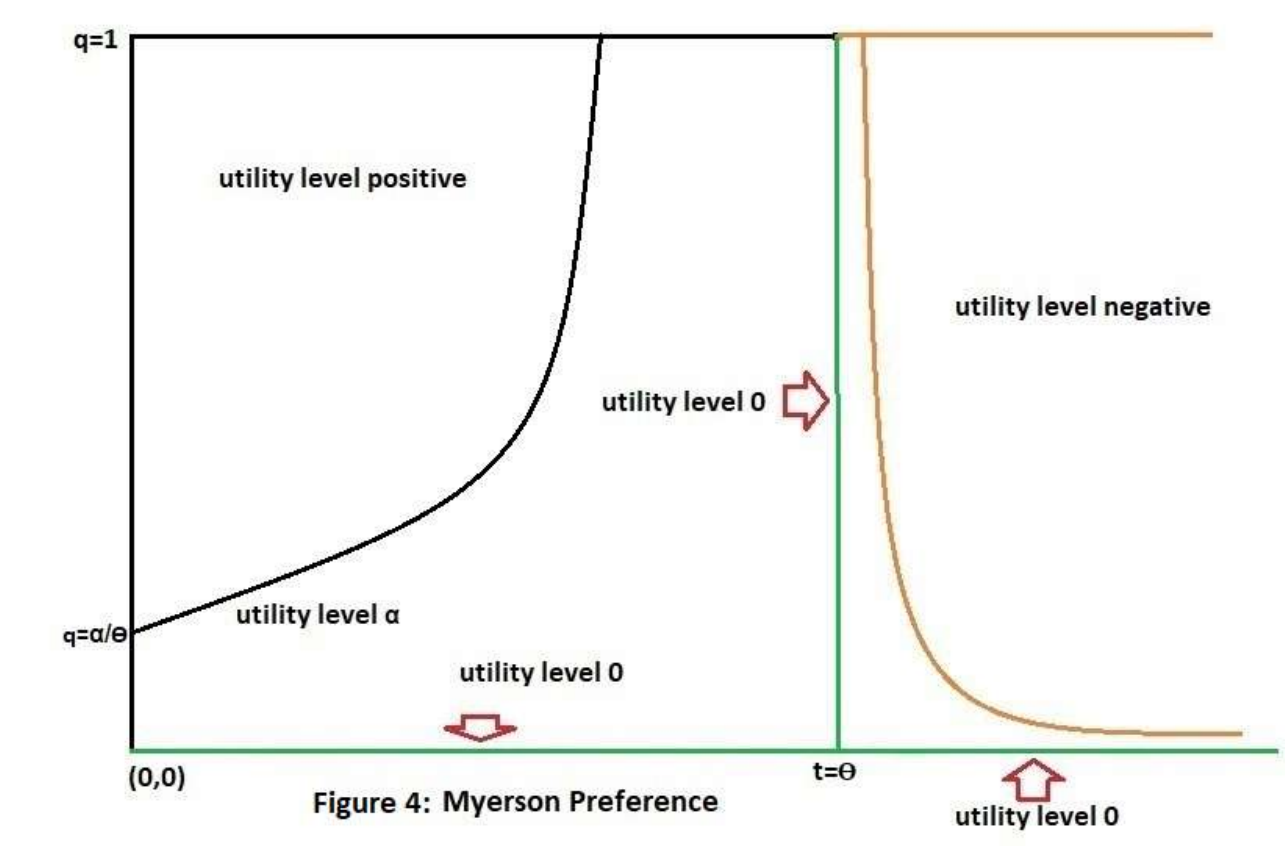}
\end{center}

\noindent For $(t,q)$ with $t<\theta, 0<q$ the expected pay-off is positive.
The indifference curve for utility level $\alpha$ drawn in Figure $3$ is positive. 
For $(t,q)$ with either $t=\theta$ or $q=0$, the utility level is zero.
Thus the indifference curve for utility level $0$ is of inverted $T$ shape.
For $(t,q)$ with $\theta<t, 0<q$ the level of utility is negative.      
More formally for the level of utility $\alpha>0$ 
consider the indifference curve given by the equation $\theta q-qt=\alpha$. 
This entails $q=\frac{\alpha}{\theta-t}$. 
Then $0<\frac{dq}{dt}=\frac{\alpha}{(\theta-t)^{2}}$ and  
$0<\frac{d^{2}q}{dt^{2}}=\frac{2\alpha}{(\theta-t)^{3}}$ if $t<\theta$. 
For $\alpha<0$, the indifference curves are decreasing and convex since if $\alpha<0$, then $\theta<t$.   
These preferences 
do not satisfy money monotonicity if $q=0$. However for  $(t,q)$ with $t<\theta, 0<q$ the $\theta q-qt$ satisfies all the properties of classical preferences.  
By individual rationality $F(\theta)=(t,q)$ such that $\theta q-q t\geq 0$. The next lemma shows that 
for positive levels of utilities of two preferences $\theta'$ and $\theta''$ the single-crossing property is satisfied. 

\begin{lemma}\rm Suppose $\theta'<\theta''$. Let the utility levels represented by the indifference curves $IC(\theta',(t,q))$ and $IC(\theta'',(t,q))$ be positive. Then $IC(\theta'',(q,t))$ and $IC(\theta',(q,t))$ intersects at most  once.   
	\label{lemma:sing_myer}
\end{lemma}     

\begin{proof} Consider $\frac{dq}{dt}=\frac{\alpha'}{(\theta'-t)^{2}}=\frac{\theta'q-qt}{(\theta'-t)^{2}}=\frac{q}{(\theta'-t)}$. 
	Thus if $\theta'<\theta''$, then the slopes of the indifference curves of $\theta''$ are smaller than the ones for $\theta'$. Thus the proof of the lemma follows.	
\end{proof}	

\noindent As a corollary of Proposition \ref{prop:rest_sp} this result a mechanism $F:[\underline{\theta},\overline{\theta}]\rightarrow \mathbb{Z}$ which is also individually rational is monotone and $V^{F}$ continuous.  Let $I^{\theta}$ denote the indifference relation for the preference
$q\theta-qt$.   

\begin{prop} \rm Let $F:[\underline{\theta},\overline{\theta}]\rightarrow \mathbb{Z}$ be a mechanism and $Rn(F)$ be finite and individually rational. Also assume that if $F(\theta)I^{\theta}(0,0)$, then $F(\theta)=(0,0)$. Then, $F$ is strategy-proof, if and only  if $F$ is monotone and $V^{F}$ is continuous.
	\label{prop:myer_sp}	
\end{prop}  

\noindent Let $\theta^{*}=\inf\{\theta\mid q(\theta)>0, \theta\in [\underline{\theta},\overline{\theta}]\}$.

\noindent Now we proceed to study optimal mechanisms. 

\begin{prop}\rm Let $\mathcal{F}=\{F\mid F:[\underline{\theta},\overline{\theta}]\rightarrow [0,\overline{T}]\times [0,1], ~\text{strategy-proof, indvidually rational,} \\\text{has finite}~Rn(F)\}$.
	Also assume that if $F(\theta)I^{\theta}(0,0)$, then $F(\theta)=(0,0)$.
	Let $\frac{\gamma(\theta')}{1-\Gamma(\theta')}\leq \frac{\gamma(\theta'')}{1-\Gamma(\theta'')}$ if $\theta'<\theta''$. Then the $F^{*}$ exists and is deterministic.        	
	\label{prop:myer_opt}   
\end{prop}

\begin{proof} We provide a pictorial depiction of the optimal mechanism.  Let $\theta^{*}=\inf\{\theta\mid q(\theta)>0\}$.  
	
	\begin{center}
		\includegraphics[height=7cm, width=12cm]{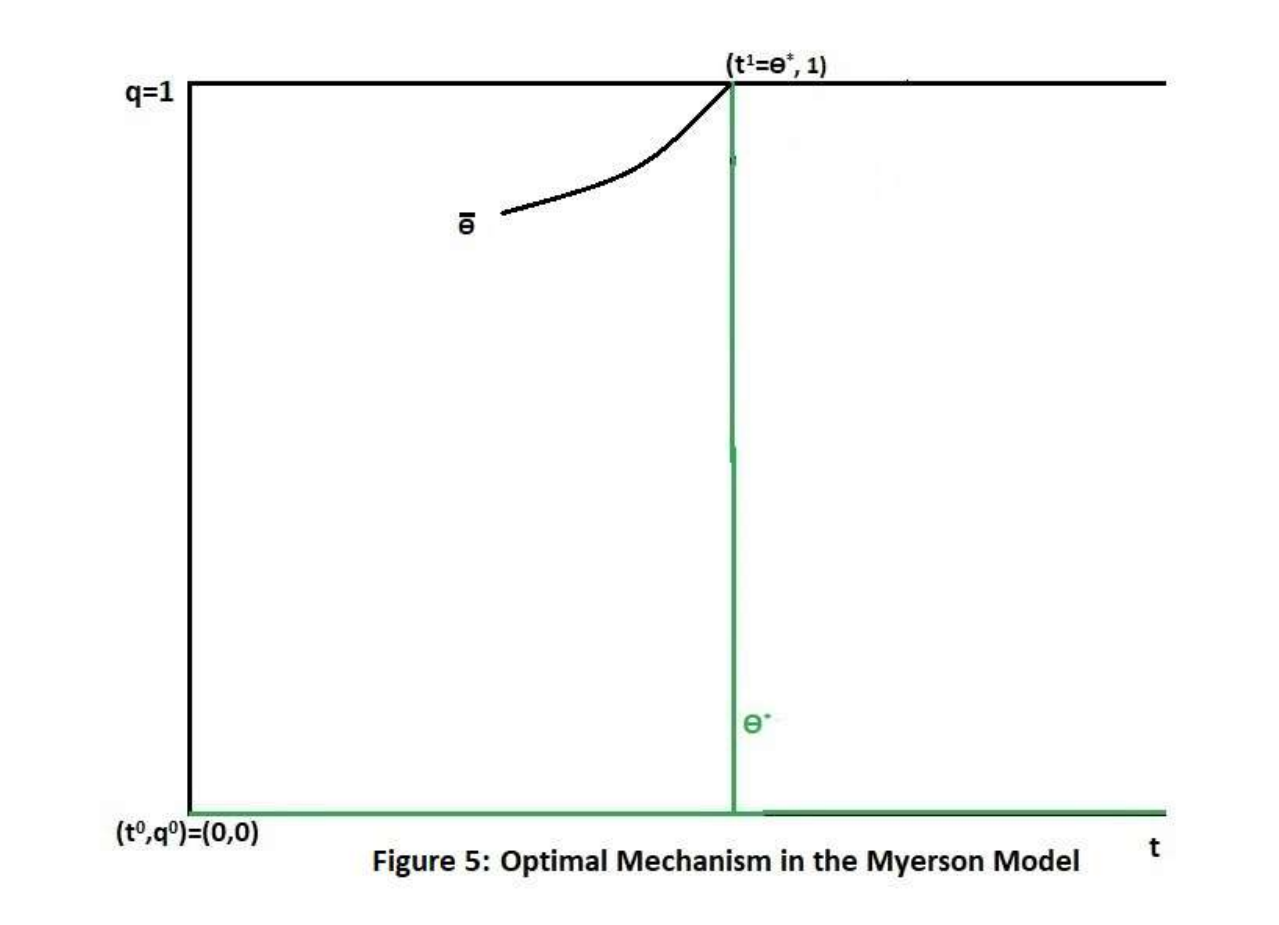}
	\end{center}
	
	\noindent $F(\theta)=(0,0)$ if $\theta \in [\underline{\theta},\theta^{*}], F(\theta)=(t^1=\theta^{*},1)$ if $\theta \in ]\theta^{*},\overline{\theta}]$. The optimal mechanism has two bundles in its range. The proof is in the Appendix.

\end{proof}

\begin{remark}\rm The optimal mechanism can be extended to the $n-$buyer model in the manner analogous to the model for quasilinear preferences. The form of the optimization problem for the risk averse model is exactly the same. The difference from the Myerson model is that the constraints are non linear for the model with risk averse preferences.   
	The optimization program has continuous objective function and compact constraint set for range with maximum $l$ distinct bundles. 
	The trick of extending preferences in $\mathcal{R}^{rrc}$ is not required here, because $t_{\theta}=\theta$ appears as a parameter in the pay-off functions and thus if for all $n$, $q^{'n}\sqrt{\theta^{'n}-t^{'n}}=q^{''n}\sqrt{\theta^{''n}-t^{''n}}$, $\theta^{'n}-t^{'n}\geq 0, \theta^{''n}-t^{''n}\geq 0$,
	then  $\lim_{n\rightarrow \infty}q^{'n}\sqrt{\theta^{'n}-t^{'n}}=\lim_{n\rightarrow \infty}q^{''n}\sqrt{\theta^{''n}-t^{''n}}$.      
	The exact optimal mechanism depends on the probability distribution of $\theta$ and thus finding the optimal mechanism is a computational problem.

\end{remark}

\noindent

\section{\bf Concluding Remarks}   
\label{sec:con}

This paper considers the implication of restricting domains of mechanisms so that 
the procedure of restricting domains satisfies a single-crossing property.
We study classical and restricted classical single-crossing domains. 
We provide a characterization of strategy-proof mechanisms that are monotone, and the indirect preference correspondences that correspond to the mechanisms are continuous.  
Our characterization brings out the geometry of strategy-proof mechanisms. This geometry 
provides a natural way to compute the optimal mechanism. 
As applications of our framework we compute optimal mechanisms for important classes of preferences that are used in the literature. The optimization program considers only a finite number of preferences in each iteration while looking for an optimal mechanism, i.e., it requires to consider only a finite number incentive constraints. 
This happens due to the single-crossing property, i.e., the intersection of indifference sets of two different preferences is at most once, and the pattern of intersections or the cutting as defined in the text is the same everywhere in the consumption space due to the classical properties of the preferences.  These simple observations make the optimization program computationally tractable. Given that computation on topological spaces is a highly active area currently, thanks to the availability of high computational power, our approach to optimal mechanism design maybe considered useful.

\[\textbf{Appendix}\]

\noindent{\bf Proof of Theorem \ref{thm:lin_cont}:} We construct an order preserving bijection between $\mathcal{R}^{rsc}$ with order $\prec$ and an open interval in the real line $\Re$. Consider $z=(0,0)$ and let $C_{\delta}(0,0)$ be the circle of radius $\delta<1$ with origin $(0,0)$. Consider $\mathbb{A}=[C_{\delta}(0,0)\cap \mathbb{Z}]\setminus \{(\delta,0)\cup (0,\delta)\}$. 
That is, $\mathbb{A}$ is the quarter circle with radius delta and center $(0,0)$ intersected with $\mathbb{Z}$ excluding the end points. 
Let $h':\mathcal{R}^{rsc}\rightarrow \mathbb{A}$ by $h'(R)=(t,q)$ where $(t, q)I(0,0)$. 
By richness of $\mathcal{R}^{rsc}$, for every $z\in \mathbb{A}$ there exists $R\in \mathcal{R}^{rsc}$ such that $zI(0,0)$. Thus, $h'$ is onto. By the single-crossing property $h'$ is one-one.

Further, $\mathbb{A}$ is homeomorphic to the open interval $]0,\delta[$ in $\Re$, where the homeomorphism is just the projection $\pi:\mathbb{A}\rightarrow ]0,\delta[$ defined by $\pi(t,z)=t$. Then we obtain an order preserving homeomorphism between 
$\mathcal{R}^{rsc}$ and $]0,\delta[$ defined by the composition $\pi~ o~ h'$, call the homeomorphism $h$. Since $]0,\delta[$ has the least upper bound property so does $\mathcal{R}^{rsc}$. Also $\mathcal{R}^{rsc}$
has the second property of a linear continuum by richness.  
For any $R',R''\in \mathcal{R}^{rsc}$, define $d(R',R'')= |h(R')-h(R'')|$ where $|\cdot|$ is the standard Euclidean metric on $\Re$. 

\medskip

\noindent{\bf Proof of Lemma \ref{lemma:closure}:} 	Let $R_{0}=\inf_{\prec} U$. Suppose there exists $R\neq R_{0}$ such that $]R_{0},R[\cap U =\emptyset$. Since $R_{0}$ is the infimum, $R_{0}\in U$. Thus $R_{0}\in cl(U)$. If for every $R\neq R_{0}$, $]R_{0},R[\cap U \neq \emptyset$, then $R_{0}$ is a limit point of $U$ and thus lies in $cl(U)$.  
The argument for the supremum is analogous.   	   

\medskip

\noindent{\bf Proof of Lemma \ref{lemma:preference_conv}:} 	The following claim is required to prove the lemma. Consider a compact subspace $K=[0, \bar{t}]\times [0,1]$
whee $0<\bar{t}<\infty$.

\begin{claim}\rm Each preference in $\mathcal{R}^{rsc}$ can be represented by a continuous function $f$ such that if $R^{n}\rightarrow R$, then $f^{n}\rightarrow f$ uniformly in $K$.  
	\label{claim:uniform}	   
\end{claim}

\noindent{\bf Proof of the Claim \ref{claim:uniform}:} 
Any preference $R\in \mathcal{R}^{rsc}$ is completely defined by the equivalence
classes formed by its indifference sets. Since $R\in \mathcal{R}^{rsc}$, by the monotonicity each equivalent class can be identified uniquely with a bundle $(t, 1)$, where $t\in [0,\infty[$.
An equivalence class of $R$ denoted by $[t]_{R}$ is: 
$[t]_{R}=\{(t',q')\mid (t',q') I (t,1), (t',q')\in \mathbb{Z}\}$.
Let $f_{R}(t',q')=t$ for all $(t', q')\in  [t]_R$. 
That is for $(t',q'),(t'',q'')\in \mathbb{Z}$ we have $(t',q')R(t'',q'')$ if and only if $f_{R}(t',q')\leq f_{R}(t'',q'')$.{\footnote{We do not require $f_{R}$ to be a utility representation.}}

This function is continuous. 
To see this let $(t^{'n},q^{'n})\rightarrow (t',q')$. 
We want to show that $\lim_{n\rightarrow \infty}f_{R}(t^{'n},q^{'n})=f_{R}(t',q')$. 
Let $(t^{'n},q^{'n})I(t^{n},1)$ and thus $f(t^{'n},q^{'n})=t^{n}$. 
Further, let $(t',q')I(t,1)$ and thus $f_{R}(t',q')=t$.  
Let by way of contradiction $t^{n}$ does not converge.
Consider the situation where $\{t^{n}\}_{n=1}^{\infty}$ is an unbounded sequence. 
Then for every $M>0$ there is $t^{n}>M$. 
Since $(t',q')I (t,1)$, by monotonicity of $R$  
there is $(t^{n},1)$, $n$ large enough so that $(t',q')P(t^{n},1)$.  
Since $(t^{'n},q^{'n})\rightarrow (t',q')$ and $R$ is continuous there is $N$ such that 
for all $m\geq N$, $(t^{m},q^{m})P(t^{n},1)$. 
Since $\{t^{n}\}_{n=1}^{\infty}$ is  unbounded there is $m>N$ such that $t^{n}<t^{m}$. 
By monotonicity of 
$R$, $(t^n,1)P(t^m,1)$ and thus by transitivity of $R$ we obtain $(t^{m},q^{m})P(t^{m},1)$, this contradicts the definition of $f_{R}$.

Thus  let $\{t^{n}\}_{n=1}^{\infty}$ be bounded and does not converge. 
Let $\{t^{n_k}\}_{k=1}^{\infty}$ be a convergent subsequence of $\{t^{n}\}_{n=1}^{\infty}$. We show that the subsequence converges to $t$. Let by way of contradiction 
$t^{n_k}\rightarrow t^{*}\neq t$. Without loss of generality let $t^{*}<t$. Consider
$[t^{*}-\epsilon, t^{*}+\epsilon]~\cap~ [t-\epsilon, t+\epsilon]=\emptyset$. 
For $K$ large enough for all $k\geq K$, $t^{n_k}\in ]t^{*}-\epsilon, t^{*}+\epsilon[$.
Now $(t^{'n_k},q^{'n_k})I (t^{n_{k}},1)$. By monotonicity of $R$, $(t^{'n_k},q^{'n_k})\in [0, t^{*}+\epsilon]\times [0,1]$. Since $\{(t^{'n},q^{'n})\}_{n=1}^{\infty}$, and thus $\{(t^{'n_k},q^{'n_k})\}_{k=1}^{\infty}$ converges to $(t',q')$ we have  $(t',q')\in [0, t^{*}+\epsilon]\times [0,1]$.  
By monotonicity of $R$, $(t-\epsilon,1)P (t,1)I(t',q')$. By continuity of $R$ there is $B(\epsilon, (t',q'))$ such that $(t-\epsilon,1)Pz$ for all $z\in B(\epsilon, (t',q'))$.  Here $B(\epsilon,(t',q')) = \{(t,q)\mid \sqrt{(t-t')^{2}+(q-q')^{2}} <\epsilon\}$. Since $(t^{'n},q^{'n}) \rightarrow (t',q')$,  for some $N$ for all $n\geq N $ we have $(t-\epsilon,1)P(t^{'n},q^{'n})$. By monotonicity of $R$, $(t^{'n_k},q^{'n_k})$s are not indifferent to $(t^{n_k},1)$ where $t^{n_k}<t^{*}+\epsilon<t-\epsilon$. This is a contradiction.    
This is a contradiction. 
Thus $f_{R}$ is continuous.

Define $f_{R}$ similarly for all $R$. Then by richness, the single-crossing property and
the definition of $f_{R}$, $f_{R^n} \rightarrow  f_R$ pointwise.
To see this let $R^{n}\rightarrow R$ and consider $(t,q)\in \mathbb{Z}$. Without loss of generality
let $\{R^{n}\}_{n=1}^{\infty}$ be an increasing sequence. 
Let $f_{R}(t,q)=t'$, i.e., $(t,q)I(t',1)$. Consider $]t'-\epsilon, t']$.
By richness there is $R^N$ such that $(t,q)I^{N}(t^{N},1)$ for some $t^{N}\in ]t'-\epsilon,t']$.
By the single-crossing property for all $n\geq N$, $(t,q)I^{n}(t^{n},1)$, where $t^{n}\in ]t'-\epsilon,t']$. 
Thus for all $n\geq N$, $f_{R^{n}}(t,q)\in ]t'-\epsilon,t']$. This establishes the required pointwise convergence.      
Also note that if $R'\prec R''$, then $f_{R'}(t,q)<f_{R''}(t,q)$. 
If $f_{R^n} \rightarrow  f_R$ pointwise,  then without loss of generality we can consider monotone sequences. Then
by Theorem $7.13$ in \citet{Rudin} $f_{R^n} \rightarrow  f_R$ 
converges to $f_{R}$ uniformly. This establishes Claim \ref{claim:uniform}. 
To see that it is enough to consider monotone subsequence of $\{f_{R^{n}}\}_{n=1}^{\infty}$.
Let $M_{n}=\sup_{x\in K}|f_{R^{n}}(x)-f_{R}(x)|$. For uniform convergence we need to show that 
$\lim_{n\rightarrow \infty}M_{n}=0$. By the single-crossing property every monotone subsequence of $\{M_{n}\}_{n=1}^{\infty}$ corresponds to a monotone subsequence of $ \{ f_{R^{n}}  \}_{n=1}^{\infty}$ and vice versa. For example, $R'\prec R'' \prec R\iff ~\text{for all}~x\in K, f_{R'}(x)< f_{R''}(x) <f_{R}(x)\iff ~\text{for all}~x\in K, f_{R}(x)-f_{R''}(x)<f_{R}(x)-f_{R'}(x)$. 
The first equivalence follows because if $R''$ cuts $R'$ from above at some $z\in \mathbb{Z}$, then $R''$ cuts $R'$ from above at every $z\in \mathbb{Z}$.
Since $\{M_n\}_{n=1}^{\infty}$ converges to $0$ if 
and only if all its monotone subsequences converge to $0$, the required uniform convergence follows from Theorem $7.13$ in \citet{Rudin}.                
We require another claim.

\begin{claim}\rm 
	Let $f_{R^n}
	\rightarrow f_R$ uniformly on a compact set $K
	\subseteq \mathbb{Z}$ 
	and $f_{R^n}$ be continuous. Let $\{x^{n}\}_{n=1}^{\infty}\subseteq K$. If $x^n \rightarrow x$,
	then $\lim_{n\rightarrow \infty} f_{R^n}(x^{n})=f_{R}(x) $
	
	\label{claim:uni_converge}
\end{claim}

\noindent{\bf Proof of Claim \ref{claim:uni_converge}:} Since $K$ is compact, $x\in K$.  
Fix $\epsilon>0$. By uniform convergence there is $N_1$ such that
$\sup_{z\in K}|f_{R^n}(z)-f_R(z)|$
$<\frac{\epsilon}{2}$ for all $n\geq N_1$.
By continuity of $f_R$, since uniform convergence preserves continuity, there is $N_2$ such that
$|f_R(x^n)-f_R(x) |< \frac{\epsilon}{2}$
for all $n\geq N_2$.  
Let $N=\max\{N_1,N_2\}$. Now
$|f_{R^n}(x^n)-f_R(x)|=
|f_{R^n}(x^n)-f_R(x^n)+ f_R(x^n)- f_R(x)|\leq|f_{R^n}(x^n)-f_R(x^n)|+ |f_R(x^n)- f_R(x)|< \epsilon$ for all $n\geq N$.
Hence the proof of the claim follows.

Back to the proof of lemma. Since the sequence of bundles converge we can assume that they lie in a compact set $K$.  By Claim \ref{claim:uniform}  and Claim \ref{claim:uni_converge} $f_{R^n}(t^{1n},q^{1n})\rightarrow f_{R}(t^{1},q^{1})$
and $f_{R^n}(t^{2n},q^{2n})\rightarrow f_{R}(t^{2},q^{2})$. 
Since $f_{R^n}(t^{1n},q^{1n})=f_{R^n}(t^{2n},q^{2n})$, $f_{R}(t^{1},q^{1})=f_{R}(t^{2},q^{2})$.
Therefore, $(t^{1},q^{1})I(t^{2},q^{2})$. This establishes the lemma.    

\medskip

\noindent{\bf Proof of Lemma \ref{lemma:mon}:} 	Let $F:\mathcal{R}^{rsc}\rightarrow \mathbb{Z}$ be a strategy-proof SCF. Let $R', R''\in \mathcal{R}^{rsc}$ be two preferences such that $R'\prec R''$, and $F(R')\neq F(R'')$.
Since $F$ is strategy-proof, $F(R'')R''F(R')$. Thus, 
\begin{equation}\label{nec1}
	F(R')\in LC(R'', F(R'')).
\end{equation}
Again by strategy-proofness, $F(R')R' F(R'')$. Thus,
\begin{equation}\label{nec2}
	F(R')\in UC(R', F(R'')).
\end{equation}
Combining the two equations \eqref{nec1} and \eqref{nec2}, we get,
\begin{equation}\label{nec3}
	F(R')\in LC(R'', F(R''))\cap UC(R', F(R'')).
\end{equation}
Since $R'\prec R''$, by monotonicity of classical preferences 
\begin{equation}\label{nec4}
	\begin{split}
		[LC(R'', F(R''))\cap UC(R', F(R''))\subseteq \{z| z\leq F(R'')\}].
	\end{split}
\end{equation}
From equations \eqref{nec3} and \eqref{nec4} we obtain $F(R')\leq  F(R'')$.
Hence the proof of the lemma follows.

\medskip

\noindent{\bf Proof of Lemma \ref{lemma:preference_preserve}:} We note that $(i)$ is just a rewriting of
the single-crossing condition by using the order on $\mathcal{R}^{rsc}$. For $(ii)$, note that
by richness there is $\widetilde{R}\in \mathcal{R}^{rsc}$ such that $z'\widetilde{I}z''$. 
If $z'Pz''$, then $R\prec \widetilde{R}$; and if $z''Pz'$, then $\widetilde{R}\prec R$. Now by $(i)$ the proof of the lemma follows.   

\medskip

\noindent{\bf Proof of Lemma \ref{lemma:cont_correspondence}:}Let $F:\mathcal{R}^{rsc}\rightarrow \mathbb{Z}$ be strategy-proof. Without loss of generality consider a decreasing sequence $\{R^{n}\}_{n=1}^{\infty}$
that converges to $R$.  	
From Lemma \ref{lemma:mon}, it follows that the sequence $\{F(R^n)\}_{n=1}^{\infty}$ is a decreasing sequence bounded below by $F(R)$. Thus, $\{F(R^n)\}_{n=1}^{\infty}$ converges. Let $z=(t,q)=\lim_{n\rightarrow \infty}F(R^n)$. 
Let by way of contradiction $zIF(R)$ does not hold. Then either $(i)$ $zPF(R)$ or $(ii)$ $F(R)Pz$. 

\noindent {\bf Consider Case $(i)$}: By continuity of $R$ there is an open ball $B(\epsilon,z)=\{z'\in \mathbb{Z}\mid ||z-z'||<\epsilon\}$ such that for all $z'\in B(\epsilon,z)$, $z'PF(R)$. If $z=(t,q), z'=(t',q')$ then $||z'-z||=\sqrt{(t-t')^{2}+(q-q')^{2}}$. Since $\lim_{n\rightarrow \infty }F(R^{n})=z$, there is some positive integer $N$ such that 
$F(R^{N})\in B(\epsilon,z)$ such that $F(R^{N})PF(R)$. This is a contradiction to 
strategy-proofness.

\medskip     

\noindent {\bf Consider Case $(ii)$}: Let $F(R)=(t(R),q(R))$. Since $\{R^{n}\}_{n=1}^{\infty}$ is a decreasing sequence, and $F(R)\neq z$, 
$(t(R),q(R))\leq(t,q)$. By monotonicity of $R$, $t(R)<t$. Note that if $t(R)=t$, then by contradiction hypothesis $q(R)<q$. Then by monotonicity of $R$, $(t,q)P F(R)$, and thus we are in case $(i)$. 
Back to case $(ii)$. By continuity of $R$ there is $B(\epsilon,z)$ such that 
for all $z'\in B(\epsilon,z)$, $F(R)Pz'$. Consider $[t,t']\times [q,q']\subseteq B(\epsilon,z)$.
Note $z=(t,q)\in [t,t']\times [q,q']$. Now $t(R)<t$ and $q(R)<q'$. By richness consider $R^{*}$ such that 
$F(R)I^{*}(t,q')$. By Lemma \ref{lemma:preference_preserve} $R\prec R^{*}$. 
By monotonicity of $R^{*}$, if $z'\neq (t,q'), z'\in [t,t']\times [q,q']$, then $(t,q')P^{*}z'$. 
By By Lemma \ref{lemma:preference_preserve} for all $R'$ such that $R'\prec R^{*}$,
if $z'\neq (t,q'), z'\in [t,t']\times [q,q']$, then $(t,q')P'z'$.
By the single-crossing property $F(R)P'(t,q')$. Since $\lim_{n\rightarrow \infty}F(R^{n})=z$, there is 
$R^{n}\prec R^{*}$ such that $F(R^{n})\in ]t,t'[\times ]q,q'[$. Thus $F(R)P^{n}(t,q')P^{n}F(R^{n})$. 
Hence, $F(R)P^{n}F(R^{n})$ and it is a contradiction to strategy-proofness.     

\begin{remark}\rm This proof goes through for $\mathcal{R}^{rrca}$ if we consider the domain of $F$ to be $[\underline{R},\overline{R}]\subseteq \mathcal{R}^{rsc}$ since for our arguments we do not need $R^{*}\in [\underline{R},\overline{R}]$. Since the proof of the monotonicity of $F$ uses only two preferences, this proof holds if the domain of $F$ is $[\underline{R},\overline{R}]$.      
\end{remark}	

\noindent This completes the proof of the lemma. 

\medskip

\noindent{\bf Proof of Theorem \ref{thm:implies_strtagey_proof_finite_range}:} 	Without loss of generality we assume  $\#Rn(F)=6$, 
i.e., we let $Rn(F)=\{e, c,a,b,d,f\}$ and $e<c<a<b<d<f$. Assuming $\#Rn(F)=6$ is without loss of generality because the argument for the bundles that are non-extreme such as $a,b$ is independent of the number of elements in $Rn(F)$. Also the argument is same for the extreme bundles such as $e$ or $f$. 
We first prove that $F$ is locally strategy-proof in range, it is enough to prove this for the bundles $a$ and $b$.

\begin{claim} \rm ({\bf Local Strategy-proofness in range}) Let $R'\prec  R''$, $F(R')=a$ and $F(R'')=b$. Then $aR'b$ and $bR''a$.
	\label{claim:ab}
\end{claim}

\noindent{\bf Proof of Claim \ref{claim:ab}:}
Let by way of contradiction the claim is false. 
Let without loss of generality $aP''b$. 
Define $S=\{R \in \mathcal{R}^{rsc}\mid F(R)=a\}$ and $T=\{R \in \mathcal{R}^{rsc}\mid F(R)=b\}$. 
Let $R^0$ and $R_0$ be the supremum and the infimum of the sets $S$ and $T$ respectively under the ordering $\prec$. That is, 
$R^0=\sup_{\prec}\{R\mid F(R)=a\}$ and
$R_0=\inf_{\prec}\{R\mid F(R)=b\}$. 
We note that $R'\in S$ and $R''\in T$. Further $R'\precsim R^{0}$ and 
$R_{0}\precsim R''$.
We prove the claim in several steps.

\noindent{\bf Step $1$:}  The following statements hold: $(a)$ $R^0\precsim R''$ , 
$(b)$ $R'\precsim R_0$.

\noindent{\bf Proof of Step $1$: } 
To establish statement (a), let by way of contradiction $R''\prec R^0$. 
By monotonicity of $F$, $F(R'')=b\leq F(R)$ for all $R$ such that $R''\prec R$.       
Thus, $R''$ is an upper bound of $S$. Therefore, if $R''\prec R^{0}$, 
then $R^{0}$ is not the supremum of $S$. This is a contradiction.  
To establish $(b)$ let by way of contradiction $R_{0}\prec R'$. By monotonicity of $F$, 
$F(R)\leq a=F(R')$ for all $R$ such that $R\prec R'$.  Thus $R'$ is a lower bound of $T$.
Therefore, if $R_{0}\prec R'$, then $R_{0}$ is not the infimum of $T$.

\medskip

\noindent{\bf Proof of Step $2$:} $R_{0}=R^{0}$. 

\noindent{\bf Proof of Step $2$:} 
Let by way of contradiction  $R_{0}\neq R^{0}$. 
Thus either Case $(i)$ $R^{0}\prec R_{0}$ or Case $(ii)$ $R_{0}\prec R^{0}$. 

\medskip

\noindent Case $(i):$ Suppose $R^{0}\prec R_{0}$. Since $\mathcal{R}^{rsc}$ is rich there exists a preference $\widehat{R}$ such that $R^0\prec\widehat{R}\prec R_0$. 
Since $R^{0}$ is the supremum of $S$, $F(R^{0})\nless a$, and since $R_{0}$ is the infimum of $T$,
$F(R_{0})\ngtr b$.  
Thus by monotonicity of $F$, $F(R^{0})\in \{a,b\}$ and  
$F(R_{0})\in \{a,b\}$. Again by monotonicity of $F$, $F(\widehat{R})\in \{a,b\}$. 
Let $F(\widehat{R})=a$, and then $\widehat{R}\in S$.
Since $ R^0\prec \widehat{R}$, $R^0$ is not the supremum of $S$. 
This is a contradiction.  
Alternatively, let $F(\widehat{R})=b$. 
Since $ \widehat{R}\prec R_0$, $R_0$ is not the infimum of $T$.
This is a contradiction.  
This establishes that Case $(i)$ cannot occur.

\medskip

\noindent Case $(ii):$ Suppose $ R_0\prec R^0 $. Since $\mathcal{R}^{rsc}$ is rich there exists a preference $\widehat{R}$ such that $R_0\prec\widehat{R}\prec R^0$. We show $F(\widehat{R})\in \{a,b\}$.
If $F(\widehat{R})<a$, then by monotonicity of $F$, $F(R)< a$ for all $R\prec \widehat{R}$. Since $F(R')=a$, $R'$ is a lower bound on $T$. Further, $F(R_{0})<a=F(R')$, thus by monotonicity of $F$, $R_{0}\prec R'$. Thus $R_{0}$ is not the infimum of $T$, which is a contradiction.    
If $F(\widehat{R})>b$, 
by monotonicity of $F$, $F(R^{0})>b$. Since $F(R'')=b$, $R''$ is an upper bound 
on $S$. Further $b=F(R'')<F(R^{0})$, thus by monotonicity of $F$, $R''\prec R^{0}$.   
Thus $R^{0}$ is not the supremum of $S$, which is a contradiction.  
Therefore, $F(\widehat{R})\in \{a,b\}$.

If $F(\widehat{R})=a$, then by monotonicity of $F$ if $F(R)=b$, then $\widehat{R}\prec R$. 
Further, $F(R)\leq a$ for all $R$ such that $R\prec \widehat{R}$. Thus $\widehat{R}$ is a lower bound of $T$, and therefore   
$R_{0}$ is not the infimum of $T$. 
This is a contradiction.     
Now let $F(\widehat{R})=b$. Monotonicity of $F$ implies $F(R)\geq b$ for all $R$ such that $\widehat{R}\prec R$. Further, by monotonicity of $F$, $F(R)=a$, implies $R\prec \widehat{R}$. Thus $\widehat{R}$ is an upper bound of $S$. Therefore, $R^{0}$ is not the supremum of $S$. This is a contradiction.     
This implies neither Case $(i)$ nor Case $(ii)$ hold. Consequently $R^0=R_0$.

\noindent{\bf Step $3$:} Now we complete the proof of the claim.  
We show $F(R_0)=F(R^0)\in \{a,b\}$. 
If $F(R^{0})<a$, then $R^{0}\prec R'$ because $F(R')=a$ and $F$ is monotone. Thus $R^{0}$ is not the supremum of $S$.
If $F(R^{0})>b$, then $R''\prec R_{0}$ because $F(R'')=b$ and $F$ is monotone. Thus $R_{0}$ is not the infimum of $T$. Thus both cases lead to a contradiction. Thus we have $F(R_0)=F(R^0)\in \{a,b\}$.

\noindent We consider two cases. 
We have $a< b$, and by the contradiction hypothesis $aP''b$.

\noindent Case $(i):$ $F(R_0)=F(R^0)=a$. Since $F(R'')=b$, by monotonicity of $F$ it follows that $R_0\prec R''$. 
Hence by Lemma \ref{lemma:preference_preserve} $aP_0b$.
We will now show that this leads to a contradiction.
Consider a sequence $\{R^n\}_{n=1}^{\infty}$ such that for all $n$, $(i)$ $F(R^n)=b$ and $(ii)$ $R_0\prec R^{n+1}\prec R^{n}\prec R''$, $(iii)$ $R^n\rightarrow R_0$.
Since $\mathcal{R}^{rsc}$ is metrizable, by Lemma \ref{lemma:closure} and richness of the domain such a sequence exists. Now $\lim_{n\rightarrow \infty}F(R^n)=b$. Thus by continuity of $V^F$, $b=\lim_{n\rightarrow \infty}F(R^n)I_0F(R_0)$. Since $F(R_0)=a$, and $aP_0b$ this is a contradiction. 

\medskip

\noindent Case $(ii):$ Suppose instead $F(R_0)=F(R^0)=b$. Since $R_{0}=\inf_{\prec}T$, $R_{0}\precsim R''$. 
Hence by Lemma \ref{lemma:preference_preserve} $aP_0b$. Also by monotonicity of $F$, $R'\prec R_0$
Consider a sequence $\{R^n\}_{n=1}^{\infty}$ such that for all $n$, $(i)$ $F(R^n)=a$ and $(ii)$ $R'\prec R^{n}\prec R^{n+1}\prec R_{0}$, $(iii)$ $ R^n\rightarrow R_0$.
Now $\lim_{n\rightarrow \infty}F(R_n)=a$. Again  by continuity of $V^F$, $a=\lim_{n\rightarrow \infty}F(R^n)I_0F(R_0)=b$. Since $F(R_0)=b$ and $aP_0b$, it contradicts the continuity of $V^{F}$.

\noindent An analogous argument leads to a contradiction if $bP'a$. Thus Claim \ref{claim:ab} is established. 

\begin{remark}\rm In the proof of Claim \ref{claim:ab} we have used preferences 
	in $[R',R'']$. Thus we can use this argument if the domain of $F$ were a closed interval 
	$[\underline{R},\overline{R}]$.  	     
\end{remark}

\begin{claim}\rm $aI_{0}b$. 
	\label{claim:indiff}
\end{claim}

\noindent{\bf Proof of Claim \ref{claim:indiff}:}  In Step $2$ in the proof of Claim \ref{claim:ab} we have obtained $R_{0}=R^{0}$. We establish that $aI_{0}b$. 
By the richness of $\mathcal{R}^{rsc}$ there is $R$ such that $aIb$.

\begin{remark}\rm By Claim \ref{claim:ab} and Lemma \ref{lemma:preference_preserve}
	$R'\precsim R\precsim R''$. Thus we can use this argument if the domain of $F$ were a closed interval 
	$[\underline{R},\overline{R}]$.	  
\end{remark}

\noindent We show that $R=R_{0}=R^{0}$. First we show $F(R_{0})=F(R^{0})\in \{a,b\}$. If $F(R^{0})<a$, then $R^{0}$ is not the supremum of $S$ since $F(R')=a$ and thus $R^{0}\prec R'$ by monotonicity of $F$. If $F(R_{0})>b$, then $R_{0}$ is not the infimum of $T$ since $F(R'')=b$ and thus $R''\prec R_{0}$ by monotonicity of 
$F$. Therefore  $F(R^{0})<a$ and $F(R_{0})>b$ lead to a contradiction. Thus $F(R_{0})=F(R^{0})\in \{a,b\}$.

Let by way of contradiction $aP_{0}b$. Then by Claim \ref{claim:ab}, $F(R_{0})=a$. Since $aIb$, by Lemma \ref{lemma:preference_preserve} $R'\precsim R_{0}\prec R$. Since $R_{0}$ is the supremum of $S$, 
$F(R)\geq b$. By Claim \ref{claim:ab} $bR'' a$.
If $aI''b$, then by the single-crossing property $R=R''$. Therefore, by the single-crossing property or by Lemma \ref{lemma:preference_preserve} 
for all $R^*\in [R',R''[$, $aP^{*}b$. By Claim \ref{claim:ab} and monotonicity of $F$, $F(R^{*})=a$. Thus $R''=\inf_{\prec}T$. Since $aP_{0}b$, $R_{0}\neq R''$ and thus $R_{0}$ is not the infimum of $T$. This is a contradiction.     
Let $bP''a$. Then by the single crossing-property $R'\prec R\prec R''$. Since $F(R)\geq b$, by Claim \ref{claim:ab} and monotonicity of $F$, 
$F(R)=b$. By the single-crossing property for all $R^{*}\in [R',R[$, $aP^{*}b$. Thus by Claim \ref{claim:ab} and monotonicity of $F$, $F(R^{*})=a$. Since $aP_{0}b$, $R_{0}\neq R$ and thus $R_{0}$ is not the infimum of $T$. This is a contradiction.

Let by way of contradiction $bP_{0}a$.Then by Claim \ref{claim:ab}, $F(R_{0})=b$. Then $R\prec R_{0}\precsim R''$ where $aIb$. Since $R_{0}$ is the infimum of $T$, 
$F(R)\leq a$. By Claim \ref{claim:ab} $aR' b$. If $aI'b$, then by the single-crossing property $R=R'$. 
Then by the single-crossing property for all $R^{*}\in ]R',R'']$, $bP^*a$. By Claim \ref{claim:ab} and monotonicity of $F$, $F(R^{*})=b$. Thus $R'=\sup_{\prec} S$. Since $bR_{0}a$ and $R_{0}\neq R'$, $R_{0}$ is not the supremum of $S$. This is a contradiction. Let $aP'b$. Then by the single-crossing property 
$R'\prec R \prec R''$. Since $F(R)\leq a$ by Claim \ref{claim:ab} and monotonicity of $F$, $F(R)=a$.
Then, by the single-crossing property for all $R^{*}\in ]R,R'']$, $bP^{*}a$. Thus by Claim \ref{claim:ab} and monotonicity of $F$, $F(R^{*})=b$. Since $bP_{0}a$, $R_{0}\neq R$ and thus $R_{0}$ is not the supremum of $S$. This is a contradiction. 
Thus Claim \ref{claim:indiff} is established. 	

\begin{claim}\rm Let $F(R_{e})=e, F(R_{c})=c, F(R_{a})=a, F(R_{b})=b, F(R_{d})=d, F(R_{f})=f$. Then 
	$xR_{x}y$ for all $x,y\in Rn(F)$. 
	
	\label{claim:sp_all_bundles}	
\end{claim}

\noindent{\bf Proof of Claim \ref{claim:sp_all_bundles}:} By Claim \ref{claim:ab} 
$eR_{e}c$, $cR_ce,cR_ca, aR_ac,aR_ab, bR_ba,bR_bd, dR_da, dR_df, fR_fd$. 
Recall that $e<c<a<b<d<f$. Thus 
$F$ is locally strategy-proof in range. 
Now we show that $F$ is strategy-proof. From Claim \ref{claim:indiff} we have $R_{k},k=1,\ldots,5$ such that 
$eI_{1}c, F(R_1)\in \{e,c\}; cI_2a, F(R_{2})\in \{c,a\}; aI_{3}b, F(R_{3})\in \{a,b\}; bI_{4}d, F(R_{4})\in \{b,d\}; dI_{5}f, F(R_{5})\in \{d,f\}$. 
Further, $R_{1}\precsim R_{2}\precsim R_3 \precsim R_4  \precsim R_5$. 
Suppose by way of contradiction $R_{2}\prec R_1$. Since $eI_1c$, by the single-crossing property or by  Lemma \ref{lemma:preference_preserve} $ e P_{2} c$.
By Claim \ref{claim:ab}, $F(R_{2})\neq c$. That is if $F(R_2)=c$, then local strategy-proofness in range is violated. If $F(R_2)>c$, then monotonicity of $F$ is violated since $F(R_{1})\in \{e,c\}$.
Thus $R_{1}\precsim R_2$. Let $R_{3}\prec R_{2}$. Since $cI_2a$, by the single-crossing property 
$cP_3a$. By Claim \ref{claim:ab}, $F(R_{3})\neq a$. Thus if $F(R_3)=b$, then monotonicity of $F$ is violated because $F(R_2)\in \{c,a\}$. Thus $R_1\precsim R_2\precsim R_3$. By induction $R_1\precsim R_2\precsim R_3  \precsim R_4\precsim R_5$.

We show that $F$ restricted to $\{R_{i}\mid i=1,\ldots,5 \}$ is strategy-proof. 
We have $F(R_{1})\in \{e,c\}$ and $eI_1c$. Let $R_{1}\prec R_{2}$. 
We have $cI_{2}a$. By the single-crossing property or by Lemma \ref{lemma:preference_preserve} $cP_{1}a$.
Further we have $R_1\prec R_{3}$ and $aI_{3}b$. By the single-crossing property $aP_{1}b$. 
Thus by transitivity of $R_1$, $cP_{1}b$. 
In this way we obtain $eI_{1}cP_{1}x, x\notin \{e,c\}$.  
The arguments for $R_{2}$ are the same. 
Consider $R_{3}$. Let $R_{2}\prec R_{3}$. 
Since $cI_{2}a$, by the single-crossing property $aP_{3}c$.  
Since $eI_{1}c$, by the single-crossing property $cP_{3}e$. Thus by transitivity of $R_3$ $aP_{3}e$. The rest of the arguments are similar.   

Now we construct the strategy-proof mechanism.    
Consider $R\prec R_{1}$. By monotonicity of $F(R)\leq F(R_1)\leq c$.
By the single-crossing property $ePc$. 
Thus by Claim \ref{claim:ab}, $F(R)=e$.   
Since $R_{1}\precsim R_{2}$, it follows that $R\prec R_{2}$.
Since $cI_{2}a$ by the single-crossing property $cPa$. Thus by transitivity of $R$, $ePa$. Continuing this argument finitely many times we obtain $ePx$ for all $x\neq e$ and $x\in Rn(F)$. Thus for $R\prec R_{1}$, $F$ is strategy-proof.      

Now consider $R\in ] R_1,R_2[$.  By monotonicity of $F$, $F(R)\in \{e,c,a\}$. Since $eI_1c$, by the single-crossing property $cPe$. Since $cI_{2}a$, by the single-crossing property $cPa$. 
Thus by Claim \ref{claim:ab}, $F(R)=c$. Since $R\prec R_{3}$, and $aI_{3}b$, by the single-crossing property
$aPb$. Then $cPa$ implies $cPb$. This argument can be used finite number of times to show that $cPx$ for all $x \in Rn(F)$ such that $x\neq c$. Now it follows that $F$  restricted to $]\infty,R_{2}]$ is strategy-proof and the restriction is defined below,

$$F(R) = \begin{cases}
	e, & \text{if $R \prec R_{1}$;}\\
	\text{either}~ e~\text{or}~ c, & \text{if $R=R_1$}\\
	c & \text{if $R_1\prec R \prec R_2$}\\
	\text{either}~c~\text{or}~a & \text{if $R=R_2$.}
\end{cases}$$

\noindent For the sake of completion we extend the argument for $R\in ]R_2,R_3[$. Since $aI_{3}b$, by the single-crossing property for $R \prec R_3$, $aPb$. Since $cI_2a$ and
$R_{2}\prec R$, $aPc$. The continuation of this argument entails $aPx$ for all $x\neq a$ and $x\in Rn(F)$. Thus the general $F$ is defined as follows.

$$F(R)=\begin{cases}
	e, & \text{if $R \prec R_{1}$;}\\
	\text{either}~ e~\text{or}~ c, & \text{if $R=R_1$}\\
	c & \text{if $R_1\prec R \prec R_2$}\\
	\text{either}~c~\text{or}~a & \text{if $R=R_2$}\\
	a & \text{if $R_2\prec R\prec R_3$}\\
	\text{either}~a~\text{or}~b & \text{if $R=R_3$}\\
	b & \text{if $R_3\prec R\prec R_4$}\\
	\text{either}~b~\text{or}~d & \text{if $R=R_4$}\\
	d & \text{if $R_4\prec R\prec R_5$}\\
	\text{either}~d~\text{or}~f & \text{if $R=R_5$}\\
	f & \text{if $R_5\prec R$.}
\end{cases}$$

\noindent To complete the proof we note $\#\{x\in Rn(F)\mid y\in Rn(F), xR_iy~ \text{for some}~R_{i}\in \{R_{1},\ldots, R_5 \}\}\leq 3$. This says that the special preferences $R_i$s, i.e., the preferences that are indifferent between two bundles in the range,   can be indifferent with at most one more bundle. To see this without loss of generality consider $R_{3}$, we have $aI_{3}b$. Let $a<x<b$, such that $aI_{3}bI_{3}x$ and $x\in Rn(F)$ . By the single-crossing property $aPx$ for $R\prec R_{3}$, and $bPx$ for $R_{3}\prec R$. 
By Claim \ref{claim:ab}, $F(R)\neq x$ for $R\neq R_{3}$. Thus $F(R_{3})=x$.
If there is any other $y\in Rn(F)$ with $a<y<b$, then $F(R_{3})=y$. Since $F$ is a function $x=y$.    

\begin{remark}\rm This constriction of $F$ holds for interval $[\underline{R},\overline{R}]$, since the construction is unaffected if there are no preferences $R$ such that $R\prec R_1$ or $R_{5}\prec R$.

\end{remark}

\noindent This completes a proof of the theorem. 

\medskip

\medskip

\noindent{\bf Proof of Lemma \ref{lemma:measurable}:} Let $Rn(F)=\{(t^i,q^i),\mid i=1,\ldots ,n\}$. For any $t^i$, 
$t^{-1}(\{t^{i}\})$ is an interval in $\mathcal{R}^{rsc}$.  All intervals are in $\mathcal{B}$. Let $B\in B(\Re)$. Then $t^{-1}(B)=\cup_{j_i\in B}t^{-1}(\{t^{j_i}\})\in \mathcal{B}$. The proof for $q$ is analogous. Hence the lemma follows.

\medskip

\noindent{\bf Proof of Lemma \ref{lemma:soln_exists}:} In Step $1$ we show that the objective function is continuous, and then in Step $2$ we show that the constraint set is compact. Then from Step $1$ and $2$ the existence of maximum follows. For the sake of simplicity of notations, in this proof we denote preferences by $S$ and $ R$. 	  

\noindent{\bf Step $1$: The objective function is continuous} Consider the function: $M:[\underline{R},\overline{R}]\times [\underline{R},\overline{R}]\rightarrow [0,1] $ defined by $M(S,R)=\mu([S,R])$ if  $S\precsim R$ and $M(S,R)=\mu([R,S])$ if $R \prec S$. We show that this function is continuous, which follows because $\mu$ is continuous. To see this let $S^{n},R^{n}$ converge to $S,R$. 
Let without loss of generality $S\precsim R$, 
i.e., we have the interval $[S,R]$. Let $S\prec R$. Since $[\underline{R},\overline{R}]\times [\underline{R},\overline{R}]$ is metrizable, 
it is enough to show $\lim_{n\rightarrow \infty}M(S^{n},R^{n})=M(S,R)$, where $\lim_{n\rightarrow \infty}S^{n}=S, \lim_{n\rightarrow \infty}R^{n}=R$. 
Further, without loss of generality it is enough to consider monotone sequences. That is, $(a)$ $ S^{n} \uparrow S, R^{n}\uparrow R$,
$(b)$ $ S^{n} \downarrow S, R^{n}\uparrow R$ , $(c)$ $ S^{n} \downarrow S                                                                                                                                                                                                                                                                                                                                                                                                                                                                                                                                                                                                                                                                                                                                                                                                                                                                                                                                                                                                                                                                                                                                                                                                                                                                                                                                                                                                                   , R^{n}\downarrow R$, $(d)$ $ S^{n} \uparrow S, R^{n}\downarrow R$.

Consider $(a)$ $S^{n} \uparrow S, R^{n}\uparrow R$. There exists $N$ such that $S\prec R^{n} $ for all $n\geq N$. Let $[S^{n},R^{n}]=
[S^{n},S]\cup]S,R^{n}]$. Then $[S^{n},S]\downarrow\{S\}$ and $ ]S,R^{n}]\uparrow]S,R]$. Since $\mu$ is a probability 
$\lim_{n\rightarrow \infty}\mu([S^{n},S])=\mu(\{S\})=0$ and 
$\lim_{n\rightarrow \infty} \mu(]S,R^{n}]=\mu]S,R]=\mu([S,R])$. Since $\mu$ is continuous, 
$\mu(]S,R])=\mu([S,R])$. The other cases can be proved analogously. The function $(t,S,R)\mapsto tM(S,R)$ is continuous; the objective function is the sum of $l$ such functions. Therefore, the objective function is continuous.

\noindent Next we show that the constraint set is compact.

\noindent{\bf Step $2$: The constraint set is compact} 
The constraint set can be rewritten as:  	

$C=\{(t^{0},t^{1},\ldots,t^{l-1}, q^{0},\ldots,q^{l-1}, R^{0},\ldots, R^{l})\mid
(t^{k},q^{k})I_{R^{k}}(t^{k-1},q^{k-1}),t^{k-1}\leq t^{k}, q^{k-1}\leq q^{k};
k=1\ldots,l-1, 
R^{k-1}\precsim R^{k},k=1,\ldots,l$, $R^{0}=\underline{R}$,  $R^{l}=\overline{R}\}$.

\noindent By individual rationality   
$F_l(\overline{R})\overline{R}(0,0)$. Let $(\overline{T},1)\overline{I}(0,0)$. 
By the single-crossing property for $R\prec \overline{R}$ if $(t,q)I(0,0)$, then 
$t<\overline{T}$. Thus by individual rationality let $F_l(R)\leq \overline{T}$ for all $F_l\in \mathcal{F}_l$.

\noindent Thus, 
\[C\subseteq [0,\overline{T}]\times [0,\overline{T}]\times\ldots \times[0,\overline{T}]\times\ldots\times[0,1]\times\ldots \times[0,1]\times \{\underline{R}\} \times [\underline{R},\overline{R}]\times\ldots [\underline{R},\overline{R}]\times \{\overline{R}\}\equiv \Sigma\]

\noindent Being a finite product of compact spaces, $\Sigma$ is compact in the product metric topology. We show that $C$ is closed in $\Sigma$. 
Let, \[x^{n}=\{(t^{0n},\ldots,t^{(l-1)n}, q^{0n},\ldots,q^{(l-1)l}, R^{0n}, R^{1n},\ldots, R^{ln})\}_{n=1}^{\infty}\]
\noindent  be a sequence in $C$ that converges to $x=(t^{0},\ldots,t^{(l-1)}, q^{0},\ldots,q^{(l-1)}, R^{0}, R^{1},\ldots, R^{l})$.
Inequalities are maintained in the limit. Further, by Lemma \ref{lemma:preference_conv} indifference is maintained in the limit. Thus, $x\in C$ and hence $C$ is a closed subset of a compact space. Thus $C$ is compact. 
Therefore the optimization problem in Theorem \ref{thm:optimal} has a solution. 
This completes the proof of Lemma \ref{lemma:soln_exists}.  

Now we complete the proof of Theorem \ref{thm:optimal}. 
For all $l$, $0\leq E(F_l^{*})\leq \sum_{k=0}^{l-1}\overline{T}\mu([R^{k},R^{k+1}])=\overline{T}$. Thus 
$E=\{E(F)\mid F\in \mathcal{F}\}$ is bounded. This completes the proof of Theorem \ref{thm:optimal}.

\noindent{\bf Proof of Theorem \ref{thm:countable}:}
We need the following intermediary result. 

\begin{lemma}\rm Let $F:[\underline{R},\overline{R}]\rightarrow \mathbb{Z}$ be strategy-proof and individually rational. Let $F(R')$ be a limit point of $Rn(F)$. 
	Let $R^{*}=\sup\{R\mid F(R)=F(R')\}$. If there is a sequence $\{F(R^{n})\}_{n=1}^{\infty}$ such that $F(R')<F(R^{n+1})<F(R^{n})$ and $F(R^{n})\downarrow F(R')$, then $F(R^*)=F(R')$. Further, let 
	$R^{**}=\inf\{R\mid F(R)=F(R')\}$. If there is a sequence $\{F(R^{n})\}_{n=1}^{\infty}$ such that $F(R^n)<F(R^{n+1})<F(R')$ and $F(R^{n})\uparrow F(R')$. Then $F(R^{**})=F(R')$.      
	\label{lemma:limit_point_supremum}
\end{lemma}

\begin{proof} Since $F$ is strategy-proof, $F$ is monotone. Therefore, given that 
	$F(R')$ is a limit point of $Rn(F)$ either an increasing or a decreasing sequence of bundles from the range must converge to $F(R')$.  		
	We prove the claim for supremum.  
	Let by way of contradiction $F(R')< F(R^{*})$. 
	But then there is $N$ such that $F(R')<F(R^{n})<F(R^{*})$ for all $n\geq N$. 
	By monotonicity of $F$, for all  $R\in]R',R^n[$ we have $F(R')\leq F(R)\leq F(R^{n})$. Further if $R^{n}\precsim R$, then $F(R^n)\leq F(R)$. Thus given that $F(R')<F(R^{n})$, if $R\in \{R\mid F(R')=F(R)\}$, then $R\precsim R^{n}$ i.e., $R^{n}$ is an upper bound of $\{R\mid F(R')=F(R)\}$. 
	Thus $R^*$ is not the supremum since $F(R^{n})<F(R^{*})$.     
	This is a contradiction. The argument for infimum is the same. 
	This completes the proof of the Lemma.

\end{proof}

\noindent Back to the proof of Theorem \label{thm:countable}.  
First we consider the situation where the number of limit points of $Rn(F)$ is empty, i.e., 
$Rn(F)$ is a countable discrete set. Claim \ref{claim:next} shows that the notions of 'the next' and `the preceding' bundle in $Rn(F)$ are well defined.

\begin{claim}\rm Let $Rn(F)$ be countable, closed and has no limit point. If $F(R')<F(R)<F(R'')$, then there are $F(R^*), F(R^{**})$ 
	such that $F(R')\leq F(R^{*})<F(R)<F(R^{**})\leq F(R'')$ and no other bundle from the range lies between $F(R^{*})$ and $F(R^{**})$. Thus $F(R^{*})$ is the bundle preceding $F(R)$, and $F(R^{**})$ is bundle next to $F(R)$. If $F(R)$ is the smallest bundle in the range, then $F(R^{**})$ is the next bundle, and if  $F(R)$ is the largest bundle in the range then $F(R^{*})$ is the bundle preceding $F(R)$.{\footnote{Largest and smallest is defined according to the order $<$ on $Rn(F)$. We have defined earlier that $(t',q')<(t'',q'')$ if and only if $t'<t''$ and $q'<q''$.}}           
	\label{claim:next} 
\end{claim}	 

\noindent{\bf Proof of Claim \ref{claim:next}:} Let $F(R)$ be not the largest bundle. Let by way of contradiction there is no bundle in the range next to $F(R)$. Thus there is a sequence $\{F(R^{n})\}_{n=1}^{\infty}$ such that $F(R)<F(R^{n+1})<F(R^{n})$ for all $n$. Since the sequence is bounded below, it will converge. Let $F(R)\leq (t,q)=\lim_{n\rightarrow \infty}F(R^{n})$. Since $Rn(F)$ is closed $(t,q)\in Rn(F)$. But then $(t,q)$ is a limit point of $Rn(F)$ that lies in the range. This is a contradiction. Arguments for other cases are similar. Thus Claim \ref{claim:next} follows.        

\noindent  Let $Rn(F)^{*}\subseteq Rn(F)$ be finite such that $(a)$ 
if $F(R'), F(R'')\in Rn(F)^{*}$, then $ F(R)\in Rn(F)^{*}$ for $R'\precsim R \precsim R''$, and $(b)$ $\# Rn(F)^{*}\geq 2$. We note that $Rn(F)^{*}$ is well defined because the number of limit points in $Rn(F)$ is finite.   
Let $\mathcal{R}^{rsc}(Rn(F)^*)=\{R\mid F(R)\in Rn(F)^{*}\}$.  
By Theorem \ref{thm:implies_strtagey_proof_finite_range} $F$ restricted to $\mathcal{R}^{rsc}(Rn(F)^*)$ is strategy-proof.

Now we complete the proof for $Rn(F)$ with no limit points. We prove this case by induction. Consider $Rn(F)^{*1}$ with only two distinct bundles.
Let the two bundles be $F(R')$ and $F(R'')$ and $R'\prec R''$.  
Then $F$ restricted to $\mathcal{R}^{rsc}(Rn(F)^{*1})$ is strategy-proof. 
Then choose $R'''\prec R'\prec R''\prec R''''$ 
such that by Claim \ref{claim:next} 
there are no bundle in the range between $F(R'''')$ and $F(R')$, and $F(R'')$ and $F(R'''')$. By using $F(R'''), F(R'),F(R''),F(R'''')$ define $\mathcal{R}^{rsc}(Rn(F)^{*2})$. Then $F$ restricted to $\mathcal{R}^{rsc}(Rn(F)^{*2})$ is strategy-proof. By induction this process produces a strategy-proof mechanism. 

To ensure that $F$ is indeed strategy-proof, let by way of contradiction $F(R^{**})P^{*}F(R^{*})$. Let without loss of generality $F(R^*)<F(R^{**})$. Then note that $\{F(R)\mid F(R^{*})\leq F(R) \leq F(R^{**})\}$ is finite. If this set is not finite, then this set is a bounded infinite set. Then this set has a limit point, and since $Rn(F)$ is closed the limit point is in the range. This leads to a contradiction to our assumption that $Rn(F)$ has no limit points. But then by Theorem \ref{thm:implies_strtagey_proof_finite_range}, $F$ restricted to $\{R\mid F(R^{*})\leq F(R) \leq F(R^{**})\}$ is strategy-proof. This contradicts $F(R^{**})P^{*}F(R^{*})$. Further note that 
$\{R\mid F(R^{*})\leq F(R) \leq F(R^{**})\}\subseteq \mathcal{R}^{rsc}(Rn(F)^{*n})$ for some $n$. Without loss of generality let $F(R^{**})<F(R^{'})$. But then $\{F(R)\mid F(R^{*})\leq F(R)\leq  F(R'')\}$ is finite. Here $F(R')$ and $F(R'')$ refer to the bundles with which $Rn(F)^{*1}$ is defined. Thus the proof for $Rn(F)$ with no limit points follows. Now we assume that $Rn(F)$ may have finite number of limit points.

Let without loss of generality $F(R^{1})<F(R^{2})<F(R^{3})$ be three limit points of $Rn(F)$. The argument for $Rn(F)$  with limit points more than three is similar.

Consider $\{F(R)\mid F(R^{1})\leq F(R) \leq F(R^{2})\}$.    
Consider $F(R^{1})<F(R^*)<F(R^{**})<F(R^{2})$. Then note that $\{F(R)\mid F(R^{*})\leq F(R)\leq F(R^{**})\}$ is finite. Then by Theorem \ref{thm:implies_strtagey_proof_finite_range}, $F$ restricted  to $\{R\mid F(R^{*})\leq F(R)\leq F(R^{**})\}$ is strategy-proof. By the argument applied to the case for no limit point, $F$ restricted to $\{R\mid F(R^1)<F(R)<F(R^2)\}$ is strategy-proof. We can apply the same argument because $\{F(R)\mid F(R^1)<F(R)<F(R^2)\}$ is discrete.    
That is, $F(R')R'F(R'')$ 
if $F(R')$ and $F(R'')$ are strictly between $F(R^1)$ and $F(R^2)$.

Now we show for $i=1,2$, $F(R^i)R^iF(R)$ if $F(R^{1})<F(R)<F(R^{2})$.  Let by way of contradiction
$F(R)P^{1}F(R^1)$. Now two situations  emerge. The first situation is $\{F(R')\mid F(R^1)\leq F(R')\leq F(R)\}$ finite. 
In this situation an increasing sequence of bundles from the range converge to $F(R^1)$ since $F(R^1)$ is a limit point. However, $\{F(R')\mid F(R^1)\leq F(R')\leq F(R)\}$ finite means that by Theorem \ref{thm:implies_strtagey_proof_finite_range} $F$ restricted to $\{R\mid F(R')\mid F(R^1)\leq F(R')\leq F(R)\}$ is strategy-proof.
To see this note that for $R^{*}\in [R^1,R]$, $F(R^1)\leq F(R^{*})\leq F(R)$. 
Thus we can use Theorem \ref{thm:implies_strtagey_proof_finite_range} since this theorem holds for closed intervals as well. 
Therefore $F(R)P^{1}F(R^1)$ cannot happen.           

Now, let $\{F(R')\mid F(R^1)\leq F(R')\leq F(R)\}$ be not finite and $F(R)P^1 F(R^1)$. Let $R^{*}=\sup \{R'\mid F(R')=F(R^1)\}$. Since $F(R^1)<F(R)<F(R^2)$, and we have assumed that the number of limit points is three, by Lemma \ref{lemma:limit_point_supremum}, $F(R^{*})=F(R^1)$.
To see this let by way of contradiction $F(R^{0})$ be such that 
$F(R^1)<F(R^{0})<F(R)$ and $\{F(R')\mid F(R^1)\leq F(R')\leq F(R^0)\}$ is finite.
Then $\{F(R')\mid F(R^0)\leq F(R')\leq F(R) \}$ is infinite and has a limit point. 
This limit point is not from $\{F(R^1),F(R^2),F(R^3)\}$ entailing a contradiction.  
By Lemma \ref{lemma:preference_preserve}, $F(R)P^*F(R^1)$. Also by Lemma \ref{lemma:preference_preserve}, $F(R)P'F(R^1)$ for all $R'$ such that $R^*\prec R'$. By continuity of $R^{*}$ let $B(\epsilon, F(R^1))$ be such that if $z\in B(\epsilon, F(R^1))$, then $F(R)P^*z$.
Since $F(R^*)=F(R^1)$ is a limit point of $Rn(F)$ and $\{F(R')\mid F(R^1)<F(R')<F(R)\}$ is not finite where $F(R)<F(R^2)$, there is a sequence $\{F(R^{n})\}_{n=1}^{\infty}$, $F(R^{*})<F(R^{n+1})<F(R^{n})$ such that $F(R^{n})\downarrow F(R^{*})$. By monotonicity of $F$, $F(R^{*})<F(R^{n+1})<F(R^{n})$ holds. Choose $n$ such that $F(R^{n})\in B(\epsilon, F(R^1))$. Then $F(R)P^{n}F(R^{n})$.            	   
Since $\{F(R^{***})\mid F(R^{n})\leq F(R^{***})\leq F(R)\}$ is finite, by Theorem \ref{thm:implies_strtagey_proof_finite_range}, $F(R^{n})R^{n}F(R)$. This is a contradiction. Thus    
$F(R^{1})R^{1}F(R)$. Similarly, $F(R^{2})R^{2}F(R)$.

Now we show that for $R$ with $F(R^1)<F(R)<F(R^2)$ it follows that $F(R)RF(R^{i}),i=1,2$. Let $\{F(R^{n})\}_{n=1}^{\infty}$ be a decreasing sequence such that $F(R^{n})\downarrow F(R^{1})$. Now $\{F(R^{***})\mid F(R^{n})\leq F(R^{***})\leq F(R)\}$ is finite, and thus by Theorem \ref{thm:implies_strtagey_proof_finite_range}, $F(R)RF(R^{n})$. By continuity of $R$, 
$F(R)RF(R^{1})$. Similarly, $F(R)RF(R^{2})$. 
We have $F(R^1)R^1 F(R)$ where $F(R^1)<F(R)<F(R^2)$. We also have $F(R)R F(R^2)$. Since $R^1\prec R$, by Lemma \ref{lemma:preference_preserve}, $ F(R) R^1  F(R^2)$. By transitivity, $F(R^1)R^1F(R^2)$. The argument for $F(R^2)R^2F(R^1)$ is similar.

Thus $F$ restricted to $\{R\mid F(R^{1})\leq F(R) \leq F(R^{2})\}$ is strategy-proof. 
Similarly, $F$ restricted to $\{R\mid F(R^{2})\leq F(R) \leq F(R^{3})\}$ is strategy-proof. 
Next we argue that $F$ is strategy-proof. Assume that $\{F(R)\mid F(R^1)<F(R)<F(R^{2})\}$ is not finite.
Suppose $F(R^1)I^{1}F(R^{2})$. Consider $z$ such that $F(R^1)R^{1}z$ and $F(R^{1})<z<F(R^{2})$. By the single-crossing property if $R^{1}\prec R$, then $F(R^{2})P z$. This violates strategy-proofness for $\{R\mid F(R^{1})< F(R)< F(R^{2})\}$. Thus $F(R^{1})P^{1}F(R^{2})$. 
Let by way of contradiction $F(R^{3})P^{1}F(R^{1})$. 
Since $F$ restricted to $\{R\mid F(R^{1})\leq F(R) \leq F(R^{2})\}$ and $\{R\mid F(R^{2})\leq F(R) \leq F(R^{3})\}$ is strategy-proof, $F(R^{2})R^{2}F(R^{1})$, and $F(R^{2})R^{2}F(R^{3})$. Since $F(R^{1})P^{1}F(R^{2})$, this contradicts the single-crossing property. In particular $R^1$ and $R^2$ cross more than once.   
Similar arguments show that if $F(R^{1})<F(R)<F(R^{2})$, then $F(R)RF(R')$ where $F(R^{2})<F(R')<F(R^{3})$. Thus, $F$ restricted to $\{R\mid F(R^{1}) \leq F(R)\leq F(R^{3})\}$ is strategy-proof. 

If there is no sequence $\{F(R^{n})\}_{n=1}^{\infty}$ such that $F(R^{3})<F(R^{n})$ with the sequence converging to $F(R^{3})$, then $F$ restricted to $\{R\mid F(R^{3})\leq F(R)\}$ is strategy-proof since $\{F(R)|F(R^{3})\leq F(R)\}$ is discrete. If such a sequence exists, then the arguments are similar to the one for $\{R\mid F(R^{1})\leq F(R) \leq F(R^{2})\}$. Similarly, $F$ restricted to $\{R\mid F(R)\leq F(R^1)\}$ is strategy-proof.   
Let by way of contradiction $F(R^3)<F(R)$ and $F(R^{2})PF(R)$. Since $F(R^{3})R^{3}F(R^{2})$, $F(R^{3})R^{3}F(R)$, this contradicts the single-crossing property. That is $R$ and $R^{3}$ cross more than once because $F(R)RF(R^{3})$.     
Hence the single-crossing property ensures that $F$ is strategy-prof. This completes the proof of Theorem \ref{thm:countable}. 

\medskip

\noindent{\bf Proof of Theorem \ref{prop:opt_count}: } 

Without loss of generality we assume that $Rn(F^c)$ has one limit point, and that it can be approached by both increasing and decreasing sequences.  Fix $\epsilon>0$. Let $(t^*,q^*)$ be the limit point of $Rn(F^c)$. Let $R'=\inf\{R\mid F^c(R)=(t^{*},q^*)\}$, $R''=\sup\{R\mid F^c(R)=(t^{*},q^*)\}$. 
Then by Lemma \ref{lemma:limit_point_supremum}  
$F^c(R')=F^c(R'')=(t^*,q^*)$. The expected revenue is as defined in Equation (\ref{eqn:opt_count}). 
We note that $R^n\uparrow R'$. Consider the interval $]R^{0},R']$. We show that 
eventually the sequence $\{R^{n}\}_{n=1}^{\infty}$ lies in $]R^{0},R']$. 
Since $(t^*,q^*)$ is a limit point there is $N$ such that
$F^{c}(R^{0})<(t^{N},q^{N})<(t^*,q^*)$. By strategy-proofness $F^c(R^{N})\in \{(t^N,q^N), (t^{N+1},q^{N+1}) \}$. By monotonicity of $F^c$, $R^0\prec R^N$.  
Now by monotonicity of $F^c$ for all  $n>N$,
$t^n \in ]t^{N},t^*] $ and $q^n \in ]q^{N},q^*]$. 
By monotonicity of $F^{c}$, 
for all $n> N$, $R^{n}\in ]R^{N},R']\subseteq [R^{0},R']$. Analogously, $R^{k}\downarrow R''$.

Consider, $F^c(R^{w})$, where $F^c(R^{w})$ is an arbitrary element in the sequence $\{F^c(R^{n})\}_{n=1}^{\infty}$. 
By richness, let $R^{w}\prec R^{w+l}\prec R'$ be such that $F^c(R^{w})I^{w+l}(t^{*},q^*)$. 
By strategy-proofness, and since $F^c(R')=(t^{*},q^*)$ is a limit point,  $(t^{*},q^*)P'F^c(R^w)$ and $F^c(R^{w})P^{w}(t^{*},q^*)$. Thus $R^{w+l}$ is well defined. Note that
there are infinitely many elements between $F^{c}(R^w)$ and $F^{c}(R')=(t^{*},q^*)$ since $(t^{*},q^*)$
is a limit point. Then Lemma \ref{lemma:preference_preserve} rules out $F^{c}(R^{*})I'F^{c}(R^{w})$ and $F^{c}(R^{w})I^{w}F^{c}(R^{*})$. Figure $3$ depicts is pictorially.    

\begin{center} 
	\includegraphics[height=6cm, width=7cm]{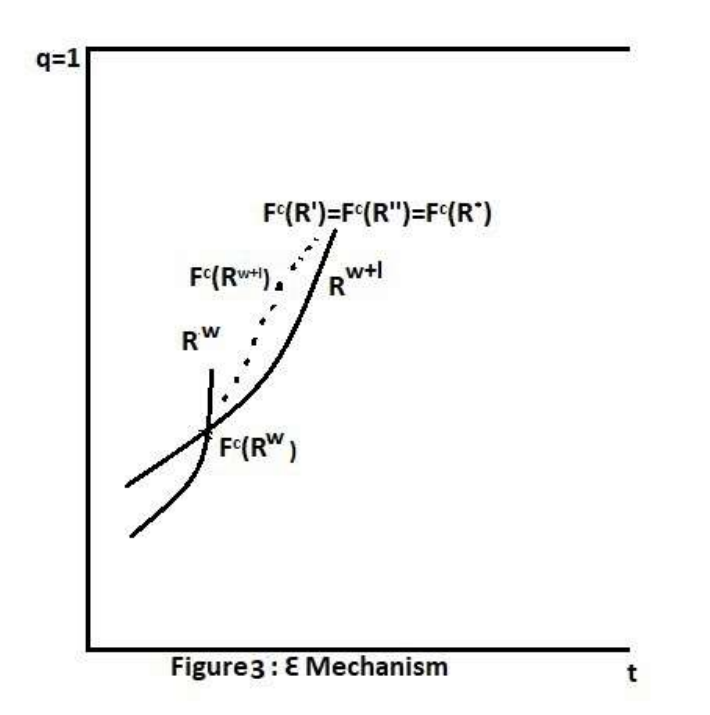}
\end{center}

\noindent We note that if for $R$ is such that $R^w\prec  R \prec R'$, $F^c(R)\in UC (R^{w+l}, (t^*,q^*))=UC (R^{w+l}, F^{c}(R^w))$ and if $F^c(R^w)\neq F^c(R)$ then $F^{c}(R) P^{w+l} (t^*,q^*)$. 
By monotonicity of $F^{c}$, $F^c(R^w)\leq F^c(R)\leq (t^*,q^*)$. By strategy proofness
$F^{c}(R)R(t^*,q^*)$ and $F^{c}(R)RF^{c}(R^w)$. Now $F(R)^c<(t^*,q^*)$. If $F^c(R^w)<F^c(R)$, then by the single-crossing property we have $F^{c}(R) P^{w+l} (t^*,q^*)$. That is why the dots in the picture are above $IC(R^{w+l}, (t^*,q^*))$.

Define $F^w$ such that
$F^{w}(R)=F^c(R^w)$ if $R\in [R^w,R^{w+l}]$, $F^{w}(R)=F^c(R^{*})$ if $R\in ]R^{w+l},R^{'}]$, $F^{w}=F^{c}$ otherwise. Note that $F^{w}$ is finite in $[\underline{R},R']$. We further note that compared to 
$F^{c}$, in $[R^w,R^{w+l}]$ the expected revenue in $F^{w}$ is smaller and in $]R^{w+l},R']$ it is higher. 
The decrease in revenue need not be compensated by the increase in revenue in $F^{w}$. 
The expected revenue in $[R^{w},R^{w+l}]$ in $F^{c}$ is $t^{w}\mu([R^{w},R^{w+1}])+\sum_{w+i=w+1}^{w+l}t^{w+i}\mu([R^{w+i},R^{w+i+1}])$.
In $F^{w}$ the corresponding revenue is $t^{w}\mu([R^{w},R^{w+l}])$. 
Thus $\sum_{w+i=w+1}^{w+l}t^{w+i}\mu([R^{w+i},R^{w+i+1}])$ is the expected revenue that may not be compensated by the increase in the expected revenue in $F^{w}$. We note by the single-crossing property
$F^w$ is strategy-proof and individually rational. 
We have $R^{n}\uparrow R'$, and therefore $[R^{n},R']\downarrow \{R'\}$. Since $\mu$ is continuous, $\mu([R^{n},R'])\downarrow \mu(\{R'\})=0$.

Now $\sum_{w+i=w+1}^{w+l}t^{w+i}\mu([R^{w+i},R^{w+i+1}])\leq \sum_{w+i=w+1}^{w+l}t(R')\mu([R^{w+i},R^{w+i+1}])\leq t(R')\mu([R^{w},R'])$, note $t(R')=t^*$. We can choose $w$ large enough
such that $t^{c}(R')\mu([R^{w},R'])<\frac{\epsilon}{2}$. For this large $w$ consider  $F^{w}$, which is finite in $[\underline{R},R']$. Note that $\{F^c(R)\mid F^c(\underline{R})\leq F^c(R)\leq F^c(R^{w})\}$ is finite.
If it were not finite, then being bounded $\{F^c(R)\mid F^c(\underline{R})\leq F^c(R)\leq F^c(R^{w})\}$ will have a limit point and since $Rn(F^{c})$ is closed the limit point will belong to $Rn(F^{c})$ thus adding a limit point to $Rn(F^{c})$.  

Now consider $R^{k}\downarrow R''$. Let $R^{q}$ be such that $\overline{T}\mu([R'',R^{q}])<\frac{\epsilon}{2}$. Consider $R^{q-l}$ such that $F(R'')I^{q-l}F(R^{q})$. Define $F^{q}$ by setting $F^{q}(R)=F(R^{*})$ for all $R\in [R'',R^{q-l}[$ and $F^{q}(R)=F(R^{q})$ for $R\in [R^{q-l},R^{q}]$, ans let $F^{q}=F^{c}$ otherwise. 
Finally define the $F^{\epsilon}$ as follows: 

$$F^{\epsilon}(R) = \begin{cases}
	F^{w}(R), & \text{if $R\in [\underline{R},R' [$;}\\
	F^{c}(R^{*}) & \text{if $R\in [R',R'']$}\\
	F^{q}(R) & \text{if $R\in ]R'',\overline{R}]$}\\ 
\end{cases}$$

\noindent By construction $F^{\epsilon}$ is strategy-proof and individually rational, and $Rn(F^{\epsilon})$ is finite. Now let $F^{*}$ exist. Consider $\epsilon=\frac{1}{n}$. 
We have shown that there exists $F^{\frac{1}{n}}$ and $Rn(F^{\frac{1}{n}})$ such that $0\leq E(F^{c})-E(F^{\frac{1}{n}})\leq \frac{1}{n}$. 
Thus $E(F^{c})\leq E(F^{\frac{1}{n}})+\frac{1}{n}\leq E(F^{*})+\frac{1}{n}$. By letting $n\rightarrow \infty$ we obtain $E(F^{c})\leq E(F^{*})$. This completes the proof of Proposition \ref{prop:opt_count}.

\medskip

\noindent{\bf Proof of Proposition \ref{prop:rest_sp}.} 	We prove the proposition in two separate Lemmas.  

\begin{lemma}\rm Let $F:[\underline{R},\overline{R}]\rightarrow \mathbb{Z}$ be a mechanism and $Rn(F)$ be finite, $[\underline{R},\overline{R}]\subseteq \mathcal{R}^{rrc}$. Also assume that if $F(R)I(0,0)$, then $F(R)=(0,0)$. If $F$ is strategy-proof, then  $F$ is monotone and $V^{F}$ is continuous. 
	\label{lemma:rest_sp_implies}	  
\end{lemma}

\begin{proof} Let $R^{*}=\inf\{R\mid q(R)>0\}$. We argue that if $R\prec R^*$, then $F(R)=(0,0)$.
	Suppose instead that $F(R)P(0,0)$, then by $0$-equivalence $q(R)>0$. Then $R^*$ is not the infimum. Thus $F(R)I(0,0)$ but then by assumption 
	$F(R)=(0,0)$.

	Let $F$ be strategy-proof.  
	Now we consider two cases: $(i)$ $q(R^{*})>0$, $(ii)$ $q(R^*)=0$. 
	Let $q(R^{*})>0$. Then $F(R^*)P^*(0,0)$ because if $F(R^*)I^*(0,0)$, then $F(R^{*})=(0,0)$ by assumption. Thus $t(R^*)<t_{R^{*}}$. But then for $R\prec R^{*}$ with 
	$t(R^{*})<t_R<t_{R^{*}}$ we have $(t(R^{*}),q(R^{*})) PF(R)=(0,0)$. This contradicts strategy-proofness. Thus case $(i)$ holds only if $R^{*}=\underline{R}$. If $\underline{R}\prec R^{*}$, then $q(R^{*})=0$. Then by $0$-equivalence $(0,0)I^{*}(t(R^{*}),0)$. Thus $F(R^{*})=(0,0)$.    
	
	Now if $R^{*}\prec R$, then $q(R)>0$. If $q(R)=0$, then $F(R)I(0,0)$.      
	Then $F(R')I'(0,0)$ for all $R'\prec R$.
	Suppose not, i.e., let $F(R')P'(0,0)$ for some $R'\prec R$. Then $F(R')=(t',q')$ and $t'<t_{R'}$ and $0<q'$. Since $R'\prec R$, $t'<t_{R}$. Thus $F(R')PF(R)$, this contradicts strategy-proofness.
	That is if  $q(R)=0$, then $F(R')=(0,0)$ for $R'\precsim R$. Thus $R^{*}\neq \inf \{R\mid q(R)>0\}$, a contradiction. Also the argument pertaining to strategy-proofness shows that $q(R)>0$ and $F(R)P(0,0)$ if $R^*\prec R$.
	
	Thus for any $R',R''$ such that $R^*\prec R'\prec R''$, $t(R')<t_{R'}<t_{R''}$, $q(R')>0$. Then $R''$ cuts $R'$ at $(t(R'),q(R'))$ from above. By strategy-proofness $(t(R'),q(R'))\leq (t(R''),q(R''))$. That is for allocations that are better than $(0,0)$ the single-crossing property holds and $F$ is monotone for such bundles. 
	
	Now let $Rn(F)=\{(q^0,0),(t^i,q^i)| i=1,\ldots l-1, 0<t^i<t^{i+1}, 0<q^i<q^{i+1},0\leq q^0\}$. If $q^0>0$, then we know $F(R^{*})=F(\underline{R})=(q^0,0)$. If $q^0=0$, then $t^{1}=t_{R^{*}}$. If $t^1>t_{R^{*}}$, then consider $R$ such that $R^*\prec R$ with $t_{R^{*}}<t_{R}<t^1$. But then $F(R)$ is not defined, a contradiction. If $t^1<t_{R^*}$, since $q^1>0$, then $(t^1,q^1)P^{*}(0,0)=F(R^*)$ contradicting strategy-proofness. If $R^{*}=\underline{R}$ and $q(\underline{R})>0$ then $(0,0)$ is not in the range. Also $F(\underline{R})\underline{P}(0,0)$. Thus $F$ is monotone in this case as well.
	
	Thus now the continuity of $V^{F}$ is straightforward to establish. Let $\underline{R}\prec R^{*}$ so that $q(R^{*})=0$. We know $t^1=t_{R^*}$. Consider a sequence $\{R^{n}\}_{n=1}^{\infty}$ such that 
	$R^{n}\downarrow R^{*}$. Since $t_{R^n}\downarrow t_{R^{*}}$. Then after some $n$, $t_{R^{*}}<t_{R^n}<(t^2,q^2)$. Thus after some $n$, $F(R^n)=(t_{R^{*}},q^1)$. Since $(0,0)=F(R^{*})I^{*}(t_{R^{*}},q^1)$ continuity of $V^F$ holds at $R^*$. Proofs for the other cases are similar. Let $(t^1,q^1)P^1(t^2,q^2)$ and $(t^2,q^2)P^2(t^1,q^1)$. Then by richness there is 
	$R$ such that $R^1\prec R\prec R^2$ and $(t^1,q^1)I(t^2,q^2)$ and $t^2<t_{R}$ since $R\in \mathcal{R}^{rrc}$. With this observation, arguments are analogous to the one for $R^*$.    
	
	\noindent This competes the proof of Lemma \ref{lemma:rest_sp_implies}.   
	
\end{proof}

\begin{remark}\rm Consider an extension of restricted classical preferences $R$ to $R^a$, where the latter is defined on $\mathbb{Z}$. Let $R^a=R$ in $[0,t_R]\times [0,1]$ and for all $q',q''\in [0,1]$, $(0,0)P^a(t',q') P^a (t'',q'')$ 
	where $t_{R}<t'<t''$, and $(t,q')I^{a}(t,q'')$ 
	where $t_R<t$;  
	and let the extended preference be continuous. This extension can be done. 
	For $[0,t_R]\times [0,1]$ represent $R$ by the function $f_{t_R}:[0,t_R]\times [0,1]\rightarrow \Re$
	by $f_{R}(t,q)=r$, where $(0,r)I_{R}(t,q)$. Then extend $f_{R}$, call the extension  $f_{R^a}:\mathbb{Z}\rightarrow \Re$ defined as $f_{R^a}(t,q)=f_{R}(t,q)$ if $(t,q)\in [0,t_R]\times [0,1]$ and otherwise $f_{R^a}(t,q)=-(t-t_R)$. Then $f_{R^a}$ is continuous and $(t',q')R^a (t'',q'')$ if and only if $f_{R^a}(t',q')\geq f_{R^a}(t'',q'')$. 	Call the extended domain to be $\mathcal{R}^{rrca}$. Further $R^{'a}\prec R^{''a}$ if and only if $t_{R'}<t_{R''}$ and $f_{R^{'a}}(t,q)\leq f_{R^{''a}}(t,q) $ . 
	Definition \ref{defn:single_crosiing_restricted} holds for $\mathcal{R}^{rrca}$. 
	Let $F:\mathcal{R}^{rrca}\rightarrow \mathbb{Z}$
	be individually rational. Also let if  $F(R^a)I^{a}(0,0)$ then $F(R^a)=(0,0)$. Then the geometry of strategy-proof individually rational mechanism in $\mathcal{R}^{rrca} $ are same as the mechanism in Lemma \ref{lemma:rest_sp_implies}.   
	This is because in the proof the lemma wherever we have used ``$F(R)=(t,q)$ not defined''  $(0,0) P^a (t,q)$ and thus individual rationality does not permit $F(R^a)=(t,q)$. 
	A lemma analogous to Lemma \ref{lemma:preference_conv} holds for $[\underline{R}^a,\overline{R}^{a}]$, if we use the representation $f_{R^a}$. Thus a result analogous to Lemma \ref{lemma:soln_exists} holds for $[\underline{R}^{a},\overline{R}^a]$. The optimal mechanism for $[\underline{R}^{a},\overline{R}^a]$ is the optimal mechanism for $[\underline{R},\overline{R}]$ also.          
	\label{remark:pref_extn_1}
\end{remark}

\begin{lemma}\rm Let $F:[\underline{R},\overline{R}]\rightarrow \mathbb{Z}$ be a mechanism and $Rn(F)$ be finite, $[\underline{R},\overline{R}]\subseteq \mathcal{R}^{rrc}$. Also assume that if $F(R)I(0,0)$, then $F(R)=(0,0)$. If $F$ is monotone and $V^{F}$ is continuous, then $F$ is strategy-proof.   
	\label{lemma:rest_implies_sp}
	
\end{lemma}

\begin{proof} Let $R^{*}=\inf\{R\mid q(R)>0\}$. We argue that if $R\prec R^*$, then $F(R)=(0,0)$.
	Suppose instead that $F(R)P(0,0)$, then by $0$-equivalence $q(R)>0$. Then $R^*$ is not the infimum. Thus $F(R)I(0,0)$ but then by assumption 
	$F(R)=(0,0)$. 
	
	Let $F$ be monotone and $V^{F}$ be continuous. Without loss of generality assume that $\underline{R}\prec R^{*}$. 
	By monotonicity of $F$, $(0,0)<(t(R),q(R))$ for all $R$ with $R^*\prec R$. By finiteness of $F$
	let $R_{n}(F)=\{(0,0), (t^i,q^i)\mid 0<t^i<t^{i+1}, 0<q^{i}<q^{i+1}, i=1,\ldots,l-1\}$.
	We argue that $q(R^{*})=0$. Suppose $q(R^{*})>0$, then $F(R^{*})P^{*}(0,0)$. This is because
	if $F(R^*)I^*(0,0)$, then by assumption $F(R^{*})=(0,0)$. 
	We know that 
	$F(R)=(0,0)$ for all $R\prec R^{*}$. Take a sequence $R^{n}\uparrow R^{*}$. But then $(0,0)=\lim_{n\rightarrow \infty}F(R^{{n}})I^*F(R^{*})$ does not hold, this contradicts continuity of $V^{F}$. Thus $q(R^{*})=0$, and thus $F(R^{*})=(0,0)$.    
	
	Then $t^{1}=t_{R^{*}}$. Since $R^{*}=\inf\{R\mid q(R)>0\}$, $q(R)>0$ if $R^*\prec R$. 
	If indeed $q(R)=0$, then $F(R)=(0,0)$ and by monotonicity of $F$, $F(R)=(0,0)$ for all $R'\precsim R$. Thus $R^*$  is not the infimum. 
	If $t^1>t_{R^{*}}$, then consider $R$ such that $R^*\prec R$ with $t_{R^{*}}<t_{R}<t^1$. But then $F(R)$ is not defined, a contradiction. 	
	If $t^1<t_{R^*}$, then given that  $q^1>0$, $(t^{1},q^1)P^{*}F(R^{*})=(0,0)$. Consider a sequence 
	$R^{n}\downarrow R^{*}$ where $F(R^n)=(t^1,q^1)$. 
	Then $(t^1,q^1)=\lim_{n\rightarrow \infty}F(R^n)P^{*}F(R^{*})$. Thus continuity of $V^{F}$ is violated.

	If $R^{*}=\underline{R}$ and $q(\underline{R})>0$, then $(0,0)$ is not in the range. 
	Thus $F$ is strategy-proof for $R\in [\underline{R}, R^*]$. 
	Consider $(t^1,q^1), (t^2,q^2)$. Let $R^{'}$ be such that, $F(R')=(t^1,q^1)$ and $t_{R'}=t^2$. To see that $R'$ is well defined note that $(0,0)<(t^2,q^2)$, and there is  $R''$ such that $F(R'')=(t^2,q^2)$ with $t^2<t_{R''}$.
	Thus, since  there is $R$ such that  $R^{*}\prec R$
	and $F(R)=(t^1,q^1)$,  $R'$ is well defined. 
	Now $F(R'')=(t^2,q^2)$. Then by monotonicity of $F$, $F(R)\in \{(t^1,q^1),(t^2,q^2)\}$ for all 
	$R\in [R',R'']$. Arguments similar to the ones applied in the proof of Claim \ref{claim:ab} shows that 
	$F$ restricted to $[R',R'']$ is local strategy-proof. We can apply the arguments in Claim \ref{claim:ab} because the richness condition that we have used there is that between any two preferences there is another one, which holds for $\mathcal{R}^{rrc}$ also. Strategy-proofness for $R\prec R^{*}$ holds trivially. For preference $R^{*}$, $(0,0)I^{*}(t^1,q^1)$ since $t^1=t_{R^*}$. For 
	preference $R\in ]R^*,R']$, $F(R)=(t^1,q^1)P (0,0)$ and $F(R')P'(t^2,q^2)$ since $t^2=t_{R'}$. Also for $R$ with $R'\prec R\precsim R''$, $F(R)=(t^2,q^2)P(0,0)$. 
	Since $F(R'')P''(0,0)$ for any $R$ with $F(R)=(t^2,q^2)$ and $R''\prec R$, $F(R)P(t^1,q^1)P(0,0)$. 
	Thus local strategy-proofness holds for $\{R\mid F(R)\in \{(t^1,q^1), (t^2,q^2)\}\}$. 
	Now we show that there is $R^1$ such that $(t^1,q^1)I^1 (t^2,q^2)$. Consider  $[R',R'']$ where $F(R')=(t^1,q^1), t_{R'}=t^2$, $F(R'')=(t^2,q^2)$. 
	Then, $F(R')P'(t^2,q^2)I'(0,0)$. By monotonicity of $F$ consider two cases:
	
	\noindent {\bf Case $(i)$:} Let  $[R',R^0]\cup ]R^0,R'']=[R',R'']$ where $F(R)=(t^1,q^1)$ if $R'\precsim R\precsim R^0$ and $F(R)=(t^2,q^2)$ if $R^0\prec R\precsim R''$. Also let by way of contradiction if $R\in [R',R^0]$ then $F(R)PF(R'')$ and if 
	$R\in ]R^0,R'']$, then $F(R)PF(R')$. Consider a decreasing sequence $\{R^n\}_{n=1}^{\infty}$ to $R^0$. We have $F(R^n)=(t^2,q^2)$, and thus $\lim_{n\rightarrow \infty}F(R^n)=(t^2,q^2)$ 
	Since $F(R^0)=(t^1,q^1)$ and $F(R^0)P(t^2,q^2)$ we have a contradiction to continuity of $V^{F}$.

	\noindent {\bf Case $(ii)$:} $[R',R^0[\cup [R^0,R'']=[R',R'']$ where $F(R)=(t^1,q^1)$ if $R'\precsim R\prec R^0$ and $F(R)=(t^2,q^2)$ if $R^0\precsim R\precsim R''$. Also let by way of contradiction if $R\in [R',R^0[$, then $F(R)PF(R'')$ and if 
	$R\in [R^0,R'']$, then $F(R)PF(R')$. In this case consider  an increasing sequence to $R^{0}$. 
	Then we have $F(R^0)P^0(t^1,q^1)=\lim_{n\rightarrow \infty}F(R^{n})$. This contradicts continuity of $V^{F}$.   
	
	Thus arguing in this manner we obtain $R^{1},\ldots, R^{l-1}$ special preferences. 
	For example when finding $R^{2}$, we first find $R'''$ such that $F(R''')=(t^2,q^2)$ and $t_{R'''}=t^3$.
	After we find the special preferences the argument is same as for classical preferences since the single-crossing property holds at bundles that are strictly preferred to $(0,0)$. 
	The extra observation that we need is that preferences $R$ for which $t^{i}<t_{R}\leq t^{i+1}$, 
	$F(R)=(t^i,q^i)$ and $(t^i,q^i)P(t^{j},q^j), j<i$. This completes the proof of Lemma \ref{lemma:rest_implies_sp}. 
\end{proof}

\begin{remark}\rm Lemma \ref{lemma:rest_implies_sp} also goes through for $\mathcal{R}^{rrca}$ if we assume $F:\mathcal{R}^{rrca}\rightarrow \mathbb{Z}$ to be $F$ is individually rational, and  $(0,0)I^{a}F(R^a)\implies F(R^a)=(0,0)$.   
\end{remark}

\noindent This completes the proof of Proposition \ref{prop:rest_sp}.

\medskip

\noindent{\bf Proof of Lemma \ref{lemma:order_payment_bound}:} Let $R$ be a restricted classical preference. We note that if $(t,q)$ is such that $t<t_{R}$, $0<q$, then $(t,q)P(0,0)$. 
This holds because by $0$-equivalence $(0,0)I(t_R,q)$. By money-monotonicity $(t,q)P(t_R, q)$. 
By transitivity $(t,q)P(0,0)$.  
Now we argue that for any $(t,q)$ such that $t<t_R$ and $0<q<1$ there is $(t',1)$ where $(t,q)I(t',1)$ and $t'<t_R$. By $q-$monotonicity $(t,1)P(t,q)$. By $0$-equivalence $(t_R,1)I(0,0)$. Since $(t,q)P(0,0)$, by transitivity $(t,q)P(t_R,1)$. Thus we have $(t,1)P(t,q)P(t_R,1)$. Continuity of $R$ entails
$(t,q)I(t',1)$. 

Since $R',R''$ are restricted classical $t_{R'}\neq t_{R''}$. Let by way of contradiction $t_{R''}<t_{R'}$. Then $(t_{R''},1)P'(0,0)$. Thus by continuity of $R'$  consider $(t,q)I'(t_{R''},1)$ where $t<t_{R''}$ and $q<1$. Since $R''$ cuts $R'$ from above at $(t,q)$, there is no $(t',1)$ such that  $(t,q)I'' (t',1)$. This is a contradiction. Thus we have  $t_{R'}<t_{R''}$. The proof of the lemma follows.

\medskip

\noindent{\bf Proof of Proposition \ref{prop:myer_opt}:} 	
We need an intermediary result.	

\begin{lemma}\rm Let $F\in \mathcal{F}$ be as defined in Proposition \ref{prop:myer_opt}. 
	Let $\theta^{*}q(\theta^{*})-q(\theta^{*})t(\theta^{*})>0$ and $\theta^{*}<\overline{\theta}$
	Then there is $G\in \mathcal{F}$ such that $E[G]\geq E[F]$ and $\theta^{*}q_G(\theta^{*})-q_G(\theta^{*})t_G(\theta^{*})=0$, where $G(\theta)=(t_G(\theta), q_G(\theta))$, and $G(\theta)=F(\theta)$ otherwise.    
	\label{lemma:myer}
\end{lemma}

\begin{proof}
	We show that if $\theta^{*}q(\theta^{*})-q(\theta^{*})t(\theta^{*})>0$, then $\theta=\underline{\theta}$. 
	Let by the way of contradiction $\theta^{*}q(\theta^{*})-q(\theta^{*})t(\theta^{*})>0$ and $\underline{\theta}<\theta^{*}$. Then consider $\theta^{*}-\epsilon$ such that 
	$[\theta^{*}-\epsilon]q(\theta^{*})-q(\theta^{*})t(\theta^{*})>0$. By strategy-proofness 
	$q(\theta-\epsilon)>0$. This contradicts $\theta^{*}=\inf\{\theta\mid q(\theta)>0\}$.

	Thus now we have $\theta^*=\underline{\theta}$. Since $\underline{\theta}q(\underline{\theta})-q(\underline{\theta})t(\underline{\theta})>0$, we have
	$q(\underline{\theta})>0$. Further, $t(\underline{\theta})<\underline{\theta}$. 
	By finiteness of $Rn(F)$ consider the situation where $(t',q')\in Rn(F)$ is either in $IC(\underline{\theta},(\underline{\theta},1))
	=\{(t,q)\mid \underline{\theta}q-qt=0\}$ or the smallest bundle in $Rn(F)$ such that $(t',q')$ with $\underline{\theta}<t'$. Let $(t'',q'')$ be the largest bundle such that $t''<\underline{\theta}$. Let $\theta^{**}q'-q't'=\theta^{**}q''-q''t''=\theta^{**}q'''-q'''t'''$, where $(t''',q''')\in IC(\underline{\theta},(\underline{\theta},1))$. 
	To see that such $\theta^{**}$ exists, let $F(\theta')=(t',q'),F(\theta'')=(t'',q'')$. By monotonicity 
	of $F$, $\theta''<\theta'$. By strategy-proofness $\theta'q'-q't'\geq \theta'q''-q''t''$, and
	$\theta''q''-q''t''\geq \theta''q'-q't'$. Now let $\theta^{**}q'-q't'=\theta^{**}q''-q''t''$. 
	Then $\theta''\leq \theta^{**}\leq \theta'$. Thus $F(\theta^{**})\in \{(t',q'),(t'',q'')\}$. 
	The by continuity of $\theta^{**}q-qt$, there is $(t''',q''')$.    
	Let $G(\theta)=(t''',q''')$ for all $\theta\in]\underline{\theta},\theta^{**}]$, $G(\underline{\theta})=(0,0)$ and  $G=F$ otherwise. If $F$ is such that for all $\theta$, $t(\theta)<\underline{\theta}$, then define 
	$G(\theta)=(\underline{\theta},1)$ if $\theta\in ]\underline{\theta}, \overline{\theta}]$
	and $G(\underline{\theta})=(0,0)$. Since $\Gamma$ is strictly increasing the result follows. Thus the proof of the lemma follows.  
	
\end{proof}

By Lemma \ref{lemma:myer} without loss of generality let by way of contradiction  $Rn(F)=\{(t^0,q^0)=(0,0), (\theta^{*},q^{*})=(t^1,q^{1}),(t^{2},q^{2}),(t^{3},q^{3}),(t^{4},q^{4})\}$.
That is $\theta^{0}=\underline{\theta}\leq \theta^*=\theta^1\leq \theta^2\leq\theta^3\leq\theta^4\leq\theta^5=\overline{\theta}$. 
The bundles are $(0,0)\leq (t^1,q^1)\leq (t^2,q^2)\leq (t^3,q^3)\leq(t^4,q^4)$.  
Also note that $F(\theta^{*})=(0,0)$. Further set $q^{4}=1$. We fix the $q$s to entail a contradiction 
from the first order conditions that involve $\theta$s and $t$s.      
The relevant Lagrange is

$L=\theta^{*}q^{*}[\Gamma(\theta^{2})-\Gamma(\theta^{*})]+t^{2}q^{2}[\Gamma(\theta^{3})-\Gamma(\theta^{2})]+t^{3}q^{3}[\Gamma(\theta^{4})-\Gamma(\theta^{3})]+t^{4}q^{4}[1-\Gamma(\theta^{4})]$ 
\[-\lambda_{1}[\theta^{2}q^{2}-t^{2}q^{2}-\theta^{2}q^{*}+q^{*}\theta^{*}]-\lambda_{2}[\theta^{3}q^{3}-q^{3}t^{3}-\theta^{3}q^{2}+q^{2}t^{2}]-\lambda_{3}[\theta^{4}q^{4}-q^{4}t^{4}-\theta^{4}q^{3}+q^{3}t^{3}]\]

\noindent A solution to this optimization problem exists. The constraint set is compact and the objective function is continuous.          
Now we consider the first order conditions with respect to $\theta$s. 
\begin{equation}
	-q^{*}[\theta^{*}\gamma(\theta^{*})+\Gamma(\theta^{*})]-\lambda_{1}q^{*}=0\\
	\implies \lambda_{1}=-[\theta^{*}\gamma(\theta^{*})+\Gamma(\theta^{*})]
	\label{eqn:theta_star}
\end{equation} 
\begin{equation}
	\theta^{*}q^{*}\gamma(\theta^{2})-t^{2}q^{2}\gamma(\theta^{2})-\lambda_{1}[q^{2}-q^{*}]=0\label{eqn:theta-2}
\end{equation} 
\begin{equation}
	t^{2}q^{2}\gamma(\theta^{3})-t^{3}q^{3}\gamma(\theta^{3})-\lambda_{2}[q^{3}-q^{2}]=0
	\label{eqn:theta-3}
\end{equation} 
\begin{equation}
	t^{3}q^{3}\gamma(\theta^{4})-t^{4}q^{4}\gamma(\theta^{4})-\lambda_{3}[q^{4}-q^{3}]=0
	\label{eqn:theta-4}
\end{equation}
\noindent Now we consider the first order conditions with respect to $t$s.
\begin{equation}
	q^{2}[\Gamma (\theta^{3})-\Gamma (\theta^{2})]+\lambda_{1}q^{2}-\lambda_{2}q^{2}=0\\
	\implies [\Gamma (\theta^{2})-\Gamma (\theta^{2})]+\lambda_{1}-\lambda_{2}=0
	\label{eqn:t_2}
\end{equation} 
\begin{equation}
	q^{3}[\Gamma(\theta^{4})-\Gamma (\theta^{3})]+\lambda_{2}q^{3}-\lambda_{3}q^{3}=0\\
	\implies [\Gamma (\theta^{4})-\Gamma (\theta^{3})]+\lambda_{2}-\lambda_{3}=0
	\label{eqn:t_3}
\end{equation} 
\begin{equation}
	q^{4}[1-\Gamma (\theta^{4})]+\lambda_{3}q^{4}=0\\
	\implies [1-\Gamma (\theta^{4})]+\lambda_{3}=0,
	\label{eqn:t_4}
\end{equation} 

\medskip

\noindent Now we have the following:
$\lambda_{3}=-[1-\Gamma(\theta^{4})]$. Then, $\lambda_{2}=-[1-\Gamma(\theta^{3})]$ and 
$\lambda_{1}=-[1-\Gamma(\theta^{2})]$. Now the equations (\ref{eqn:theta-2}), (\ref{eqn:theta-3}) (\ref{eqn:theta-4})
can be written as 
\[\frac{\gamma(\theta^{2})}{1-\Gamma(\theta^{2})}=\frac{q^{2}-q^{*}}{t^{2}q^{2}-\theta^{*}q^{*}}\]
\[\frac{\gamma(\theta^{3})}{1-\Gamma(\theta^{3})}=\frac{q^{3}-q^{2}}{t^{3}q^{3}-t^{2}q^{2}}\]
\[\frac{\gamma(\theta^{4})}{1-\Gamma(\theta^{4})}=\frac{q^{4}-q^{3}}{t^{4}q^{4}-t^{3}q^{3}}=\frac{1-q^{3}}{t^{4}-t^{3}q^{3}},~\text{since}~q^{4}=1\]

\bigskip

\noindent Further, we have $\theta^{4}-t^{4}=\theta^{4}q^{3}-q^{3}t^{3}$. Therefore, $t^{4}-q^{3}t^{3}=\theta^{4}-\theta^{4}q^{3}=\theta^{4}[1-q^{3}]$. Similarly, 
$t^{3}q^{3}-t^{2}q^{2}=\theta^{3}[q^{3}-q^{2}]$. 
Hence, $\frac{\gamma(\theta^{3})}{1-\Gamma(\theta^{3})}=\frac{1}{\theta^{3}}$,  $\frac{\gamma(\theta^{4})}{1-\Gamma(\theta^{4})}=\frac{1}{\theta^{4}}$ and 
$\frac{\gamma(\theta^{2})}{1-\Gamma(\theta^{2})}=\frac{1}{\theta^{2}}$.  If $\theta^{4}>\theta^{3}>\theta^{2}$ we have a contradiction. 
Thus we must have $\theta^2=\theta^3=\theta^4$.  

\medskip

\noindent  Therefore  we assume that the mechanism has three bundles in its range:
$\{(0,0), (\theta^{*},q^{*}), (t,1)\}$ and $\theta - t =\theta q^*-q^*\theta^*$
Assume by the way of contradiction that $\theta^{*}<\theta$, $F(\theta')=(0,0)$ if $\theta'\in [\underline{\theta},\theta^{*}]$, $F(\theta')=(\theta^*,q^*)$ if $\theta'\in ]\theta^*,\theta]$, $F(\theta')=(t,1)$ if $\theta'\in ]\theta,\overline{\theta}]$. 	
The Lagrange for $3$ bundle range is:
$L=\theta^*q^* [\Gamma(\theta)-\Gamma(\theta^{*})]+t[1-\Gamma(\theta)]-\lambda_{1}[\theta-t-\theta q^{*}+q^{*}\theta^{*}]$. Let FOC denote first order condition.  

\noindent FOC for $\theta^*$ is $q^{*}[\Gamma(\theta)-\Gamma(\theta^{*})]-\theta^{*}q^{*}\gamma(\theta^{*})-\lambda_{1}q^{*}=0$

\medskip

\noindent FOC for $t$ is $[1-\Gamma(\theta)]+\lambda_1=0$

\medskip

\noindent FOC for $\theta$ is $\theta^{*}q^{*}\gamma(\theta)-t\gamma(\theta)-\lambda_1[1-q^{*}]=0$  

\medskip

\noindent From FOC for $t$, $\lambda_1=-[1-\Gamma(\theta)]$. 
In FOC for $\theta^{*}$, if $q^*>0$, then we have after replacing for $\lambda_1$, 
$$[\Gamma(\theta)-\Gamma(\theta^{*})]-\theta^{*}\gamma(\theta^{*})+1-\Gamma(\theta)=0.$$
Thus, $$\theta^{*}=\frac{1-\Gamma(\theta^{*})}{\gamma(\theta^{*})}.$$
or
$$\frac{1}{\theta^{*}}=\frac{ \gamma(\theta^{*})          )}{1-\Gamma(\theta^{*})}.$$  

\noindent Now replacing for $\lambda_{1}$ in FOC for $\theta$ we obtain 

$$  \theta^{*}q^{*}\gamma(\theta)-t\gamma(\theta)+[1-\Gamma(\theta)][1-q^{*}]=0.$$
Thus 
$$ \gamma(\theta)[\theta^*q^*-t]=-[1-\Gamma(\theta)][1-q^{*}].$$
After multiplying by $-1$ on both sides we obtain 

$$ \gamma(\theta)[t-\theta^*q^*]=[1-\Gamma(\theta)][1-q^{*}].$$  

Thus, $$\frac{\gamma(\theta)}{1-\Gamma(\theta)}=\frac{1-q^{*}}{t-\theta^{*}q^{*}}.$$

\noindent If $\theta^{*}<\theta$ and the hazard rate is increasing, then we obtain
$\frac{1}{\theta^{*}}<\frac{1-q^{*}}{t-\theta^{*}q^{*}}$. 
Thus, 
$$t-\theta^*q^*<\theta^{*}-\theta^{*}q^{*}.$$ 

\noindent This implies $t<\theta^{*}$. This is a contradiction to monotonicity of $F$. This completes the proof of Proposition \ref{prop:myer_opt}.

\end{document}